\PassOptionsToPackage{dvipsnames}{xcolor}
\documentclass[11pt,letterpaper]{article}

\usepackage[utf8]{inputenc}
\usepackage{tabularx}
\usepackage[table]{xcolor}

\usepackage{wrapfig}

\usepackage[margin=1in]{geometry}
\usepackage{setspace}

\usepackage{tikz}
\usetikzlibrary{arrows}

\usepackage{mathtools}

\usepackage{subfigure}
\usepackage{float}
\usepackage[justification=centering]{caption}
\usepackage{graphicx}
\usepackage{amsmath,amsthm,amssymb,amsfonts}
\usepackage{algorithm}
\usepackage{listings}
\usepackage[colorlinks,linkcolor=blue]{hyperref}
\usepackage{url}

\usepackage{colortbl}
\usepackage{authblk}

\newtheorem{theorem}{Theorem}
\newtheorem{lemma}{Lemma}
\newtheorem{definition}{Definition}

\newtheorem{problem}{Problem}

\newcommand{\col}[2]{\cellcolor[HTML]{#1}#2}
\newcommand{\pc}[0]{\col{b7e1cd}{}}
\newcommand{\npc}[0]{\col{f4c7c3}{}}
\newcommand{\qc}[0]{\col{fce8b2}{?}}

\renewcommand{\leq}{\leqslant} 
 
\renewcommand{\geq}{\geqslant}

\title{On the Complexity of the Virtual Network Embedding \\ in Specific Tree Topologies}

\author[1]{Sergey Pankratov}
\affil[1]{ITMO University, Russia}

\author[2]{Vitaly Aksenov}
\affil[2]{City, University of London, United Kingdom}

\author[3]{Stefan Schmid}
\affil[3]{TU Berlin, Germany}

\date{}

\begin{document}

\singlespacing

\maketitle

\begin{abstract}
Virtual networks are an innovative abstraction that extends cloud computing concepts to the network:
by supporting bandwidth reservations between compute nodes (e.g., virtual machines), virtual networks can provide a predictable performance to distributed and communication-intensive cloud applications.
However, in order to make the most efficient use of the shared resources, the Virtual Network Embedding (VNE) problem has to be solved:
a virtual network should be mapped onto the given physical network so that resource reservations are minimized.
The problem has been studied intensively already and is known to be NP-hard in general.
In this paper, we revisit this problem and consider it on specific topologies, as they often arise in practice. 

To be more precise, we study the weighted version of the VNE problem: we consider a virtual weighted network of a specific topology which we want to embed onto a weighted network with capacities and specific topology.
As for topologies, we consider most fundamental and commonly used ones: line, star, $2$-tiered star, oversubscribed $2$-tiered star, and tree, in addition to also considering arbitrary topologies.
We show that typically the VNE problem is NP-hard even in more specialized cases, however, sometimes there exists a polynomial algorithm: 
for example, an embedding
of the oversubscribed $2$-tiered star onto the tree is polynomial while an embedding of an arbitrary $2$-tiered star is not.

\end{abstract}

\newpage

\section{Introduction}

In the era of data-intensive applications and AI, the network has become a performance bottleneck.
Accordingly, over the last years, the Cloud computing paradigm is extended to include the network interconnecting distributed applications:
by providing bandwidth guarantees between compute nodes (e.g., processes, virtual machines, or containers), the resulting virtual networks promise a predictable application performance.
In order to optimally benefit from resource allocation flexibility and sharing, for example, a virtual network operator needs to efficiently map such a virtual network (VN) on a given physical network (PN).
This combinatorial optimization problem is known as the virtual network embedding (VNE) problem and is NP-hard in the general case.

The main objective of VNE is to map the networks in a way that all resource requirements are satisfied while optimizing the cost.
The typical resource requirements are link bandwidths.
In our work, we consider satisfying the bandwidth requirements and optimizing the cost of VNE on weighted networks both simultaneously and separately.

It is known that VNE of a generic virtual network onto a generic physical network is typically NP-complete~\cite{RostS20_complexity}.
That means, there is no polynomial time algorithm that solves VNE in the most generic case, unless P $=$ NP.
One of the ways to simplify it is to consider topology restrictions on the virtual and physical networks as done in~\cite{RostFS15_stars, Wu20_path, FigielKNR0Z21_Tree}.
In our work, we revisit the commonly-used topologies such as line~\cite{DiazPS02_survey}, star~\cite{RostFS15_stars}, 2-Tiered star or $2$-star~\cite{BallaniCKR11_virtualNet}, and tree~\cite{Leiserson85_fat}.
In particular, we study the important case when these topologies have weights on top of them.

We use the following formulation of the VNE. 
At first, we ask the physical and virtual networks to have the same number of computational nodes and there should be a bijection between them.
As for the resource requirements, the capacity of each link in the physical network should not be exceeded by the bandwidth demand of edges in the virtual network mapped on them.
As for the cost, each edge in the physical network has cost per request, thus, our goal is to optimize the total cost of all requests.
We consider three types of the problem: 1)~minimizing the total cost (weight) with unlimited capacities of the physical network, wVNE; 2)~only capacity requirements without the cost minimization, cVNE; and 3)~both capacity requirements and minimization together, wcVNE or, simply, VNE.
Our goal is to understand which topology requirements lead to polynomial algorithms in these variations while proving NP-completeness for the rest.
Some of our algorithms and proofs can be generalized to other cost functions, and sometimes they can work without the bijection constraint.

There are a several results that are very close to our work and we want to overview them here.
Ballani et al.~\cite{BallaniCKR11_virtualNet} introduced popular virtual network abstractions: a virtual cluster (a uniform star) and an oversubscribed virtual cluster (an oversubscribed $2$-star).
Rost et al.~\cite{RostFS15_stars} presented a polynomial algorithm of an embedding of an uniform star VN.
Our paper continues their research by considering non-uniform networks and $2$-star topologies.
As for linear topologies, Wu et al.~\cite{Wu20_path} showed that cVNE of a uniform line on an arbitrary graph is NP-hard.
In our work, we specify the topologies of PN by tree ones: an arbitrary tree, a star, and a $2$-star.
In addition, wVNE of a graph on a uniform line is equivalent to the \emph{Minimal Linear Arrangement}~\cite{DiazPS02_survey} problem, which is NP-hard for arbitrary graphs~\cite{Garey74_someNP}.
We continue further by considering wVNE on VN with tree topologies, i.e., an arbitrary tree, a star, and a $2$-star.
More discussion on the related work appears closer to the end of the paper, in Section~\ref{sec:related}.

\begin{table}[ht!]
    
    \centering
    \caption{Our wVNE Results \\
    Red color indicates NP-complete problems, green  color is for problems with a polynomial algorithm, while yellow color denoted an open question.
    An arrow in a cell indicates that the complexity of the variant follows from the version pointed onto.
    }\label{table:wVNE:res}
    \centering
    \begin{tabular}{ |>{\raggedright}l|l|l|l|l|l|l| }
    \hline
    VN & Generic PN & Tree PN & Line PN & Uni. PN & Uni. Tree PN & Uni. Line PN \\
    \hline
    Uni. line 
    & \npc Thm.~\ref{th:ULEP_wVNE_NPC}\phantom{0} & \pc $\S$~\ref{sec:line_on_tree_DP} & \pc Trivial
    & \npc Thm.~\ref{th:ULEP_wVNE_NPC} & \pc $\S$~\ref{sec:line_on_tree_DP} & \pc $\leftarrow$ \\
    \hline
    Line 
    & \npc $\rightarrow$ & \npc $\rightarrow$ & \npc Thm.~\ref{th:WLLEP_NPC}  
    & \npc $\rightarrow$ & \npc Thm.~\ref{th:WLUTEP_NPC} & \pc Trivial \\
    \hline
    Uni. star 
    & \pc \cite{RostFS15_stars} & \pc $\leftarrow$ & \pc $\leftarrow$ 
    & \pc $\leftarrow$ & \pc $\leftarrow$ & \pc $\leftarrow$ \\
    \hline
    Star
    & \pc Thm.~\ref{th:w_star} & \pc $\leftarrow$ & \pc $\leftarrow$  
    & \pc Thm.~\ref{th:w_star} & \pc $\leftarrow$ & \pc $\leftarrow$  \\
    \hline
    Overs. $2$-star 
    & \npc $\S$~\ref{th:oversub_2Star_wVNE_NP} & \pc $\S$~\ref{sec:oversub_on_tree} & \pc $\leftarrow$  
    & \npc $\S$~\ref{th:oversub_2Star_wVNE_NP} & \pc $\S$~\ref{sec:oversub_on_tree} & \pc $\leftarrow$ \\
    \hline
    $2$-star 
    & \npc $\rightarrow$ & \npc $\rightarrow$ & \npc Thm.~\ref{th:W2SEP_NPC}  
    & \npc $\rightarrow$ & \npc Thm.~\ref{th:WU_2SEP_NPC} & \qc \\
    \hline
    \end{tabular}
\end{table}

\begin{table}[ht!]
    
    \centering
    \caption{Our VNE and cVNE Results\\
    Red color indicates NP-complete problems, while green color is for problems with a polynomial algorithm.
    An arrow in a cell indicates that the complexity of the variant follows from the version pointed onto.
    }\label{table:cVNE:res}
    \begin{tabular}{ |>{\raggedright}l|l|l|l|l| }
    \hline
    VN & Generic PN & Tree PN & Line PN & Star PN \\
    \hline
    Uni. line 
    & \npc \cite{Wu20_path} & \pc $\S$~\ref{sec:line_on_tree_DP} & \pc Trivial & \pc $\downarrow$ \\
    \hline
    Line 
    & \npc $\rightarrow$ & \npc $\rightarrow$ & \npc Thm.~\ref{th:WLLEP_NPC} & \pc $\S$~\ref{sec:star_PN_Poly} \\
    \hline
    Uni. star 
    & \pc \cite{RostFS15_stars} & \pc $\leftarrow$ & \pc $\leftarrow$ & \pc $\downarrow$ \\
    \hline
    Star
    & \npc $\rightarrow$ & \npc $\rightarrow$ & \npc Thm.~\ref{th:CSTEP_NPC} & \pc $\S$~\ref{sec:star_PN_Poly} \\
    \hline
    Overs. $2$-star
    & \npc Thm.~\ref{th:oversub_2Star_cVNE_NP} & \pc $\S$~\ref{sec:oversub_on_tree} & \pc $\leftarrow$ & \pc $\downarrow$ \\
    \hline
    $2$-star 
    & \npc $\rightarrow$ & \npc $\rightarrow$ & \npc Thm.~\ref{th:W2SEP_NPC} & \pc $\S$~\ref{sec:star_PN_Poly} \\
    \hline
    \end{tabular}

\end{table}

\subsection{Contributions} \label{sec:contribution}

In this work, we conduct a comprehensive study of the complexity of VNE problem under various topological constraints.
We cover the variations that were previously unknown to be NP-complete or polynomial, thus, these findings can be critical for future work on VNE.
The results for the various combinations of VN and PN topologies for wVNE, cVNE, and VNE are shown in Tables~\ref{table:wVNE:res}~and~\ref{table:cVNE:res}.
In the tables, columns correspond to PN topologies and rows correspond to VN topologies.
Namely, we provide the results for the following VN topologies: (uniform) line, (uniform) star, and (oversubscribed) $2$-star; 
and the following PN topologies: (uniform) generic, (uniform) tree, (uniform) line, star.
Here ``generic'' indicates unrestricted topologies.
Uniform means that the weights or costs of edges in a network are equal to one.
A $2$-star is simply a two-level tree, while an oversubscribed $2$-star is a special case of $2$-star that we introduce later.
Red cells indicate NP-complete problems, while green ones are for the problems with the polynomial algorithm.
All cells are provided with a reference to corresponding proofs.
For polynomial VNE problems, we provide polynomial time algorithms to solve them.
An arrow in a cell indicates that the complexity of the variant follows from the version pointed to.
Note that since the results for cVNE and VNE are identical, we group them in Table~\ref{table:cVNE:res}.

Now, we briefly discuss the obtained results.
At first, we consider wVNE problem.
We find that when both virtual and physical networks are non-uniform, wVNE remains NP-complete even for line-on-line embedding (Theorem~\ref{th:WLLEP_NPC}).
Interestingly, wVNE is polynomial for a star (and uniform star) VN (Theorem~\ref{th:w_star}) continuing the work in~\cite{RostFS15_stars}.
We further show that increasing the depth of a star VN from one to two already makes wVNE NP-complete for any considered version (Theorem~\ref{th:W2SEP_NPC}).
However, under the constraint of being oversubscribed (the special version of the $2$-star) wVNE becomes polynomial for Tree PN (and consequently line PN) (Section~\ref{sec:oversub_on_tree}).
The only remaining non-solved case is the complexity of embedding a $2$-star on a uniform line which is marked with a question mark.

As mentioned, the results for VNE and cVNE could be considered together in Table~\ref{table:cVNE:res}.
When PN does not have a restricted topology only uniform star VN has a polynomial algorithm while the variation with unconstrained star VN does not.
However, when PN becomes a star all considered variations of VNE are polynomial as shown in Section~\ref{th:W2SEP_NPC}.
Finally, when PN has line or tree topology we get similar results: for uniform line, uniform star, and oversubscribed $2$-star we provide the polynomial algorithm, while all other versions are NP-complete.

\subsection{Model and Problem}
\label{subsec:model}
Now we introduce all the necessary definitions.
Let $\mathcal{G}$ be a set of physical networks (or, simply, graphs) with a common \emph{topology} that defines the network arrangement: for example, linear topologies, tree topologies, etc.
For a set of $n$ computational nodes $V_G$, a \emph{physical network (PN)} $G \in \mathcal{G}$ is an undirected connected graph $(V_G, E_G)$.
Likewise, $\mathcal{S}$  denotes a set of virtual networks of a given topology.
A \emph{virtual network} (VN) is an undirected connected graph $S = (V_S, E_S) \in \mathcal{S}$, such that $|V_S| = |V_G| = n$.
Vertices in the VN represent some computation agents. 
Let $W : E_S \rightarrow \mathbb{N}$ be the VN demand function.
Each edge $e = (i, j) \in E_S$ is assigned some demand $w_{e} \geq 0$~--- the number of routing requests from agent $i$ to agent $j$.
Let $C: V_G \cup E_G \rightarrow \mathbb{N}$ be the PN capacity function.
$c_e \geq 0$ denotes the capacity of edge $e$, meaning the maximum number of requests from Virtual Network that can go through the edge $e$ is $c_e$.
Finally, let $T: E_G \rightarrow \mathbb{R}^{+} \cup \{0\}$ be the PN cost function.
$t_{e} \geq 0$ represents the cost of traversing a single request through the edge~$e$.

A \emph{node embedding} function $f_V: V_S \rightarrow V_G $ is a function that maps vertices of VN to vertices of PN.
Any node embedding $f_V$ maintains a bijection between $V_S$ and $V_G$.
In other words, each agent of $S$ is assigned to a single node of $G$ and each computational node of $G$ holds exactly one agent from $S$.
Let $Paths_G$ be a set of all unique paths in $G$ and $Paths_G(i, j)$ be a set of all unique paths from $i$ to $j$ in $G$.
An \emph{edge embedding} function $f_E: E_S \rightarrow Paths_G$ is a function that maps edges of VN to paths of PN.
It is required that $f_E((i, j)) \in Paths_G(f_V(i), f_V(j))$.
Please note that we do not allow to map different requests between the same nodes to different paths.
An \emph{embedding} $f: V_S \cup E_S \rightarrow V_G \cup Paths_G$ is a union of node and edge embeddings $f = f_V \cup f_E$.
$\mathcal{D}(S, G)$ is a set of all possible embeddings of VN $S$ on PN $G$.

\begin{definition}\label{def:Cost}
Given $f \in \mathcal{D}(S, G)$~--- an embedding for some virtual network $S$ and some physical network $G$, the embedding cost of $f$ is:

$$ Cost(f) = \sum_{e \in V_S}w_{e} \cdot len(f(e))$$

In the equation above, $w_{e}$ is the demand of $e$, and $len(f(e))$ is the length of the path $f(e)$.
\end{definition}

\emph{Virtual network embedding} (VNE) is a problem of finding an embedding $f \in \mathcal{D}(S, G)$, such that 1)~PN capacities are satisfied and 2)~the embedding cost is minimized.
To minimize the cost we need to find a minimal constant $C$, such that there exists an embedding $f$ with $Cost(f) \leq C$.
The decision version of VNE is described as follows:

\begin{problem}[Virtual Network Embedding (VNE)]\label{def:wcVNE}
Given VN $S = (V_S, E_S) \in \mathcal{S}$ with demand $W$, PN $G = (V_G, E_G) \in \mathcal{G}$ with capacity $C$ and cost $T$, and constant $\theta$ is the desired cost.

\noindent Question: Is there an embedding $f \in \mathcal{D}(S, G)$, such that

\begin{gather}
    \forall e \in E_G: \sum_{\substack{x \in E_S \\ e \in f(x)}} w_{x}  \leq c_e \label{def:wcVNE:cap} \\
    Cost(f) \leq \theta \label{def:wcVNE:cost}
\end{gather}
\end{problem}

This definition includes both capacity restrictions (\ref{def:wcVNE:cap}) and cost optimization (\ref{def:wcVNE:cost}) metrics.
In this paper, we also consider two additional restricted VNE problems: with one of the two requirements.
VNE with the cost optimization is called \emph{weighted} (wVNE) and is examined in Section~\ref{sec:wVNE}.
VNE with the capacity restrictions is called \emph{capacitated} (cVNE) and is studied in Section~\ref{sec:cVNE}.
VNE with both metrics is considered in Section~\ref{sec:wcVNE} and mostly is referred to as wcVNE for clarity.

\paragraph{Topologies} 

The main focus of the paper is to study the complexity of VNE with different topologies of VN and of PN.
Let VNE \emph{variant} $\mathcal{S}$ topology on $\mathcal{G}$ topology be a special case of VNE with topological restrictions: VN $\in \mathcal{S}$ and PN $\in \mathcal{G}$.
It implies the problem of embedding a VN with topology $\mathcal{S}$ on a PN with topology $\mathcal{G}$.
In this paper we cover a variety of graph topologies, namely line, star, and $2$-star for VN; line and tree for PN.
We assume there is no topological restriction on VN or PN if it is not specified.

\paragraph{Topologies: uniform graphs}

The \emph{uniform graph} can be interpreted as a weighted graph with all weights equal to some constant $B$.
Without the loss of generality, we can assume $B = 1$.

For each graph topology covered in this paper, we also consider the VNE variations with uniform graphs: in PN the length of each edge is $1$ and in VN the weight of each edge is $1$.

\paragraph{Topologies: line}

A \emph{linear graph} (also referred to as \emph{line}) is a graph where vertices are connected sequentially from left to right.
We define an element at position $i$ from the left  as the $i$-th node.

\paragraph{Topologies: star}

A \emph{star graph} is a tree of depth one with a central vertex of degree $n - 1$.
For simplicity, we denote the index of a center as $n$ and $w_i = w_{i, n}$ is the weight of an edge from the $i$-th vertex to the central vertex.

\paragraph{Toplogies: $2$-star}

We call a tree of depth two a \emph{two-tiered star} or simply a $2$-star.
Additionally, we consider an oversubscribed $2$-star~\cite{BallaniCKR11_virtualNet} which is the special case of a $2$-star.
In the oversubscribed $2$-star each subtree on the second level (otherwise, known as \emph{group}) has the same number of leaves $s$. 
The demand of all edges in these subtrees is one.
The connections from these subtrees to the root are oversubscribed: 
the demand of an edge between the root and a subtree is $\frac{s}{o}$, where $o$ is the \emph{oversubscription factor}.
In the case where $o = s$, we get exactly a uniform $2$-star with subtrees of equal sizes.

\paragraph{Roadmap} 

In Section~\ref{sec:wVNE}, we consider wVNE, i.e., VNE without capacity requirements while optimizing the total length of requests.
In Section~\ref{sec:cVNE}, we consider cVNE, i.e., VNE with capacity requirements without the cost minimization.
In Section~\ref{sec:wcVNE}, we consider VNE (or wcVNE) with capacity requirements while minimizing the total length of requests.
In Section~\ref{sec:related} we discuss all the previous work connected to our problem.
We conclude with Section~\ref{sec:conclusion}.

\section{VNE with weights and without capacities}\label{sec:wVNE}

\emph{Weighted virtual network embedding} (wVNE) is a problem of finding an embedding $f \in \mathcal{D}(S, G)$ of minimal cost (without capacity restrictions).
For the wVNE problem we can assume that $f \in \mathcal{D}(S, G)$ always maps edges of VN on the cheapest paths of PN.
This is possible since we omit the capacity restrictions in this problem.
Hence, it is easy to check that the cheapest path is always the optimal.

The general wVNE without topological restrictions is NP-hard~\cite{HouidiLBZ11} and decision wVNE is NP-complete.
Since the generic variant of wVNE is NP-complete, then any restricted wVNE variant is at \textbf{worst} NP-complete.
Thus, it is sufficient to provide NP-hardness proof of wVNE variants to show their NP-completeness.

In addition to that, some restricted variants can be solved in polynomial time.
For example, exploiting a uniform star virtual cluster yields a polynomial algorithm~\cite{RostFS15_stars}.
In this section, we check whether some variants of wVNE are NP-complete or belong to P.
In particular, we deal with tree topologies such as line (Section~\ref{sec:wVNE_line}), star (Section~\ref{sec:wVNE_star}), or $2$-star (Section~\ref{sec:wVNE_2star}).

\subsection{wVNE of Linear VN}\label{sec:wVNE_line}

In this section we consider wVNE with VN being a line.
We start by exploring the case with uniform VN since NP-completeness results will also apply to non-uniform scenario.

\paragraph{Uniform linear VN topologies}

We show NP-completeness of wVNE with uniform linear VN and an arbitrary uniform PN. For that we use the following NP-complete problem~\cite{Karp72_NP}:

\begin{problem}[Hamiltonian path problem (HAM)]
    Given an undirected graph $G_H$, determine whether there exists a path that visits each vertex exactly once.
\end{problem}

Wu et al.~\cite{Wu20_path} proved a theorem, similar to Theorem~\ref{th:ULEP_wVNE_NPC}. 
They reduced Supereulerian graph problem to cVNE of uniform linear VN on a uniform PN.
We provide the same result for wVNE instead of cVNE.
Ignoring details, the reductions are similar.
However for our purposes, it is still necessary to show NP-completeness of embedding a uniform line in wVNE.

\begin{theorem}\label{th:ULEP_wVNE_NPC}
wVNE of a uniform linear VN on a uniform PN is NP-complete. \end{theorem}

\begin{proof}
    This can be shown by reducing HAM to wVNE.
    Let $G_H$ be the graph with $n$ vertices and $L_n$ be a line of length $n$.
    $G_H$ has a Hamiltonian path if and only if there exists an embedding $f \in \mathcal{D}(L_n, G_H)$ with $Cost(f) = n - 1$.

    We now show the correctness of the reduction.
    
    $\Rightarrow$: First we note that if there is some Hamiltonian path in $G_H$, then there is a correct embedding $f$ of $L_n$ onto $G_H$.
    That is, $f(i)$ is the vertex in $G$ that is $i$-th on the Hamiltonian path.
    Since each vertex is visited exactly once by the path, each edge on the path is also visited exactly once.
    Meaning, the cost of embedding equals the total length of a Hamiltonian path $Cost(f) = n - 1$.
    It cannot be smaller since the number of edges in line is $n - 1$.

    $\Leftarrow$: Now, we show that if there is a correct mapping $f \in \mathcal{D}(L, G_H)$ of cost $n - 1$, then there exists a Hamiltonian path in $G_H$.
    Since $Cost(f) = n - 1$ and $L_n$ contains $n - 1$ edges, each edge of $L$ is placed on a single edge of $G_H$.
    
    Consequently, the mapping forms a non-self-intersecting path of length $n - 1$, which is a Hamiltonian path.
\end{proof}

As for other VNE scenarios with uniform line VN: it can be embedded on a tree PN in polynomial time as discussed further in Section~\ref{sec:line_on_tree_DP}.

\paragraph{Weighted linear VN topologies}

In this paragraph, we discuss the variation of the wVNE problem where we have a linear VN and a linear PN.
In Theorem~\ref{th:WLLEP_NPC} we prove that this problem is NP-complete.

First, we need to introduce another well-known NP-complete problem.

\begin{problem}[Bin packing problem (BPP)]
    Given an array of $n$ positive integers $A = [a_1, a_2, \ldots, a_n]$ and two positive integers $B$ and $K$, such that $\sum a_i \leq B \cdot K$.
    
    \noindent Question: Is it possible to partition $A$ into $K$ disjoint subsets, such that the sum of the elements in each subset does not exceed $B$?
\end{problem}

BPP is strongly NP-complete~\cite{GareyJ79_NP}, meaning that it remains NP-complete even when the numerical inputs are in the unary notation, i.e., $a_i$ is presented with $a_i$ ones.
In other words, it is NP-complete when $\sum a_i$ is a variable.

Now, we can prove the NP-hardness of linear VN on linear PN wVNE problem by reducing the BPP to it.

\begin{figure}[!t]
    \centering

\tikzset{every picture/.style={line width=0.75pt}} 

\begin{tikzpicture}[x=0.75pt,y=0.75pt,yscale=-1,xscale=1]

\draw   (80.53,252.94) .. controls (80.53,257.61) and (82.86,259.94) .. (87.53,259.94) -- (160.27,259.94) .. controls (166.94,259.94) and (170.27,262.27) .. (170.27,266.94) .. controls (170.27,262.27) and (173.6,259.94) .. (180.27,259.94)(177.27,259.94) -- (253,259.94) .. controls (257.67,259.94) and (260,257.61) .. (260,252.94) ;
\draw   (280,253) .. controls (280,257.67) and (282.33,260) .. (287,260) -- (360,260) .. controls (366.67,260) and (370,262.33) .. (370,267) .. controls (370,262.33) and (373.33,260) .. (380,260)(377,260) -- (453,260) .. controls (457.67,260) and (460,257.67) .. (460,253) ;
\draw   (160,293.37) .. controls (160.01,298.04) and (162.34,300.37) .. (167.01,300.36) -- (260.01,300.28) .. controls (266.68,300.27) and (270.01,302.6) .. (270.01,307.27) .. controls (270.01,302.6) and (273.34,300.27) .. (280.01,300.26)(277.01,300.26) -- (373.01,300.18) .. controls (377.68,300.17) and (380.01,297.84) .. (380,293.17) ;
\draw [color={rgb, 255:red, 214; green, 39; blue, 40 }  ,draw opacity=1 ]   (90,140) -- (130,140) ;
\draw [color={rgb, 255:red, 214; green, 39; blue, 40 }  ,draw opacity=1 ]   (130,140) -- (170,140) ;
\draw [color={rgb, 255:red, 214; green, 39; blue, 40 }  ,draw opacity=1 ]   (170,140) -- (210,140) ;
\draw [color={rgb, 255:red, 214; green, 39; blue, 40 }  ,draw opacity=1 ]   (210,140) -- (250,140) ;
\draw    (250,140) -- (290,140) ;
\draw [color={rgb, 255:red, 31; green, 119; blue, 180 }  ,draw opacity=1 ]   (290,140) -- (330,140) ;
\draw [color={rgb, 255:red, 31; green, 119; blue, 180 }  ,draw opacity=1 ]   (330,140) -- (370,140) ;
\draw    (370,140) -- (410,140) ;
\draw    (410,140) -- (450,140) ;
\draw  [color={rgb, 255:red, 214; green, 39; blue, 40 }  ,draw opacity=1 ][fill={rgb, 255:red, 255; green, 255; blue, 255 }  ,fill opacity=1 ] (80,140) .. controls (80,134.48) and (84.48,130) .. (90,130) .. controls (95.52,130) and (100,134.48) .. (100,140) .. controls (100,145.52) and (95.52,150) .. (90,150) .. controls (84.48,150) and (80,145.52) .. (80,140) -- cycle ;
\draw  [color={rgb, 255:red, 214; green, 39; blue, 40 }  ,draw opacity=1 ][fill={rgb, 255:red, 255; green, 255; blue, 255 }  ,fill opacity=1 ] (120,140) .. controls (120,134.48) and (124.48,130) .. (130,130) .. controls (135.52,130) and (140,134.48) .. (140,140) .. controls (140,145.52) and (135.52,150) .. (130,150) .. controls (124.48,150) and (120,145.52) .. (120,140) -- cycle ;
\draw  [color={rgb, 255:red, 214; green, 39; blue, 40 }  ,draw opacity=1 ][fill={rgb, 255:red, 255; green, 255; blue, 255 }  ,fill opacity=1 ] (160,140) .. controls (160,134.48) and (164.48,130) .. (170,130) .. controls (175.52,130) and (180,134.48) .. (180,140) .. controls (180,145.52) and (175.52,150) .. (170,150) .. controls (164.48,150) and (160,145.52) .. (160,140) -- cycle ;
\draw  [color={rgb, 255:red, 214; green, 39; blue, 40 }  ,draw opacity=1 ][fill={rgb, 255:red, 255; green, 255; blue, 255 }  ,fill opacity=1 ] (200,140) .. controls (200,134.48) and (204.48,130) .. (210,130) .. controls (215.52,130) and (220,134.48) .. (220,140) .. controls (220,145.52) and (215.52,150) .. (210,150) .. controls (204.48,150) and (200,145.52) .. (200,140) -- cycle ;
\draw  [color={rgb, 255:red, 214; green, 39; blue, 40 }  ,draw opacity=1 ][fill={rgb, 255:red, 255; green, 255; blue, 255 }  ,fill opacity=1 ] (240,140) .. controls (240,134.48) and (244.48,130) .. (250,130) .. controls (255.52,130) and (260,134.48) .. (260,140) .. controls (260,145.52) and (255.52,150) .. (250,150) .. controls (244.48,150) and (240,145.52) .. (240,140) -- cycle ;
\draw  [color={rgb, 255:red, 31; green, 119; blue, 180 }  ,draw opacity=1 ][fill={rgb, 255:red, 255; green, 255; blue, 255 }  ,fill opacity=1 ] (280,140) .. controls (280,134.48) and (284.48,130) .. (290,130) .. controls (295.52,130) and (300,134.48) .. (300,140) .. controls (300,145.52) and (295.52,150) .. (290,150) .. controls (284.48,150) and (280,145.52) .. (280,140) -- cycle ;
\draw  [color={rgb, 255:red, 31; green, 119; blue, 180 }  ,draw opacity=1 ][fill={rgb, 255:red, 255; green, 255; blue, 255 }  ,fill opacity=1 ] (320,140) .. controls (320,134.48) and (324.48,130) .. (330,130) .. controls (335.52,130) and (340,134.48) .. (340,140) .. controls (340,145.52) and (335.52,150) .. (330,150) .. controls (324.48,150) and (320,145.52) .. (320,140) -- cycle ;
\draw  [color={rgb, 255:red, 31; green, 119; blue, 180 }  ,draw opacity=1 ][fill={rgb, 255:red, 255; green, 255; blue, 255 }  ,fill opacity=1 ] (360,140) .. controls (360,134.48) and (364.48,130) .. (370,130) .. controls (375.52,130) and (380,134.48) .. (380,140) .. controls (380,145.52) and (375.52,150) .. (370,150) .. controls (364.48,150) and (360,145.52) .. (360,140) -- cycle ;
\draw  [color={rgb, 255:red, 44; green, 160; blue, 44 }  ,draw opacity=1 ][fill={rgb, 255:red, 255; green, 255; blue, 255 }  ,fill opacity=1 ] (400,140) .. controls (400,134.48) and (404.48,130) .. (410,130) .. controls (415.52,130) and (420,134.48) .. (420,140) .. controls (420,145.52) and (415.52,150) .. (410,150) .. controls (404.48,150) and (400,145.52) .. (400,140) -- cycle ;
\draw  [color={rgb, 255:red, 0; green, 0; blue, 0 }  ,draw opacity=1 ][fill={rgb, 255:red, 255; green, 255; blue, 255 }  ,fill opacity=1 ] (440,140) .. controls (440,134.48) and (444.48,130) .. (450,130) .. controls (455.52,130) and (460,134.48) .. (460,140) .. controls (460,145.52) and (455.52,150) .. (450,150) .. controls (444.48,150) and (440,145.52) .. (440,140) -- cycle ;
\draw [color={rgb, 255:red, 214; green, 39; blue, 40 }  ,draw opacity=1 ]   (90,230) -- (130,230) ;
\draw [color={rgb, 255:red, 214; green, 39; blue, 40 }  ,draw opacity=1 ]   (130,230) -- (170,230) ;
\draw [color={rgb, 255:red, 214; green, 39; blue, 40 }  ,draw opacity=1 ]   (170,230) -- (210,230) ;
\draw [color={rgb, 255:red, 214; green, 39; blue, 40 }  ,draw opacity=1 ]   (210,230) -- (250,230) ;
\draw    (250,230) -- (290,230) ;
\draw [color={rgb, 255:red, 31; green, 119; blue, 180 }  ,draw opacity=1 ]   (290,230) -- (330,230) ;
\draw [color={rgb, 255:red, 31; green, 119; blue, 180 }  ,draw opacity=1 ]   (330,230) -- (370,230) ;
\draw    (370,230) -- (410,230) ;
\draw    (410,230) -- (450,230) ;
\draw  [color={rgb, 255:red, 214; green, 39; blue, 40 }  ,draw opacity=1 ][fill={rgb, 255:red, 255; green, 255; blue, 255 }  ,fill opacity=1 ] (80,230) .. controls (80,224.48) and (84.48,220) .. (90,220) .. controls (95.52,220) and (100,224.48) .. (100,230) .. controls (100,235.52) and (95.52,240) .. (90,240) .. controls (84.48,240) and (80,235.52) .. (80,230) -- cycle ;
\draw  [color={rgb, 255:red, 214; green, 39; blue, 40 }  ,draw opacity=1 ][fill={rgb, 255:red, 255; green, 255; blue, 255 }  ,fill opacity=1 ] (120,230) .. controls (120,224.48) and (124.48,220) .. (130,220) .. controls (135.52,220) and (140,224.48) .. (140,230) .. controls (140,235.52) and (135.52,240) .. (130,240) .. controls (124.48,240) and (120,235.52) .. (120,230) -- cycle ;
\draw  [color={rgb, 255:red, 214; green, 39; blue, 40 }  ,draw opacity=1 ][fill={rgb, 255:red, 255; green, 255; blue, 255 }  ,fill opacity=1 ] (160,230) .. controls (160,224.48) and (164.48,220) .. (170,220) .. controls (175.52,220) and (180,224.48) .. (180,230) .. controls (180,235.52) and (175.52,240) .. (170,240) .. controls (164.48,240) and (160,235.52) .. (160,230) -- cycle ;
\draw  [color={rgb, 255:red, 214; green, 39; blue, 40 }  ,draw opacity=1 ][fill={rgb, 255:red, 255; green, 255; blue, 255 }  ,fill opacity=1 ] (200,230) .. controls (200,224.48) and (204.48,220) .. (210,220) .. controls (215.52,220) and (220,224.48) .. (220,230) .. controls (220,235.52) and (215.52,240) .. (210,240) .. controls (204.48,240) and (200,235.52) .. (200,230) -- cycle ;
\draw  [color={rgb, 255:red, 214; green, 39; blue, 40 }  ,draw opacity=1 ][fill={rgb, 255:red, 255; green, 255; blue, 255 }  ,fill opacity=1 ] (240,230) .. controls (240,224.48) and (244.48,220) .. (250,220) .. controls (255.52,220) and (260,224.48) .. (260,230) .. controls (260,235.52) and (255.52,240) .. (250,240) .. controls (244.48,240) and (240,235.52) .. (240,230) -- cycle ;
\draw  [color={rgb, 255:red, 31; green, 119; blue, 180 }  ,draw opacity=1 ][fill={rgb, 255:red, 255; green, 255; blue, 255 }  ,fill opacity=1 ] (280,230) .. controls (280,224.48) and (284.48,220) .. (290,220) .. controls (295.52,220) and (300,224.48) .. (300,230) .. controls (300,235.52) and (295.52,240) .. (290,240) .. controls (284.48,240) and (280,235.52) .. (280,230) -- cycle ;
\draw  [color={rgb, 255:red, 31; green, 119; blue, 180 }  ,draw opacity=1 ][fill={rgb, 255:red, 255; green, 255; blue, 255 }  ,fill opacity=1 ] (320,230) .. controls (320,224.48) and (324.48,220) .. (330,220) .. controls (335.52,220) and (340,224.48) .. (340,230) .. controls (340,235.52) and (335.52,240) .. (330,240) .. controls (324.48,240) and (320,235.52) .. (320,230) -- cycle ;
\draw  [color={rgb, 255:red, 31; green, 119; blue, 180 }  ,draw opacity=1 ][fill={rgb, 255:red, 255; green, 255; blue, 255 }  ,fill opacity=1 ] (360,230) .. controls (360,224.48) and (364.48,220) .. (370,220) .. controls (375.52,220) and (380,224.48) .. (380,230) .. controls (380,235.52) and (375.52,240) .. (370,240) .. controls (364.48,240) and (360,235.52) .. (360,230) -- cycle ;
\draw  [color={rgb, 255:red, 44; green, 160; blue, 44 }  ,draw opacity=1 ][fill={rgb, 255:red, 255; green, 255; blue, 255 }  ,fill opacity=1 ] (400,230) .. controls (400,224.48) and (404.48,220) .. (410,220) .. controls (415.52,220) and (420,224.48) .. (420,230) .. controls (420,235.52) and (415.52,240) .. (410,240) .. controls (404.48,240) and (400,235.52) .. (400,230) -- cycle ;
\draw  [color={rgb, 255:red, 0; green, 0; blue, 0 }  ,draw opacity=1 ][fill={rgb, 255:red, 255; green, 255; blue, 255 }  ,fill opacity=1 ] (440,230) .. controls (440,224.48) and (444.48,220) .. (450,220) .. controls (455.52,220) and (460,224.48) .. (460,230) .. controls (460,235.52) and (455.52,240) .. (450,240) .. controls (444.48,240) and (440,235.52) .. (440,230) -- cycle ;
\draw [color={rgb, 255:red, 214; green, 39; blue, 40 }  ,draw opacity=1 ]   (140,70) -- (168.46,117.43) ;
\draw [shift={(170,120)}, rotate = 239.04] [fill={rgb, 255:red, 214; green, 39; blue, 40 }  ,fill opacity=1 ][line width=0.08]  [draw opacity=0] (10.72,-5.15) -- (0,0) -- (10.72,5.15) -- (7.12,0) -- cycle    ;
\draw [color={rgb, 255:red, 214; green, 39; blue, 40 }  ,draw opacity=1 ]   (170,160) -- (170,207) ;
\draw [shift={(170,210)}, rotate = 270] [fill={rgb, 255:red, 214; green, 39; blue, 40 }  ,fill opacity=1 ][line width=0.08]  [draw opacity=0] (10.72,-5.15) -- (0,0) -- (10.72,5.15) -- (7.12,0) -- cycle    ;
\draw [color={rgb, 255:red, 0; green, 0; blue, 0 }  ,draw opacity=1 ] [dash pattern={on 4.5pt off 4.5pt}]  (450,160) -- (450,207) ;
\draw [shift={(450,210)}, rotate = 270] [fill={rgb, 255:red, 0; green, 0; blue, 0 }  ,fill opacity=1 ][line width=0.08]  [draw opacity=0] (10.72,-5.15) -- (0,0) -- (10.72,5.15) -- (7.12,0) -- cycle    ;
\draw [color={rgb, 255:red, 31; green, 119; blue, 180 }  ,draw opacity=1 ]   (330,160) -- (330,207) ;
\draw [shift={(330,210)}, rotate = 270] [fill={rgb, 255:red, 31; green, 119; blue, 180 }  ,fill opacity=1 ][line width=0.08]  [draw opacity=0] (10.72,-5.15) -- (0,0) -- (10.72,5.15) -- (7.12,0) -- cycle    ;
\draw [color={rgb, 255:red, 44; green, 160; blue, 44 }  ,draw opacity=1 ]   (410,160) -- (410,207) ;
\draw [shift={(410,210)}, rotate = 270] [fill={rgb, 255:red, 44; green, 160; blue, 44 }  ,fill opacity=1 ][line width=0.08]  [draw opacity=0] (10.72,-5.15) -- (0,0) -- (10.72,5.15) -- (7.12,0) -- cycle    ;
\draw [color={rgb, 255:red, 31; green, 119; blue, 180 }  ,draw opacity=1 ]   (175,70) -- (327.14,119.08) ;
\draw [shift={(330,120)}, rotate = 197.88] [fill={rgb, 255:red, 31; green, 119; blue, 180 }  ,fill opacity=1 ][line width=0.08]  [draw opacity=0] (10.72,-5.15) -- (0,0) -- (10.72,5.15) -- (7.12,0) -- cycle    ;
\draw [color={rgb, 255:red, 44; green, 160; blue, 44 }  ,draw opacity=1 ]   (210,70) -- (407.09,119.27) ;
\draw [shift={(410,120)}, rotate = 194.04] [fill={rgb, 255:red, 44; green, 160; blue, 44 }  ,fill opacity=1 ][line width=0.08]  [draw opacity=0] (10.72,-5.15) -- (0,0) -- (10.72,5.15) -- (7.12,0) -- cycle    ;

\draw (32,61) node [anchor=west] [inner sep=0.75pt]  [font=\large]  {$A\ \ \ \ \ =\ \ \ \ \ \ [\textcolor[rgb]{0.84,0.15,0.16}{5} ,\ \ \ \textcolor[rgb]{0.12,0.47,0.71}{3} ,\ \ \ \textcolor[rgb]{0.17,0.63,0.17}{1}]$};
\draw (171,281) node    {$B$};
\draw (371,281) node    {$B$};
\draw (271,321) node    {$K$};
\draw (110.5,131) node    {$1$};
\draw (150.5,131) node    {$1$};
\draw (190.5,131) node    {$1$};
\draw (230.5,131) node    {$1$};
\draw (270.5,131) node    {$0$};
\draw (310.5,131) node    {$1$};
\draw (350.5,131) node    {$1$};
\draw (390.5,131) node    {$0$};
\draw (430.5,131) node    {$0$};
\draw (110.5,221) node    {$0$};
\draw (150.5,221) node    {$0$};
\draw (190.5,221) node    {$0$};
\draw (230.5,221) node    {$0$};
\draw (270.5,221) node    {$1$};
\draw (310.5,221) node    {$0$};
\draw (350.5,221) node    {$0$};
\draw (390.5,221) node    {$0$};
\draw (430.5,221) node    {$0$};
\draw (33,141) node [anchor=west] [inner sep=0.75pt]  [font=\large] [align=left] {VN:};
\draw (32,231) node [anchor=west] [inner sep=0.75pt]  [font=\large] [align=left] {PN:};

\end{tikzpicture}

    \caption{Reducing BPP to line on line wVNE for $A = [5, 3, 1], B = 5, K = 2$.}
    \label{fig:BPP_to_WLL}
    
\end{figure}
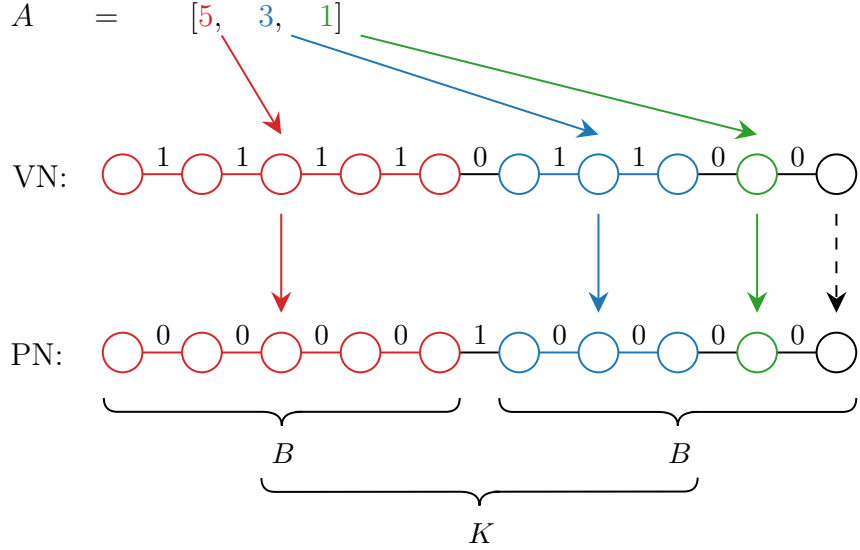

\begin{theorem}\label{th:WLLEP_NPC}
wVNE of a linear VN on a linear PN is NP-complete.
\end{theorem}
\begin{proof}
    The theorem is proven by the reduction of BPP with inputs in the unary notation to line-on-line wVNE as follows.
    
    Suppose we have some instance of BPP:
    
    $$\langle A = [a_1, a_2, \ldots, a_n], B, K \rangle$$
    
    For each $a_i \in A$, a linear graph $s_i$ with $a_i$ vertices is constructed. These $s_i$ graphs are called \emph{``element sections''}.
    Additionally, to match VN and PN sizes we create $B \cdot K - \sum a_i$ ``singleton sections'' that contain only one vertex each.
    All edges in ``element sections'', if there are any, have weight $1$. 
    To build a required VN $S$, all ``element sections'' are connected sequentially by $0$ weighted edges.
    (We allow the edges to have weight $0$ as defined in Section~\ref{subsec:model}).
    In other words, the last vertex of section $s_i$ is connected to the first vertex of section $s_{i+1}$ for $1 = 1, 2, \ldots, n - 1$.
    The size of VN $|S| = \sum a_i + (B \cdot K - \sum a_i) = B \cdot K$ by construction.

    PN $G$ consists of $K$ \emph{``bin sections''} $L_1 \ldots L_K$. Each ``bin section'' is a linear graph with $B$ vertices each.
    All edges in ``bin sections'' have $0$ weights.
    The ``bin sections'' are connected with $1$-weighted edges sequentially.
    
    For the clarification, please, see the example on Figure~\ref{fig:BPP_to_WLL}.

    There is a positive solution for BPP if and only if there is wVNE solution of $0$ cost for the graphs $S$ and $G$ constructed above.

    $\Rightarrow$: Assume there is some partition of $A$ into $K$ disjoint subsets $I_1, I_2, \ldots I_K$. 
    Then, there exists an embedding $f$, such that for all $j \leq K$ it maps all vertices of $s_i$ for all $a_i \in I_j$ to vertices of $L_j$. 
    That is possible since $|L_j| = B \geq |I_j|$. 
    Due to this mapping, all $1$-weighted edges from $S$ are mapped to $0$-weighted paths.
    Where we map ``singleton sections'' does not matter, since the demand of all the adjacent edges of ``singleton sections'' is $0$.
    Hence, $Cost(f) = 0$.

    $\Leftarrow$: Now, assume that for some input of BPP $\langle A, B, K \rangle$ there exist an embedding $f \in \mathcal{D}(S, G)$, such that $Cost(f) = 0$.
    That is only possible when every $1$-weighted edge from $S$ is mapped to a $0$-weighted path from $G$. 
    Mapping of all vertices of ``element section'' $s_i$ to a single ``bin section'' $L_j$ for all $i \leq n$ follows (otherwise there exist  $1$ weighted edge from $s_i$, that is not mapped to a $0$ weighted path).
    If ``element section'' $s_i$ is mapped to ``bin section'' $L_j$, element $a_i$ should be put in the bin $I_j$.

    The given reduction of BPP to line-on-line wVNE is $O(B \cdot K + \sum a_i)$ in time, which is polynomial in the length of the input.
    Note that the input is given in the unary notation.
\end{proof}

\paragraph{Weighted linear VN on uniform PN}

Knowing that the embedding of a uniform line on a graph with an arbitrary topology is NP-complete, the reasonable question to ask is whether it is NP-complete to embed a line on a uniform graph or at least on graphs with some restricted topology.
Embedding a line on a uniform line is straightforward.
So, we consider a more generic problem of embedding a line on a uniform tree.

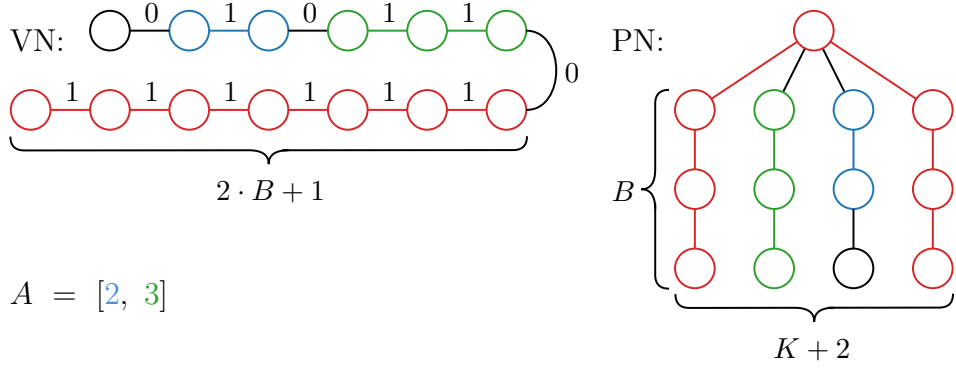
\begin{figure}[!t]
    \centering

\tikzset{every picture/.style={line width=0.75pt}} 

\begin{tikzpicture}[x=0.75pt,y=0.75pt,yscale=-1,xscale=1]

\draw  [color={rgb, 255:red, 44; green, 160; blue, 44 }  ,draw opacity=1 ][fill={rgb, 255:red, 255; green, 255; blue, 255 }  ,fill opacity=1 ] (470,210) .. controls (470,204.48) and (474.48,200) .. (480,200) .. controls (485.52,200) and (490,204.48) .. (490,210) .. controls (490,215.52) and (485.52,220) .. (480,220) .. controls (474.48,220) and (470,215.52) .. (470,210) -- cycle ;
\draw  [color={rgb, 255:red, 44; green, 160; blue, 44 }  ,draw opacity=1 ][fill={rgb, 255:red, 255; green, 255; blue, 255 }  ,fill opacity=1 ] (470,250) .. controls (470,244.48) and (474.48,240) .. (480,240) .. controls (485.52,240) and (490,244.48) .. (490,250) .. controls (490,255.52) and (485.52,260) .. (480,260) .. controls (474.48,260) and (470,255.52) .. (470,250) -- cycle ;
\draw [color={rgb, 255:red, 44; green, 160; blue, 44 }  ,draw opacity=1 ]   (480,180) -- (480,200) ;
\draw    (500,130) -- (480,170) ;
\draw  [color={rgb, 255:red, 31; green, 119; blue, 180 }  ,draw opacity=1 ][fill={rgb, 255:red, 255; green, 255; blue, 255 }  ,fill opacity=1 ] (510,210) .. controls (510,204.48) and (514.48,200) .. (520,200) .. controls (525.52,200) and (530,204.48) .. (530,210) .. controls (530,215.52) and (525.52,220) .. (520,220) .. controls (514.48,220) and (510,215.52) .. (510,210) -- cycle ;
\draw  [fill={rgb, 255:red, 255; green, 255; blue, 255 }  ,fill opacity=1 ] (510,250) .. controls (510,244.48) and (514.48,240) .. (520,240) .. controls (525.52,240) and (530,244.48) .. (530,250) .. controls (530,255.52) and (525.52,260) .. (520,260) .. controls (514.48,260) and (510,255.52) .. (510,250) -- cycle ;
\draw [color={rgb, 255:red, 31; green, 119; blue, 180 }  ,draw opacity=1 ]   (520,180) -- (520,200) ;
\draw  [color={rgb, 255:red, 214; green, 39; blue, 40 }  ,draw opacity=1 ][fill={rgb, 255:red, 255; green, 255; blue, 255 }  ,fill opacity=1 ] (550,210) .. controls (550,204.48) and (554.48,200) .. (560,200) .. controls (565.52,200) and (570,204.48) .. (570,210) .. controls (570,215.52) and (565.52,220) .. (560,220) .. controls (554.48,220) and (550,215.52) .. (550,210) -- cycle ;
\draw  [color={rgb, 255:red, 214; green, 39; blue, 40 }  ,draw opacity=1 ][fill={rgb, 255:red, 255; green, 255; blue, 255 }  ,fill opacity=1 ] (550,250) .. controls (550,244.48) and (554.48,240) .. (560,240) .. controls (565.52,240) and (570,244.48) .. (570,250) .. controls (570,255.52) and (565.52,260) .. (560,260) .. controls (554.48,260) and (550,255.52) .. (550,250) -- cycle ;
\draw [color={rgb, 255:red, 214; green, 39; blue, 40 }  ,draw opacity=1 ]   (560,180) -- (560,200) ;
\draw  [color={rgb, 255:red, 214; green, 39; blue, 40 }  ,draw opacity=1 ][fill={rgb, 255:red, 255; green, 255; blue, 255 }  ,fill opacity=1 ] (430,210) .. controls (430,204.48) and (434.48,200) .. (440,200) .. controls (445.52,200) and (450,204.48) .. (450,210) .. controls (450,215.52) and (445.52,220) .. (440,220) .. controls (434.48,220) and (430,215.52) .. (430,210) -- cycle ;
\draw  [color={rgb, 255:red, 214; green, 39; blue, 40 }  ,draw opacity=1 ][fill={rgb, 255:red, 255; green, 255; blue, 255 }  ,fill opacity=1 ] (430,250) .. controls (430,244.48) and (434.48,240) .. (440,240) .. controls (445.52,240) and (450,244.48) .. (450,250) .. controls (450,255.52) and (445.52,260) .. (440,260) .. controls (434.48,260) and (430,255.52) .. (430,250) -- cycle ;
\draw [color={rgb, 255:red, 214; green, 39; blue, 40 }  ,draw opacity=1 ]   (440,180) -- (440,200) ;
\draw    (500,130) -- (520,170) ;
\draw [color={rgb, 255:red, 214; green, 39; blue, 40 }  ,draw opacity=1 ]   (500,130) -- (560,170) ;
\draw  [color={rgb, 255:red, 44; green, 160; blue, 44 }  ,draw opacity=1 ][fill={rgb, 255:red, 255; green, 255; blue, 255 }  ,fill opacity=1 ] (470,170) .. controls (470,164.48) and (474.48,160) .. (480,160) .. controls (485.52,160) and (490,164.48) .. (490,170) .. controls (490,175.52) and (485.52,180) .. (480,180) .. controls (474.48,180) and (470,175.52) .. (470,170) -- cycle ;
\draw  [color={rgb, 255:red, 31; green, 119; blue, 180 }  ,draw opacity=1 ][fill={rgb, 255:red, 255; green, 255; blue, 255 }  ,fill opacity=1 ] (510,170) .. controls (510,164.48) and (514.48,160) .. (520,160) .. controls (525.52,160) and (530,164.48) .. (530,170) .. controls (530,175.52) and (525.52,180) .. (520,180) .. controls (514.48,180) and (510,175.52) .. (510,170) -- cycle ;
\draw  [color={rgb, 255:red, 214; green, 39; blue, 40 }  ,draw opacity=1 ][fill={rgb, 255:red, 255; green, 255; blue, 255 }  ,fill opacity=1 ] (550,170) .. controls (550,164.48) and (554.48,160) .. (560,160) .. controls (565.52,160) and (570,164.48) .. (570,170) .. controls (570,175.52) and (565.52,180) .. (560,180) .. controls (554.48,180) and (550,175.52) .. (550,170) -- cycle ;
\draw   (430.2,262.97) .. controls (430.21,267.64) and (432.54,269.97) .. (437.21,269.96) -- (490.21,269.88) .. controls (496.88,269.87) and (500.21,272.2) .. (500.22,276.87) .. controls (500.21,272.2) and (503.54,269.86) .. (510.21,269.85)(507.21,269.86) -- (563.21,269.78) .. controls (567.88,269.77) and (570.21,267.44) .. (570.2,262.77) ;
\draw   (427,159.89) .. controls (422.33,159.89) and (420,162.22) .. (420,166.89) -- (420,200.06) .. controls (420,206.73) and (417.67,210.06) .. (413,210.06) .. controls (417.67,210.06) and (420,213.39) .. (420,220.06)(420,217.06) -- (420,253.24) .. controls (420,257.91) and (422.33,260.24) .. (427,260.24) ;
\draw [color={rgb, 255:red, 214; green, 39; blue, 40 }  ,draw opacity=1 ]   (500,130) -- (440,170) ;
\draw  [color={rgb, 255:red, 214; green, 39; blue, 40 }  ,draw opacity=1 ][fill={rgb, 255:red, 255; green, 255; blue, 255 }  ,fill opacity=1 ] (430,170) .. controls (430,164.48) and (434.48,160) .. (440,160) .. controls (445.52,160) and (450,164.48) .. (450,170) .. controls (450,175.52) and (445.52,180) .. (440,180) .. controls (434.48,180) and (430,175.52) .. (430,170) -- cycle ;
\draw  [color={rgb, 255:red, 214; green, 39; blue, 40 }  ,draw opacity=1 ][fill={rgb, 255:red, 255; green, 255; blue, 255 }  ,fill opacity=1 ] (490,130) .. controls (490,124.48) and (494.48,120) .. (500,120) .. controls (505.52,120) and (510,124.48) .. (510,130) .. controls (510,135.52) and (505.52,140) .. (500,140) .. controls (494.48,140) and (490,135.52) .. (490,130) -- cycle ;
\draw  [color={rgb, 255:red, 214; green, 39; blue, 40 }  ,draw opacity=1 ][fill={rgb, 255:red, 255; green, 255; blue, 255 }  ,fill opacity=1 ] (214.96,170) .. controls (214.96,164.48) and (219.44,160) .. (224.96,160) .. controls (230.48,160) and (234.96,164.48) .. (234.96,170) .. controls (234.96,175.52) and (230.48,180) .. (224.96,180) .. controls (219.44,180) and (214.96,175.52) .. (214.96,170) -- cycle ;
\draw  [color={rgb, 255:red, 214; green, 39; blue, 40 }  ,draw opacity=1 ][fill={rgb, 255:red, 255; green, 255; blue, 255 }  ,fill opacity=1 ] (254.96,170) .. controls (254.96,164.48) and (259.44,160) .. (264.96,160) .. controls (270.48,160) and (274.96,164.48) .. (274.96,170) .. controls (274.96,175.52) and (270.48,180) .. (264.96,180) .. controls (259.44,180) and (254.96,175.52) .. (254.96,170) -- cycle ;
\draw  [color={rgb, 255:red, 214; green, 39; blue, 40 }  ,draw opacity=1 ][fill={rgb, 255:red, 255; green, 255; blue, 255 }  ,fill opacity=1 ] (294.96,170) .. controls (294.96,164.48) and (299.44,160) .. (304.96,160) .. controls (310.48,160) and (314.96,164.48) .. (314.96,170) .. controls (314.96,175.52) and (310.48,180) .. (304.96,180) .. controls (299.44,180) and (294.96,175.52) .. (294.96,170) -- cycle ;
\draw  [color={rgb, 255:red, 214; green, 39; blue, 40 }  ,draw opacity=1 ][fill={rgb, 255:red, 255; green, 255; blue, 255 }  ,fill opacity=1 ] (334.96,170) .. controls (334.96,164.48) and (339.44,160) .. (344.96,160) .. controls (350.48,160) and (354.96,164.48) .. (354.96,170) .. controls (354.96,175.52) and (350.48,180) .. (344.96,180) .. controls (339.44,180) and (334.96,175.52) .. (334.96,170) -- cycle ;
\draw [color={rgb, 255:red, 214; green, 39; blue, 40 }  ,draw opacity=1 ]   (234.96,170) -- (254.96,170) ;
\draw [color={rgb, 255:red, 214; green, 39; blue, 40 }  ,draw opacity=1 ]   (274.96,170) -- (294.96,170) ;
\draw [color={rgb, 255:red, 214; green, 39; blue, 40 }  ,draw opacity=1 ]   (314.96,170) -- (334.96,170) ;
\draw  [color={rgb, 255:red, 31; green, 119; blue, 180 }  ,draw opacity=1 ][fill={rgb, 255:red, 255; green, 255; blue, 255 }  ,fill opacity=1 ] (214.96,130) .. controls (214.96,124.48) and (219.44,120) .. (224.96,120) .. controls (230.48,120) and (234.96,124.48) .. (234.96,130) .. controls (234.96,135.52) and (230.48,140) .. (224.96,140) .. controls (219.44,140) and (214.96,135.52) .. (214.96,130) -- cycle ;
\draw [color={rgb, 255:red, 0; green, 0; blue, 0 }  ,draw opacity=1 ]   (234.96,130) -- (254.96,130) ;
\draw  [color={rgb, 255:red, 31; green, 119; blue, 180 }  ,draw opacity=1 ][fill={rgb, 255:red, 255; green, 255; blue, 255 }  ,fill opacity=1 ] (175,130) .. controls (175,124.48) and (179.48,120) .. (185,120) .. controls (190.52,120) and (195,124.48) .. (195,130) .. controls (195,135.52) and (190.52,140) .. (185,140) .. controls (179.48,140) and (175,135.52) .. (175,130) -- cycle ;
\draw  [color={rgb, 255:red, 44; green, 160; blue, 44 }  ,draw opacity=1 ][fill={rgb, 255:red, 255; green, 255; blue, 255 }  ,fill opacity=1 ] (254.96,130) .. controls (254.96,124.48) and (259.44,120) .. (264.96,120) .. controls (270.48,120) and (274.96,124.48) .. (274.96,130) .. controls (274.96,135.52) and (270.48,140) .. (264.96,140) .. controls (259.44,140) and (254.96,135.52) .. (254.96,130) -- cycle ;
\draw  [color={rgb, 255:red, 44; green, 160; blue, 44 }  ,draw opacity=1 ][fill={rgb, 255:red, 255; green, 255; blue, 255 }  ,fill opacity=1 ] (294.96,130) .. controls (294.96,124.48) and (299.44,120) .. (304.96,120) .. controls (310.48,120) and (314.96,124.48) .. (314.96,130) .. controls (314.96,135.52) and (310.48,140) .. (304.96,140) .. controls (299.44,140) and (294.96,135.52) .. (294.96,130) -- cycle ;
\draw  [color={rgb, 255:red, 44; green, 160; blue, 44 }  ,draw opacity=1 ][fill={rgb, 255:red, 255; green, 255; blue, 255 }  ,fill opacity=1 ] (334.96,130) .. controls (334.96,124.48) and (339.44,120) .. (344.96,120) .. controls (350.48,120) and (354.96,124.48) .. (354.96,130) .. controls (354.96,135.52) and (350.48,140) .. (344.96,140) .. controls (339.44,140) and (334.96,135.52) .. (334.96,130) -- cycle ;
\draw [color={rgb, 255:red, 31; green, 119; blue, 180 }  ,draw opacity=1 ]   (195,130) -- (215,130) ;
\draw [color={rgb, 255:red, 44; green, 160; blue, 44 }  ,draw opacity=1 ]   (274.96,130) -- (294.96,130) ;
\draw [color={rgb, 255:red, 44; green, 160; blue, 44 }  ,draw opacity=1 ]   (314.96,130) -- (334.96,130) ;
\draw [color={rgb, 255:red, 0; green, 0; blue, 0 }  ,draw opacity=1 ]   (354.96,170) .. controls (375.08,170.19) and (375.21,130.19) .. (354.96,130) ;
\draw [color={rgb, 255:red, 44; green, 160; blue, 44 }  ,draw opacity=1 ]   (480,220) -- (480,240) ;
\draw    (520,220) -- (520,240) ;
\draw [color={rgb, 255:red, 214; green, 39; blue, 40 }  ,draw opacity=1 ]   (560,220) -- (560,240) ;
\draw [color={rgb, 255:red, 214; green, 39; blue, 40 }  ,draw opacity=1 ]   (440,220) -- (440,240) ;
\draw  [color={rgb, 255:red, 214; green, 39; blue, 40 }  ,draw opacity=1 ][fill={rgb, 255:red, 255; green, 255; blue, 255 }  ,fill opacity=1 ] (95,170) .. controls (95,164.48) and (99.48,160) .. (105,160) .. controls (110.52,160) and (115,164.48) .. (115,170) .. controls (115,175.52) and (110.52,180) .. (105,180) .. controls (99.48,180) and (95,175.52) .. (95,170) -- cycle ;
\draw  [color={rgb, 255:red, 214; green, 39; blue, 40 }  ,draw opacity=1 ][fill={rgb, 255:red, 255; green, 255; blue, 255 }  ,fill opacity=1 ] (135,170) .. controls (135,164.48) and (139.48,160) .. (145,160) .. controls (150.52,160) and (155,164.48) .. (155,170) .. controls (155,175.52) and (150.52,180) .. (145,180) .. controls (139.48,180) and (135,175.52) .. (135,170) -- cycle ;
\draw  [color={rgb, 255:red, 214; green, 39; blue, 40 }  ,draw opacity=1 ][fill={rgb, 255:red, 255; green, 255; blue, 255 }  ,fill opacity=1 ] (175,170) .. controls (175,164.48) and (179.48,160) .. (185,160) .. controls (190.52,160) and (195,164.48) .. (195,170) .. controls (195,175.52) and (190.52,180) .. (185,180) .. controls (179.48,180) and (175,175.52) .. (175,170) -- cycle ;
\draw [color={rgb, 255:red, 214; green, 39; blue, 40 }  ,draw opacity=1 ]   (115,170) -- (135,170) ;
\draw [color={rgb, 255:red, 214; green, 39; blue, 40 }  ,draw opacity=1 ]   (155,170) -- (175,170) ;
\draw [color={rgb, 255:red, 214; green, 39; blue, 40 }  ,draw opacity=1 ]   (195,170) -- (215,170) ;
\draw  [color={rgb, 255:red, 0; green, 0; blue, 0 }  ,draw opacity=1 ][fill={rgb, 255:red, 255; green, 255; blue, 255 }  ,fill opacity=1 ] (135,130) .. controls (135,124.48) and (139.48,120) .. (145,120) .. controls (150.52,120) and (155,124.48) .. (155,130) .. controls (155,135.52) and (150.52,140) .. (145,140) .. controls (139.48,140) and (135,135.52) .. (135,130) -- cycle ;
\draw    (155,130) -- (175,130) ;
\draw   (94.8,183.47) .. controls (94.8,188.14) and (97.13,190.47) .. (101.8,190.46) -- (215.07,190.38) .. controls (221.74,190.38) and (225.07,192.71) .. (225.08,197.38) .. controls (225.07,192.71) and (228.4,190.38) .. (235.07,190.37)(232.07,190.37) -- (348.34,190.29) .. controls (353.01,190.28) and (355.34,187.95) .. (355.33,183.28) ;

\draw (499,291) node    {$K+2$};
\draw (396.55,211) node [anchor=west] [inner sep=0.75pt]    {$B$};
\draw (93,134.6) node [anchor=west] [inner sep=0.75pt]  [font=\large] [align=left] {VN:};
\draw (396.6,134.6) node [anchor=west] [inner sep=0.75pt]  [font=\large] [align=left] {PN:};
\draw (226,211) node    {$2\cdot B+1$};
\draw (93,254.4) node [anchor=north west][inner sep=0.75pt]  [font=\large]  {$A\ =\ [\textcolor[rgb]{0.29,0.56,0.89}{2} ,\ \textcolor[rgb]{0.17,0.63,0.17}{3}]$};
\draw (166,121) node    {$0$};
\draw (206,121) node    {$1$};
\draw (246,121) node    {$0$};
\draw (286,121) node    {$1$};
\draw (326,121) node    {$1$};
\draw (166,161) node    {$1$};
\draw (206,161) node    {$1$};
\draw (246,161) node    {$1$};
\draw (286,161) node    {$1$};
\draw (326,161) node    {$1$};
\draw (126,161) node    {$1$};
\draw (373,151) node [anchor=west] [inner sep=0.75pt]    {$0$};

\end{tikzpicture}

    \caption{%
    Reducing BPP to line on uniform tree wVNE $A = [2, 3], B = 3, K = 2$.
    The placement of the red section is forced and the remaining $K$ subtrees of the PN form $K$ bins.}
    \label{fig:octopus}
    
\end{figure}

\begin{theorem}\label{th:WLUTEP_NPC}
wVNE of linear VN on uniform tree PN is NP-complete.
\end{theorem}
\begin{proof}
    We reduce BPP to our problem.
    Given some instance of BPP with $n$ elements and $K$ bins of size $B$, we construct VN and PN similar to the proof of Theorem~\ref{th:WLLEP_NPC}.

    We start with constructing VN.
    For each element $a_i$, a linear graph $s_i$ with $a_i$ vertices is constructed. These linear graphs are called \emph{``element sections''}.
    Additionally, we create $B \cdot K - \sum a_i$ ``singleton sections'' that contain only one vertex each.
    We also create a single ``long section''~--- a linear graph of $2 \cdot B + 1$ vertices that is crucial for this proof.
    All edges in the described sections have weight $1$ and sections are connected sequentially with zero-weighted edges to form linear VN $S$.
    $|S| = B \cdot (K + 2) + 1$ by construction.
    
    Now, we construct PN.
    It consists of a root and $K + 2$ subtrees.
    Each subtree forms a line with $B$ vertices (i.e., the root is the only branching node).
    Such a graph can be imagined as an ``octopus'' with $K + 2$ legs of lengths $B$ (see Figure~\ref{fig:octopus}).
    All edges have weight $1$ since PN is uniform.
    
    We show that BPP has a solution if and only if there is an embedding of cost $\theta = \sum (a_i - 1) + 2 \cdot B$, i.e., the number of $1$-edges in VN.
    Please, note that $\theta$ is the minimum possible cost.
    
    So, the goal is to find an embedding where each $1$-weighted edge in VN is transformed into a path of one edge in PN.
    From that, the placement of $2 \cdot B + 1$ section is forced since there is only one way to place the ``long section'' excluding the symmetry:
    it must take two legs of the octopus and the root. 
    This section is highlighted with red on Figure~\ref{fig:octopus}.
    That leaves the remaining ``element sections'' to be placed on $K$ legs.
    Note that an ``element section'' cannot be placed on more than one leg, since the path through the root is at least two and the root itself is already mapped to.
    Additionally, ``singleton sections'' placement is free.
    
    The proof of that construction is similar to the one in the previous Theorem~\ref{th:WLLEP_NPC}.
\end{proof}

\subsection{wVNE of Star VN}\label{sec:star_wVNE}\label{sec:wVNE_star}

The problem of embedding a uniform star VN was explored in detail by authors of~\cite{RostFS15_stars}.
They showed it to be polynomial, even with limited edge capacity of PN (which we further explore in Section~\ref{sec:star_cVNE}).
However, to the best of our knowledge, the complexity of non-uniform star VN was not considered previously.
In this subsection we analyze the algorithms of wVNE for a star VN: the weighted star can be embedded on any weighted graph in polynomial time.

\begin{theorem}\label{th:w_star}
For a star VN $S$ and an arbitrary weighted PN $G$, minimal wVNE can be computed in $O(n^2 + T)$ time, where $T$ is the complexity of finding all-pairs shortest paths in $G$.
\end{theorem}
\begin{proof}
    Consider some mapping where each leaf $i$ of the star is placed at the node $f(i)$ and the central node $n$ is placed at the node $f(n) = j$.
    $d_{i} = d(f(n), f(i)) = d(j, f(i))$~--- the shortest distance from $j$, the image of $n$, to $f(i)$. 
    In that case, we know that the embedding cost is:

    \[Cost(f) = \sum_{i = 1}^{n-1} w_{i} \cdot d_{i}. \]

    In other words, the demand from each non-central node $w_i$ is multiplied by $d_{i}$ and added to the sum exactly once. 
    By the rearrangement inequality we know that this sum is minimized, when elements are greedily arranged such that the ``heavier'' nodes are closer to the center $j$ and ``lighter'' nodes are further:

    \[w_1 \leq w_2 \leq \ldots \leq w_{n-1}\]
    \[d_{n-1} \leq d_{n-2} \leq \ldots \leq d_1\]

    Now, we calculate $Cost(f)$ for every placement of the central node $j = f(n)$ and choose minimum.
    Overall, the total complexity of this approach is $O(n^2 + T)$, where $T$ is the complexity of the all-pairs shortest paths problem on a graph. 
\end{proof}

If we use Floyd–Warshall~\cite{Floyd62_path} algorithm with $T = n^3$, the complexity becomes $O(n^3)$.
However, there exist algorithms that provide better complexity on different types of graphs, e.g., if the number of edges in $G$ is $O(n)$, the algorithm by Thorup~\cite{Thorup99_path} can be used, resulting in $O(n^2)$ time complexity.

The results from Theorem~\ref{th:w_star} can be further improved with either VN or PN being uniform.

\subsection{wVNE of $2$-star VN}\label{sec:2star_wVNE}\label{sec:wVNE_2star}

Since we can solve an embedding problem for a weighted star VN, there is a reason to argue whether or not an embedding of a weighted $2$-star VN is polynomial.
In this section, we answer this negatively and prove that this problem is NP-hard.

\begin{theorem}\label{th:W2SEP_NPC}
    $2$-star VN on $2$-star PN wVNE is NP-complete.
\end{theorem}
\begin{proof}
    To prove NP-hardness we reduce BPP to this problem.
    The proof is very similar to the proof of Theorem~\ref{th:WLLEP_NPC}, so, here we present only the general idea.

    For each item $a_i$ from BPP we create a subtree of the root with $a_i$ vertices in the VN (see Figure~\ref{fig:BPP_to_W2S}): a root and $a_i - 1$ children.
    Additionally, we create $B \cdot K - \sum a_i$ leaves connected to the root.
    All edges in each such subtree have $1$-weight while all edges from the root to subtrees have $0$-weight.
    
    Similarly, for PN, we have a root with subtrees of size $B$ with $0$-edges for bins.
    In PN, the weights of edges from the root are one and the weights of all other edges are zero.
    The desired cost of wVNE is $C = 0$.

    The rest of the proof follows the proof of the Theorem~\ref{th:WLLEP_NPC}.
    If there is an arrangement of items into buckets we can easily place $1$-weighted subtrees of VN into dedicated $0$-weighted subtrees of PN.
    All adjacent edges of other vertices have demand $0$, so these vertices can be placed anywhere.
    On the other hand, the placement of a weighted subtree of VN on multiple $0$-weighted subtrees of PN is impossible since in PN any path connecting two subtrees has non-zero weight.
\end{proof}

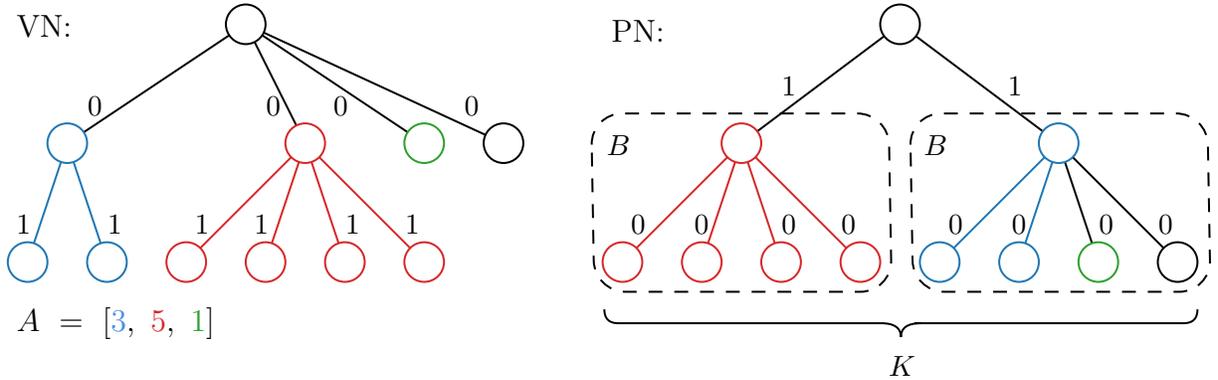
\begin{figure}[!ht]
    \centering

\tikzset{every picture/.style={line width=0.75pt}} 

\begin{tikzpicture}[x=0.75pt,y=0.75pt,yscale=-1,xscale=1]

\draw   (300.87,153.87) .. controls (300.87,158.54) and (303.2,160.87) .. (307.87,160.87) -- (440.43,160.87) .. controls (447.1,160.87) and (450.43,163.2) .. (450.43,167.87) .. controls (450.43,163.2) and (453.76,160.87) .. (460.43,160.87)(457.43,160.87) -- (593,160.87) .. controls (597.67,160.87) and (600,158.54) .. (600,153.87) ;
\draw    (120,10) -- (30,70) ;
\draw    (120,10) -- (150,70) ;
\draw    (120,10) -- (210,70) ;
\draw    (120,10) -- (250,70) ;
\draw [color={rgb, 255:red, 214; green, 39; blue, 40 }  ,draw opacity=1 ]   (150,70) -- (210,130) ;
\draw [color={rgb, 255:red, 214; green, 39; blue, 40 }  ,draw opacity=1 ]   (150,70) -- (170,130) ;
\draw [color={rgb, 255:red, 214; green, 39; blue, 40 }  ,draw opacity=1 ]   (150,70) -- (130,130) ;
\draw [color={rgb, 255:red, 214; green, 39; blue, 40 }  ,draw opacity=1 ]   (150,70) -- (90,130) ;
\draw [color={rgb, 255:red, 31; green, 119; blue, 180 }  ,draw opacity=1 ]   (30,70) -- (50,130) ;
\draw [color={rgb, 255:red, 31; green, 119; blue, 180 }  ,draw opacity=1 ]   (30,70) -- (10,130) ;
\draw  [color={rgb, 255:red, 0; green, 0; blue, 0 }  ,draw opacity=1 ][fill={rgb, 255:red, 255; green, 255; blue, 255 }  ,fill opacity=1 ] (110,10) .. controls (110,4.48) and (114.48,0) .. (120,0) .. controls (125.52,0) and (130,4.48) .. (130,10) .. controls (130,15.52) and (125.52,20) .. (120,20) .. controls (114.48,20) and (110,15.52) .. (110,10) -- cycle ;
\draw  [color={rgb, 255:red, 31; green, 119; blue, 180 }  ,draw opacity=1 ][fill={rgb, 255:red, 255; green, 255; blue, 255 }  ,fill opacity=1 ] (20,70) .. controls (20,64.48) and (24.48,60) .. (30,60) .. controls (35.52,60) and (40,64.48) .. (40,70) .. controls (40,75.52) and (35.52,80) .. (30,80) .. controls (24.48,80) and (20,75.52) .. (20,70) -- cycle ;
\draw  [color={rgb, 255:red, 31; green, 119; blue, 180 }  ,draw opacity=1 ][fill={rgb, 255:red, 255; green, 255; blue, 255 }  ,fill opacity=1 ] (40,130) .. controls (40,124.48) and (44.48,120) .. (50,120) .. controls (55.52,120) and (60,124.48) .. (60,130) .. controls (60,135.52) and (55.52,140) .. (50,140) .. controls (44.48,140) and (40,135.52) .. (40,130) -- cycle ;
\draw  [color={rgb, 255:red, 31; green, 119; blue, 180 }  ,draw opacity=1 ][fill={rgb, 255:red, 255; green, 255; blue, 255 }  ,fill opacity=1 ] (0,130) .. controls (0,124.48) and (4.48,120) .. (10,120) .. controls (15.52,120) and (20,124.48) .. (20,130) .. controls (20,135.52) and (15.52,140) .. (10,140) .. controls (4.48,140) and (0,135.52) .. (0,130) -- cycle ;
\draw  [color={rgb, 255:red, 214; green, 39; blue, 40 }  ,draw opacity=1 ][fill={rgb, 255:red, 255; green, 255; blue, 255 }  ,fill opacity=1 ] (80,130) .. controls (80,124.48) and (84.48,120) .. (90,120) .. controls (95.52,120) and (100,124.48) .. (100,130) .. controls (100,135.52) and (95.52,140) .. (90,140) .. controls (84.48,140) and (80,135.52) .. (80,130) -- cycle ;
\draw  [color={rgb, 255:red, 214; green, 39; blue, 40 }  ,draw opacity=1 ][fill={rgb, 255:red, 255; green, 255; blue, 255 }  ,fill opacity=1 ] (120,130) .. controls (120,124.48) and (124.48,120) .. (130,120) .. controls (135.52,120) and (140,124.48) .. (140,130) .. controls (140,135.52) and (135.52,140) .. (130,140) .. controls (124.48,140) and (120,135.52) .. (120,130) -- cycle ;
\draw  [color={rgb, 255:red, 214; green, 39; blue, 40 }  ,draw opacity=1 ][fill={rgb, 255:red, 255; green, 255; blue, 255 }  ,fill opacity=1 ] (160,130) .. controls (160,124.48) and (164.48,120) .. (170,120) .. controls (175.52,120) and (180,124.48) .. (180,130) .. controls (180,135.52) and (175.52,140) .. (170,140) .. controls (164.48,140) and (160,135.52) .. (160,130) -- cycle ;
\draw  [color={rgb, 255:red, 214; green, 39; blue, 40 }  ,draw opacity=1 ][fill={rgb, 255:red, 255; green, 255; blue, 255 }  ,fill opacity=1 ] (200,130) .. controls (200,124.48) and (204.48,120) .. (210,120) .. controls (215.52,120) and (220,124.48) .. (220,130) .. controls (220,135.52) and (215.52,140) .. (210,140) .. controls (204.48,140) and (200,135.52) .. (200,130) -- cycle ;
\draw  [color={rgb, 255:red, 214; green, 39; blue, 40 }  ,draw opacity=1 ][fill={rgb, 255:red, 255; green, 255; blue, 255 }  ,fill opacity=1 ] (140,70) .. controls (140,64.48) and (144.48,60) .. (150,60) .. controls (155.52,60) and (160,64.48) .. (160,70) .. controls (160,75.52) and (155.52,80) .. (150,80) .. controls (144.48,80) and (140,75.52) .. (140,70) -- cycle ;
\draw  [color={rgb, 255:red, 44; green, 160; blue, 44 }  ,draw opacity=1 ][fill={rgb, 255:red, 255; green, 255; blue, 255 }  ,fill opacity=1 ] (200,70) .. controls (200,64.48) and (204.48,60) .. (210,60) .. controls (215.52,60) and (220,64.48) .. (220,70) .. controls (220,75.52) and (215.52,80) .. (210,80) .. controls (204.48,80) and (200,75.52) .. (200,70) -- cycle ;
\draw  [color={rgb, 255:red, 0; green, 0; blue, 0 }  ,draw opacity=1 ][fill={rgb, 255:red, 255; green, 255; blue, 255 }  ,fill opacity=1 ] (240,70) .. controls (240,64.48) and (244.48,60) .. (250,60) .. controls (255.52,60) and (260,64.48) .. (260,70) .. controls (260,75.52) and (255.52,80) .. (250,80) .. controls (244.48,80) and (240,75.52) .. (240,70) -- cycle ;
\draw    (370,70) -- (450,10) ;
\draw    (530,70) -- (450,10) ;
\draw    (550,130) -- (530,70) ;
\draw    (590,130) -- (530,70) ;
\draw [color={rgb, 255:red, 214; green, 39; blue, 40 }  ,draw opacity=1 ]   (310,130) -- (370,70) ;
\draw [color={rgb, 255:red, 214; green, 39; blue, 40 }  ,draw opacity=1 ]   (370,70) -- (430,130) ;
\draw [color={rgb, 255:red, 214; green, 39; blue, 40 }  ,draw opacity=1 ]   (370,70) -- (390,130) ;
\draw [color={rgb, 255:red, 214; green, 39; blue, 40 }  ,draw opacity=1 ]   (370,70) -- (350,130) ;
\draw [color={rgb, 255:red, 31; green, 119; blue, 180 }  ,draw opacity=1 ]   (470,130) -- (530,70) ;
\draw [color={rgb, 255:red, 31; green, 119; blue, 180 }  ,draw opacity=1 ]   (510,130) -- (510.01,129.96) -- (530,70) ;
\draw  [color={rgb, 255:red, 214; green, 39; blue, 40 }  ,draw opacity=1 ][fill={rgb, 255:red, 255; green, 255; blue, 255 }  ,fill opacity=1 ] (300,130) .. controls (300,124.48) and (304.48,120) .. (310,120) .. controls (315.52,120) and (320,124.48) .. (320,130) .. controls (320,135.52) and (315.52,140) .. (310,140) .. controls (304.48,140) and (300,135.52) .. (300,130) -- cycle ;
\draw  [color={rgb, 255:red, 214; green, 39; blue, 40 }  ,draw opacity=1 ][fill={rgb, 255:red, 255; green, 255; blue, 255 }  ,fill opacity=1 ] (340,130) .. controls (340,124.48) and (344.48,120) .. (350,120) .. controls (355.52,120) and (360,124.48) .. (360,130) .. controls (360,135.52) and (355.52,140) .. (350,140) .. controls (344.48,140) and (340,135.52) .. (340,130) -- cycle ;
\draw  [color={rgb, 255:red, 214; green, 39; blue, 40 }  ,draw opacity=1 ][fill={rgb, 255:red, 255; green, 255; blue, 255 }  ,fill opacity=1 ] (380,130) .. controls (380,124.48) and (384.48,120) .. (390,120) .. controls (395.52,120) and (400,124.48) .. (400,130) .. controls (400,135.52) and (395.52,140) .. (390,140) .. controls (384.48,140) and (380,135.52) .. (380,130) -- cycle ;
\draw  [color={rgb, 255:red, 214; green, 39; blue, 40 }  ,draw opacity=1 ][fill={rgb, 255:red, 255; green, 255; blue, 255 }  ,fill opacity=1 ] (420,130) .. controls (420,124.48) and (424.48,120) .. (430,120) .. controls (435.52,120) and (440,124.48) .. (440,130) .. controls (440,135.52) and (435.52,140) .. (430,140) .. controls (424.48,140) and (420,135.52) .. (420,130) -- cycle ;
\draw  [color={rgb, 255:red, 214; green, 39; blue, 40 }  ,draw opacity=1 ][fill={rgb, 255:red, 255; green, 255; blue, 255 }  ,fill opacity=1 ] (360,70) .. controls (360,64.48) and (364.48,60) .. (370,60) .. controls (375.52,60) and (380,64.48) .. (380,70) .. controls (380,75.52) and (375.52,80) .. (370,80) .. controls (364.48,80) and (360,75.52) .. (360,70) -- cycle ;
\draw  [color={rgb, 255:red, 31; green, 119; blue, 180 }  ,draw opacity=1 ][fill={rgb, 255:red, 255; green, 255; blue, 255 }  ,fill opacity=1 ] (460,130) .. controls (460,124.48) and (464.48,120) .. (470,120) .. controls (475.52,120) and (480,124.48) .. (480,130) .. controls (480,135.52) and (475.52,140) .. (470,140) .. controls (464.48,140) and (460,135.52) .. (460,130) -- cycle ;
\draw  [color={rgb, 255:red, 31; green, 119; blue, 180 }  ,draw opacity=1 ][fill={rgb, 255:red, 255; green, 255; blue, 255 }  ,fill opacity=1 ] (500,130) .. controls (500,124.48) and (504.48,120) .. (510,120) .. controls (515.52,120) and (520,124.48) .. (520,130) .. controls (520,135.52) and (515.52,140) .. (510,140) .. controls (504.48,140) and (500,135.52) .. (500,130) -- cycle ;
\draw  [color={rgb, 255:red, 44; green, 160; blue, 44 }  ,draw opacity=1 ][fill={rgb, 255:red, 255; green, 255; blue, 255 }  ,fill opacity=1 ] (540,130) .. controls (540,124.48) and (544.48,120) .. (550,120) .. controls (555.52,120) and (560,124.48) .. (560,130) .. controls (560,135.52) and (555.52,140) .. (550,140) .. controls (544.48,140) and (540,135.52) .. (540,130) -- cycle ;
\draw  [color={rgb, 255:red, 0; green, 0; blue, 0 }  ,draw opacity=1 ][fill={rgb, 255:red, 255; green, 255; blue, 255 }  ,fill opacity=1 ] (580,130) .. controls (580,124.48) and (584.48,120) .. (590,120) .. controls (595.52,120) and (600,124.48) .. (600,130) .. controls (600,135.52) and (595.52,140) .. (590,140) .. controls (584.48,140) and (580,135.52) .. (580,130) -- cycle ;
\draw  [color={rgb, 255:red, 31; green, 119; blue, 180 }  ,draw opacity=1 ][fill={rgb, 255:red, 255; green, 255; blue, 255 }  ,fill opacity=1 ] (520,70) .. controls (520,64.48) and (524.48,60) .. (530,60) .. controls (535.52,60) and (540,64.48) .. (540,70) .. controls (540,75.52) and (535.52,80) .. (530,80) .. controls (524.48,80) and (520,75.52) .. (520,70) -- cycle ;
\draw  [color={rgb, 255:red, 0; green, 0; blue, 0 }  ,draw opacity=1 ][fill={rgb, 255:red, 255; green, 255; blue, 255 }  ,fill opacity=1 ] (440,10) .. controls (440,4.48) and (444.48,0) .. (450,0) .. controls (455.52,0) and (460,4.48) .. (460,10) .. controls (460,15.52) and (455.52,20) .. (450,20) .. controls (444.48,20) and (440,15.52) .. (440,10) -- cycle ;
\draw  [dash pattern={on 5.25pt off 4.5pt}] (294.49,70.75) .. controls (294.34,62.05) and (301.27,55) .. (309.96,55) -- (428.46,55) .. controls (437.16,55) and (444.34,62.05) .. (444.49,70.75) -- (445.51,129.25) .. controls (445.66,137.95) and (438.73,145) .. (430.04,145) -- (311.54,145) .. controls (302.84,145) and (295.66,137.95) .. (295.51,129.25) -- cycle ;
\draw  [dash pattern={on 5.25pt off 4.5pt}] (455.27,70.75) .. controls (455.12,62.05) and (462.05,55) .. (470.75,55) -- (589.25,55) .. controls (597.95,55) and (605.12,62.05) .. (605.27,70.75) -- (606.3,129.25) .. controls (606.45,137.95) and (599.52,145) .. (590.82,145) -- (472.32,145) .. controls (463.62,145) and (456.45,137.95) .. (456.3,129.25) -- cycle ;

\draw (3,161) node [anchor=west] [inner sep=0.75pt]  [font=\large]  {$A\ =\ [\textcolor[rgb]{0.29,0.56,0.89}{3} ,\ \textcolor[rgb]{0.84,0.15,0.16}{5} ,\ \textcolor[rgb]{0.17,0.63,0.17}{1}]$};
\draw (451,188.8) node [anchor=south] [inner sep=0.75pt]    {$K$};
\draw (3,11) node [anchor=west] [inner sep=0.75pt]  [font=\large] [align=left] {VN:};
\draw (303,13.5) node [anchor=west] [inner sep=0.75pt]  [font=\large] [align=left] {PN:};
\draw (44,51) node  [font=\normalsize]  {$0$};
\draw (134,51) node  [font=\normalsize]  {$0$};
\draw (168,51) node  [font=\normalsize]  {$0$};
\draw (234,51) node  [font=\normalsize]  {$0$};
\draw (8,111) node  [font=\normalsize]  {$1$};
\draw (54,111) node  [font=\normalsize]  {$1$};
\draw (98,111) node  [font=\normalsize]  {$1$};
\draw (204,111) node  [font=\normalsize]  {$1$};
\draw (128,111) node  [font=\normalsize]  {$1$};
\draw (174,111) node  [font=\normalsize]  {$1$};
\draw (308,71) node    {$B$};
\draw (468,71) node    {$B$};
\draw (318,111) node  [font=\normalsize]  {$0$};
\draw (350,111) node  [font=\normalsize]  {$0$};
\draw (394,111) node  [font=\normalsize]  {$0$};
\draw (424,111) node  [font=\normalsize]  {$0$};
\draw (478,111) node  [font=\normalsize]  {$0$};
\draw (510,111) node  [font=\normalsize]  {$0$};
\draw (554,111) node  [font=\normalsize]  {$0$};
\draw (584,111) node  [font=\normalsize]  {$0$};
\draw (394,41) node  [font=\normalsize]  {$1$};
\draw (508,41) node  [font=\normalsize]  {$1$};

\end{tikzpicture}

    \caption{Reducing BPP to $2$-star on $2$-star wVNE for $A = [3, 5, 1], B = 5, K = 2$.}
    \label{fig:BPP_to_W2S}
    
\end{figure}

The major purpose of Theorem~\ref{th:W2SEP_NPC} is to demonstrate how the bins and elements from BPP can be implemented with $2$-star graphs.
We can apply the ideas of Theorems~\ref{th:WLLEP_NPC} and~\ref{th:W2SEP_NPC} to other cases.
For example, with the similar BPP reductions, the following observations are true:
\begin{itemize}
    \item wVNE of $2$-star VN on line PN is NP-complete.
        
    Items are ``line sections'' from Theorem~\ref{th:WLLEP_NPC} and bins are from Theorem~\ref{th:W2SEP_NPC}.
    Note that we need an additional ``singleton section'' in VN to account for the root of PN.
        
    \item wVNE of line VN on $2$-star PN is NP-complete.
    
    This time, the items are subtrees from Theorem~\ref{th:W2SEP_NPC} and bins are ``line sections'' from Theorem~\ref{th:WLLEP_NPC}.
    The two above observations are true since the described line and $2$-star topologies have similar structures.
    They allow to create subgraphs of given size, such that every path between each one of them has cost zero and each path in-between is not.
        
    \item wVNE of weighted tree/line VN on a weighted line/tree is NP-complete since $2$-star is a tree.

    \item All the above NP-complete results for graphs with linear topology are true for the graphs with cycle topology since we can close a line with a $0$-weight edge.
\end{itemize}

Interestingly enough, the results of Theorem~\ref{th:W2SEP_NPC} are true for the uniform PN.
Namely, we show that in wVNE embedding of a $2$-star on \textbf{uniform} $2$-star is NP-complete.

\begin{theorem}\label{th:WU_2SEP_NPC}
    wVNE of $2$-star VN on uniform $2$-star PN is NP-complete.
\end{theorem}
\begin{proof}
    For this theorem we modify the proof of Theorem~\ref{th:W2SEP_NPC}.
    We reduce BPP to our problem.
    
    Consider some instance of BPP with $n \geq 2$ items and $K \geq 2$ bins of size $B$.
    Without loss of generality, assume $\sum a_i = B \cdot K$ (we can always add $1$-weight items to make it true).
    That is required so that sizes of VN and PN in this reduction will be equal.
    
    For described BPP instance we construct VN and PN by the following rules (Figure~\ref{fig:BPP_to_UW2S}):
    \begin{itemize}
        \item \textbf{Virtual Network:}
        For every item $a_i$ we add an ``element'' subtree $A_i$: a star with $a_i$ vertices.
        Edges inside $A_i$ have weights $1$ and the edge from $A_i$ to the root has weights $0$.
        
        Additionally, $(K + B)$ ``dummy'' leaves are connected to the root.
        The edge from a ``dummy'' leaf to the root is $x = 2 \cdot B \cdot K$.
        Overall, there are $1 + K + B + \sum a_i$ nodes in the VN.
        
        \item \textbf{Physical Network:}
        For each of $K$ bins we add a ``bin'' subtree: a star with $B+1$ vertices.
        We connect $B$ extra leaves to the root.
        
        All edge weights in PN are $1$, since it is uniform.
    \end{itemize}
    
    Further, we show that we can distribute $n$ items in $K$ bins iff there is an embedding of cost not greater than $\theta = (K + B) \cdot x + 2 \cdot \sum (a_i - 1)$.
    
    $\Rightarrow$: If there is a way to partition items in $K$ bins, then we can place all vertices from ``element'' subtrees $A_i$ onto the leaves of ``bin'' subtrees in PN.
    That will cost $2 \cdot \sum (a_i - 1)$ since the distance between in one ``bin'' two leaves is $2$ and there are $\sum (a_i - 1)$ edges in ``element'' subtrees.
    The root of VN can be placed in the root $r$ of PN and $(K + B)$ ``dummy'' leaves are placed in $(K + B)$ children of $r$, costing $(K + B) \cdot x$.
    This embedding is marked with red in Figure~\ref{fig:BPP_to_UW2S}.
    
    $\Leftarrow$: Now, we assume that we have some embedding of the VN onto the PN of cost not greater than $\theta$.
    Note that every ``dummy'' leaf has to be placed in exactly distance $1$ from the VN root placement.
    Otherwise, even if one ``dummy'' leaf is placed two edges away, the cost will be too large:
    
    $2 \cdot x + (K + B - 1) \cdot x = (K + B) \cdot x + x = (K + B) \cdot x + (2 \cdot B \cdot K) > (K + B) \cdot x + 2 \cdot \sum (a_i - 1) = \theta$
    
    The only vertex in PN that can satisfy this condition is the root, since other vertices have degree less than $K + B$, assuming $K \geq 2$.
    Thus, every correct embedding places VN root in PN root $r$ and ``dummy'' vertices in the children of $r$.
    Now, vertices from ``element'' subtrees can only be placed in the leaves of ``bin'' subtrees.
    The distance between two leaves is always at least $2$ and the distance between two leaves from different ``bin'' subtrees is always $4$.
    From that, vertices of the same ``element'' subtree are always placed in the same ``bin'' subtree in order to not exceed the cost $\theta$.
\end{proof}

Concerning oversubscribed $2$-star VN, we discuss it more in Sections~\ref{sec:oversub_2satr}~and~\ref{sec:oversub_on_tree}.
In Section~\ref{sec:oversub_2satr} we prove that cVNE of oversubscribed $2$-star VN is NP-complete.
In Section~\ref{sec:oversub_on_tree} we show that the VNE of this topology is actually in P for tree PN.
Regardless, in Theorem~\ref{th:oversub_2Star_wVNE_NP} we state that wVNE of oversubscribed $2$-star VN on an arbitrary uniform SN is NP-complete.

\begin{theorem}\label{th:oversub_2Star_wVNE_NP}
Oversubscribed $2$-star VN on the uniform PN wVNE is NP-complete.
\end{theorem}

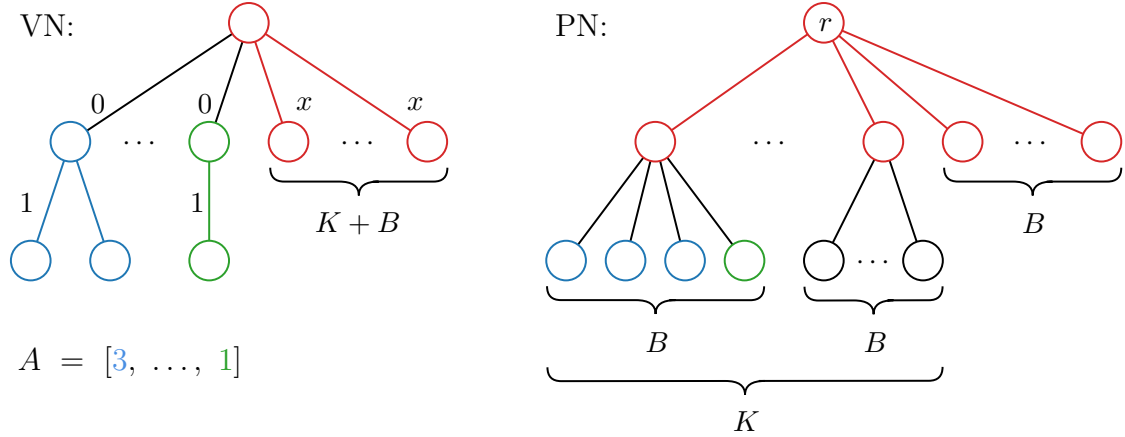
\begin{figure}[!ht]
    \centering

\tikzset{every picture/.style={line width=0.75pt}} 

\begin{tikzpicture}[x=0.75pt,y=0.75pt,yscale=-1,xscale=1]

\draw    (60,120) -- (150,60) ;
\draw    (150,60) -- (130,120) ;
\draw [color={rgb, 255:red, 214; green, 39; blue, 40 }  ,draw opacity=1 ]   (150,60) -- (170,120) ;
\draw [color={rgb, 255:red, 214; green, 39; blue, 40 }  ,draw opacity=1 ]   (150,60) -- (240,120) ;
\draw   (160.8,133.13) .. controls (160.8,137.8) and (163.13,140.13) .. (167.8,140.13) -- (195.47,140.13) .. controls (202.14,140.13) and (205.47,142.46) .. (205.47,147.13) .. controls (205.47,142.46) and (208.8,140.13) .. (215.47,140.13)(212.47,140.13) -- (243.13,140.13) .. controls (247.8,140.13) and (250.13,137.8) .. (250.13,133.13) ;
\draw [color={rgb, 255:red, 31; green, 119; blue, 180 }  ,draw opacity=1 ]   (60,120) -- (40,180) ;
\draw [color={rgb, 255:red, 31; green, 119; blue, 180 }  ,draw opacity=1 ]   (60,120) -- (80,180) ;
\draw [color={rgb, 255:red, 44; green, 160; blue, 44 }  ,draw opacity=1 ]   (130,120) -- (130,180) ;
\draw   (500.67,133.13) .. controls (500.67,137.8) and (503,140.13) .. (507.67,140.13) -- (535.33,140.13) .. controls (542,140.13) and (545.33,142.46) .. (545.33,147.13) .. controls (545.33,142.46) and (548.66,140.13) .. (555.33,140.13)(552.33,140.13) -- (583,140.13) .. controls (587.67,140.13) and (590,137.8) .. (590,133.13) ;
\draw [color={rgb, 255:red, 214; green, 39; blue, 40 }  ,draw opacity=1 ]   (440,60) -- (355,120) ;
\draw [color={rgb, 255:red, 214; green, 39; blue, 40 }  ,draw opacity=1 ]   (440,60) -- (470,120) ;
\draw [color={rgb, 255:red, 214; green, 39; blue, 40 }  ,draw opacity=1 ]   (440,60) -- (510,120) ;
\draw [color={rgb, 255:red, 214; green, 39; blue, 40 }  ,draw opacity=1 ]   (440,60) -- (580,120) ;
\draw    (470,120) -- (490,180) ;
\draw    (470,120) -- (440,180) ;
\draw    (355,120) -- (310,180) ;
\draw    (355,120) -- (340,180) ;
\draw    (355,120) -- (370,180) ;
\draw    (355,120) -- (400,180) ;
\draw   (300.4,192.94) .. controls (300.4,197.61) and (302.73,199.94) .. (307.4,199.95) -- (345.13,199.97) .. controls (351.8,199.97) and (355.13,202.3) .. (355.13,206.97) .. controls (355.13,202.3) and (358.46,199.97) .. (365.13,199.98)(362.13,199.98) -- (402.86,200) .. controls (407.53,200) and (409.86,197.67) .. (409.87,193) ;
\draw   (430.25,193.19) .. controls (430.26,197.86) and (432.6,200.18) .. (437.27,200.17) -- (455.08,200.12) .. controls (461.75,200.11) and (465.09,202.43) .. (465.1,207.1) .. controls (465.09,202.43) and (468.41,200.09) .. (475.08,200.07)(472.08,200.08) -- (492.89,200.02) .. controls (497.56,200.01) and (499.88,197.67) .. (499.87,193) ;
\draw   (300.6,233.87) .. controls (300.61,238.54) and (302.94,240.87) .. (307.61,240.86) -- (390.24,240.77) .. controls (396.91,240.76) and (400.24,243.09) .. (400.25,247.76) .. controls (400.24,243.09) and (403.57,240.76) .. (410.24,240.75)(407.24,240.75) -- (492.87,240.67) .. controls (497.54,240.66) and (499.87,238.33) .. (499.87,233.66) ;
\draw  [color={rgb, 255:red, 214; green, 39; blue, 40 }  ,draw opacity=1 ][fill={rgb, 255:red, 255; green, 255; blue, 255 }  ,fill opacity=1 ] (140,60) .. controls (140,54.48) and (144.48,50) .. (150,50) .. controls (155.52,50) and (160,54.48) .. (160,60) .. controls (160,65.52) and (155.52,70) .. (150,70) .. controls (144.48,70) and (140,65.52) .. (140,60) -- cycle ;
\draw  [color={rgb, 255:red, 44; green, 160; blue, 44 }  ,draw opacity=1 ][fill={rgb, 255:red, 255; green, 255; blue, 255 }  ,fill opacity=1 ] (120,120) .. controls (120,114.48) and (124.48,110) .. (130,110) .. controls (135.52,110) and (140,114.48) .. (140,120) .. controls (140,125.52) and (135.52,130) .. (130,130) .. controls (124.48,130) and (120,125.52) .. (120,120) -- cycle ;
\draw  [color={rgb, 255:red, 31; green, 119; blue, 180 }  ,draw opacity=1 ][fill={rgb, 255:red, 255; green, 255; blue, 255 }  ,fill opacity=1 ] (50,120) .. controls (50,114.48) and (54.48,110) .. (60,110) .. controls (65.52,110) and (70,114.48) .. (70,120) .. controls (70,125.52) and (65.52,130) .. (60,130) .. controls (54.48,130) and (50,125.52) .. (50,120) -- cycle ;
\draw  [color={rgb, 255:red, 214; green, 39; blue, 40 }  ,draw opacity=1 ][fill={rgb, 255:red, 255; green, 255; blue, 255 }  ,fill opacity=1 ] (230,120) .. controls (230,114.48) and (234.48,110) .. (240,110) .. controls (245.52,110) and (250,114.48) .. (250,120) .. controls (250,125.52) and (245.52,130) .. (240,130) .. controls (234.48,130) and (230,125.52) .. (230,120) -- cycle ;
\draw  [color={rgb, 255:red, 214; green, 39; blue, 40 }  ,draw opacity=1 ][fill={rgb, 255:red, 255; green, 255; blue, 255 }  ,fill opacity=1 ] (160,120) .. controls (160,114.48) and (164.48,110) .. (170,110) .. controls (175.52,110) and (180,114.48) .. (180,120) .. controls (180,125.52) and (175.52,130) .. (170,130) .. controls (164.48,130) and (160,125.52) .. (160,120) -- cycle ;
\draw  [color={rgb, 255:red, 31; green, 119; blue, 180 }  ,draw opacity=1 ][fill={rgb, 255:red, 255; green, 255; blue, 255 }  ,fill opacity=1 ] (70,180) .. controls (70,174.48) and (74.48,170) .. (80,170) .. controls (85.52,170) and (90,174.48) .. (90,180) .. controls (90,185.52) and (85.52,190) .. (80,190) .. controls (74.48,190) and (70,185.52) .. (70,180) -- cycle ;
\draw  [color={rgb, 255:red, 31; green, 119; blue, 180 }  ,draw opacity=1 ][fill={rgb, 255:red, 255; green, 255; blue, 255 }  ,fill opacity=1 ] (30,180) .. controls (30,174.48) and (34.48,170) .. (40,170) .. controls (45.52,170) and (50,174.48) .. (50,180) .. controls (50,185.52) and (45.52,190) .. (40,190) .. controls (34.48,190) and (30,185.52) .. (30,180) -- cycle ;
\draw  [color={rgb, 255:red, 44; green, 160; blue, 44 }  ,draw opacity=1 ][fill={rgb, 255:red, 255; green, 255; blue, 255 }  ,fill opacity=1 ] (120,180) .. controls (120,174.48) and (124.48,170) .. (130,170) .. controls (135.52,170) and (140,174.48) .. (140,180) .. controls (140,185.52) and (135.52,190) .. (130,190) .. controls (124.48,190) and (120,185.52) .. (120,180) -- cycle ;
\draw  [color={rgb, 255:red, 214; green, 39; blue, 40 }  ,draw opacity=1 ][fill={rgb, 255:red, 255; green, 255; blue, 255 }  ,fill opacity=1 ] (460,120) .. controls (460,114.48) and (464.48,110) .. (470,110) .. controls (475.52,110) and (480,114.48) .. (480,120) .. controls (480,125.52) and (475.52,130) .. (470,130) .. controls (464.48,130) and (460,125.52) .. (460,120) -- cycle ;
\draw  [color={rgb, 255:red, 214; green, 39; blue, 40 }  ,draw opacity=1 ][fill={rgb, 255:red, 255; green, 255; blue, 255 }  ,fill opacity=1 ] (345,120) .. controls (345,114.48) and (349.48,110) .. (355,110) .. controls (360.52,110) and (365,114.48) .. (365,120) .. controls (365,125.52) and (360.52,130) .. (355,130) .. controls (349.48,130) and (345,125.52) .. (345,120) -- cycle ;
\draw  [color={rgb, 255:red, 214; green, 39; blue, 40 }  ,draw opacity=1 ][fill={rgb, 255:red, 255; green, 255; blue, 255 }  ,fill opacity=1 ] (430,60) .. controls (430,54.48) and (434.48,50) .. (440,50) .. controls (445.52,50) and (450,54.48) .. (450,60) .. controls (450,65.52) and (445.52,70) .. (440,70) .. controls (434.48,70) and (430,65.52) .. (430,60) -- cycle ;
\draw  [color={rgb, 255:red, 214; green, 39; blue, 40 }  ,draw opacity=1 ][fill={rgb, 255:red, 255; green, 255; blue, 255 }  ,fill opacity=1 ] (570,120) .. controls (570,114.48) and (574.48,110) .. (580,110) .. controls (585.52,110) and (590,114.48) .. (590,120) .. controls (590,125.52) and (585.52,130) .. (580,130) .. controls (574.48,130) and (570,125.52) .. (570,120) -- cycle ;
\draw  [color={rgb, 255:red, 214; green, 39; blue, 40 }  ,draw opacity=1 ][fill={rgb, 255:red, 255; green, 255; blue, 255 }  ,fill opacity=1 ] (500,120) .. controls (500,114.48) and (504.48,110) .. (510,110) .. controls (515.52,110) and (520,114.48) .. (520,120) .. controls (520,125.52) and (515.52,130) .. (510,130) .. controls (504.48,130) and (500,125.52) .. (500,120) -- cycle ;
\draw  [color={rgb, 255:red, 31; green, 119; blue, 180 }  ,draw opacity=1 ][fill={rgb, 255:red, 255; green, 255; blue, 255 }  ,fill opacity=1 ] (330,180) .. controls (330,174.48) and (334.48,170) .. (340,170) .. controls (345.52,170) and (350,174.48) .. (350,180) .. controls (350,185.52) and (345.52,190) .. (340,190) .. controls (334.48,190) and (330,185.52) .. (330,180) -- cycle ;
\draw  [color={rgb, 255:red, 31; green, 119; blue, 180 }  ,draw opacity=1 ][fill={rgb, 255:red, 255; green, 255; blue, 255 }  ,fill opacity=1 ] (300,180) .. controls (300,174.48) and (304.48,170) .. (310,170) .. controls (315.52,170) and (320,174.48) .. (320,180) .. controls (320,185.52) and (315.52,190) .. (310,190) .. controls (304.48,190) and (300,185.52) .. (300,180) -- cycle ;
\draw  [fill={rgb, 255:red, 255; green, 255; blue, 255 }  ,fill opacity=1 ] (480,180) .. controls (480,174.48) and (484.48,170) .. (490,170) .. controls (495.52,170) and (500,174.48) .. (500,180) .. controls (500,185.52) and (495.52,190) .. (490,190) .. controls (484.48,190) and (480,185.52) .. (480,180) -- cycle ;
\draw  [color={rgb, 255:red, 44; green, 160; blue, 44 }  ,draw opacity=1 ][fill={rgb, 255:red, 255; green, 255; blue, 255 }  ,fill opacity=1 ] (390,180) .. controls (390,174.48) and (394.48,170) .. (400,170) .. controls (405.52,170) and (410,174.48) .. (410,180) .. controls (410,185.52) and (405.52,190) .. (400,190) .. controls (394.48,190) and (390,185.52) .. (390,180) -- cycle ;
\draw  [color={rgb, 255:red, 31; green, 119; blue, 180 }  ,draw opacity=1 ][fill={rgb, 255:red, 255; green, 255; blue, 255 }  ,fill opacity=1 ] (360,180) .. controls (360,174.48) and (364.48,170) .. (370,170) .. controls (375.52,170) and (380,174.48) .. (380,180) .. controls (380,185.52) and (375.52,190) .. (370,190) .. controls (364.48,190) and (360,185.52) .. (360,180) -- cycle ;
\draw  [color={rgb, 255:red, 0; green, 0; blue, 0 }  ,draw opacity=1 ][fill={rgb, 255:red, 255; green, 255; blue, 255 }  ,fill opacity=1 ] (430,180) .. controls (430,174.48) and (434.48,170) .. (440,170) .. controls (445.52,170) and (450,174.48) .. (450,180) .. controls (450,185.52) and (445.52,190) .. (440,190) .. controls (434.48,190) and (430,185.52) .. (430,180) -- cycle ;

\draw (96,121) node    {$\cdots $};
\draw (206,121) node    {$\cdots $};
\draw (205,161) node    {$K+B$};
\draw (32,231) node [anchor=west] [inner sep=0.75pt]  [font=\large]  {$A\ =\ [\textcolor[rgb]{0.29,0.56,0.89}{3} ,\ \dotsc ,\ \textcolor[rgb]{0.17,0.63,0.17}{1}]$};
\draw (413.5,121) node    {$\cdots $};
\draw (545.23,121) node    {$\cdots $};
\draw (546,161) node    {$B$};
\draw (466,181) node    {$\cdots $};
\draw (356,221) node    {$B$};
\draw (466,221) node    {$B$};
\draw (401,261) node    {$K$};
\draw (441,61) node  [font=\normalsize]  {$r$};
\draw (74,101) node  [font=\normalsize]  {$0$};
\draw (128,101) node  [font=\normalsize]  {$0$};
\draw (178,101) node  [font=\normalsize]  {$x$};
\draw (234,101) node  [font=\normalsize]  {$x$};
\draw (124,151) node  [font=\normalsize]  {$1$};
\draw (38,151) node  [font=\normalsize]  {$1$};
\draw (33,61) node [anchor=west] [inner sep=0.75pt]  [font=\large] [align=left] {VN:};
\draw (303,61) node [anchor=west] [inner sep=0.75pt]  [font=\large] [align=left] {PN:};

\end{tikzpicture}

    \caption{Reducing BPP to $2$-star on Uniform $2$-star wVNE.
    Here $B = 4$, $x = 2 \cdot B \cdot K$.}
    \label{fig:BPP_to_UW2S}
\end{figure}

We do not provide the proof of Theorem~\ref{th:oversub_2Star_wVNE_NP} as it is too similar to the proof of Theorem~\ref{th:oversub_2Star_cVNE_NP} and makes use of ideas previously shown in Section~\ref{sec:wVNE}.
We leave it as an exercise for a curious reader.

\section{VNE without weights and with capacities}\label{sec:cVNE}

Sometimes we are not interested in the minimization of the cost as in we are the previous section~--- the cost difference between different paths is negligible.
Instead, networks could have other restrictions, for example, the commonly studied one is capacity~\cite{FischerBBMH13_survey, BallaniCKR11_virtualNet, FigielKNR0Z21_Tree, LiZWGZ15}.
In this section, we focus on physical networks with capacities.
We denote this problem as cVNE.

The complexity of cVNE without restrictions on topologies was comprehensively studied in~\cite{Amaldi16_complexity, RostS20_complexity}, where it was shown to be NP-complete.
Thus, we conclude that all NP-hard cVNE variants are also NP-complete.

\subsection{cVNE of Linear and $2$-star VN}

Previously we showed that wVNE of linear/$2$-star VN on linear/$2$-star PN is NP-complete (Theorems~\ref{th:WLLEP_NPC}~and~\ref{th:W2SEP_NPC}).
These results still hold for cVNE.
We prove that in Lemma~\ref{lem:01_cVNE_reduction} using the simple reduction.
We use the fact that we allow $0$ weights in wVNE NP-completeness proof to show that these proofs apply to cVNE.

\begin{lemma}\label{lem:01_cVNE_reduction}
    If some variant of wVNE ($\mathcal{S}$ on $\mathcal{G}$) with desired cost $\theta=0$ is NP-complete, then this variant of cVNE ($\mathcal{S}$ on $\mathcal{G}$) is also NP-complete.
\end{lemma}
\begin{proof}
    The wVNE is reduced to cVNE as follows.
    
    Assume an instance of wVNE: VN $S_B$ with $m$ edges and PN $G_B$.
    Without loss of generality, we can assume all edge weights are either $0$ or $1$, since the embedding cost is $\theta = 0$ and, thus, all non-zero weights behave the same.
    The instance for cVNE is obtained by changing edge weights to capacities in PN $G_B$ as follows: 
    ``Cheap'' $0$-weight edges now have capacity $m$ and ``expensive'' $1$-weight edges have capacity $0$.
    
    With such capacities, ``cheap'' edges in PN can fit any $m$ edges from VN and ``expensive'' edges can only fit edges from VN with $0$ demand.
    That simulates the behavior of wVNE with embedding cost $\theta = 0$: $0$-weight edges of VN are embedded anywhere and $1$-weighted edges of VN are only embedded on $0$-cost paths.
\end{proof}

Since the proof of Theorems~\ref{th:WLLEP_NPC}~and~\ref{th:W2SEP_NPC} reduces BPP to wVNE with desired cost $\theta = 0$, we can apply Lemma~\ref{lem:01_cVNE_reduction} to these results.
In other words, cVNE of line/$2$-star on line/$2$-star is NP-complete.
On the contrary, other variations of cVNE must be considered separately since Lemma~\ref{lem:01_cVNE_reduction} does not apply on them.

\subsection{cVNE of Star VN}\label{sec:star_cVNE}

The star is a particularly interesting example of Virtual Network topology for cVNE.
Rost et al.~\cite{RostFS15_stars} provided a polynomial time algorithm to embed a uniform star VN on any graph with capacities.
As for non-uniform star, Theorem~\ref{th:w_star} shows that wVNE with a star VN can be solved in polynomial time.
However, surprisingly, it is NP-hard in the case of cVNE.

\paragraph{Non-uniform star VN}

Consider a star virtual network in cVNE.
When non-uniform, this problem is NP-complete even when restricting PN to be a $2$-star or a line.

To prove it we use 3PP~--- strongly NP-complete problem~\cite{GareyJ79_NP}, i.e., the input can be given in the unary notation.

\begin{problem}[3-Partition Problem (3PP)]
    Given an array of $n = 3 \cdot m$ positive integers $A = [a_1, a_2, \ldots, a_n]$ which sum is $m \cdot T$:

    Is it possible to partition $A$ into $m$ disjoint subsets of size three, such that the sum of the elements in each subset equals exactly $T$?
\end{problem}

\begin{theorem}\label{th:CSTEP_NPC}
cVNE of star VN on $2$-star PN is NP-complete.
\end{theorem}
\begin{proof}
    We reduce 3PP to this problem, showing its NP-hardness.
    We assume that $m$ is at least~$4$.
    VN has $n - 1$ leaves, such that an edge to a leaf $i$ has demand $a_i$.
    PN has $m$ subtrees with three vertices each.
    All of the edges in PN have capacity of exactly $T$ (from the definition of 3PP problem).
    See example on Figure~\ref{fig:3PP_to_WS2S}.

    $\Rightarrow$:
    If there is a 3-partition for a given array $A$, there exist a following correct embedding: the root of VN is mapped to a root of PN and leaves are mapped corresponding to the partitioning.
    This is more clearly shown on Figure~\ref{fig:3PP_to_WS2S}.
    It is easy to see that all capacities are satisfied since the sum of $a_i$ in each subset is exactly $T$.

    $\Leftarrow$:
    Consider some correct mapping.
    If the root of VN is not placed on the root of PN, it has at most three outgoing edges with capacity not-exceeding $T$. 
    However, due to the assumption that $m \geq 4$, we need at least $\sum a_i \geq 4 T$ outgoing capacity for the root, which is only true for the root of PN.
    Thus, the root of VN must be mapped to the root of PN.
    It follows that the capacity of each adjacent edge to the root of PN is maxed at $T$ and each subtree will have exactly three vertices placed in it. 
    Thus, this mapping forms a desirable partition for $A$.

\end{proof}

\begin{figure}[!ht]
    \centering
    
    \tikzset{every picture/.style={line width=0.75pt}} 
\begin{tikzpicture}[x=0.75pt,y=0.75pt,yscale=-1,xscale=1]

\draw    (348,70) -- (408,130) ;
\draw [color={rgb, 255:red, 31; green, 119; blue, 180 }  ,draw opacity=1 ]   (115,70) -- (190,180) ;
\draw [color={rgb, 255:red, 214; green, 39; blue, 40 }  ,draw opacity=1 ]   (115,70) -- (160,180) ;
\draw [color={rgb, 255:red, 214; green, 39; blue, 40 }  ,draw opacity=1 ]   (115,70) -- (40,180) ;
\draw [color={rgb, 255:red, 31; green, 119; blue, 180 }  ,draw opacity=1 ]   (115,70) -- (70,180) ;
\draw [color={rgb, 255:red, 31; green, 119; blue, 180 }  ,draw opacity=1 ]   (115,70) -- (100,180) ;
\draw [color={rgb, 255:red, 214; green, 39; blue, 40 }  ,draw opacity=1 ]   (115,70) -- (130,180) ;
\draw  [color={rgb, 255:red, 214; green, 39; blue, 40 }  ,draw opacity=1 ][fill={rgb, 255:red, 255; green, 255; blue, 255 }  ,fill opacity=1 ] (30,180) .. controls (30,174.48) and (34.48,170) .. (40,170) .. controls (45.52,170) and (50,174.48) .. (50,180) .. controls (50,185.52) and (45.52,190) .. (40,190) .. controls (34.48,190) and (30,185.52) .. (30,180) -- cycle ;
\draw  [color={rgb, 255:red, 31; green, 119; blue, 180 }  ,draw opacity=1 ][fill={rgb, 255:red, 255; green, 255; blue, 255 }  ,fill opacity=1 ] (60,180) .. controls (60,174.48) and (64.48,170) .. (70,170) .. controls (75.52,170) and (80,174.48) .. (80,180) .. controls (80,185.52) and (75.52,190) .. (70,190) .. controls (64.48,190) and (60,185.52) .. (60,180) -- cycle ;
\draw  [color={rgb, 255:red, 31; green, 119; blue, 180 }  ,draw opacity=1 ][fill={rgb, 255:red, 255; green, 255; blue, 255 }  ,fill opacity=1 ] (90,180) .. controls (90,174.48) and (94.48,170) .. (100,170) .. controls (105.52,170) and (110,174.48) .. (110,180) .. controls (110,185.52) and (105.52,190) .. (100,190) .. controls (94.48,190) and (90,185.52) .. (90,180) -- cycle ;
\draw  [color={rgb, 255:red, 214; green, 39; blue, 40 }  ,draw opacity=1 ][fill={rgb, 255:red, 255; green, 255; blue, 255 }  ,fill opacity=1 ] (120,180) .. controls (120,174.48) and (124.48,170) .. (130,170) .. controls (135.52,170) and (140,174.48) .. (140,180) .. controls (140,185.52) and (135.52,190) .. (130,190) .. controls (124.48,190) and (120,185.52) .. (120,180) -- cycle ;
\draw  [color={rgb, 255:red, 214; green, 39; blue, 40 }  ,draw opacity=1 ][fill={rgb, 255:red, 255; green, 255; blue, 255 }  ,fill opacity=1 ] (150,180) .. controls (150,174.48) and (154.48,170) .. (160,170) .. controls (165.52,170) and (170,174.48) .. (170,180) .. controls (170,185.52) and (165.52,190) .. (160,190) .. controls (154.48,190) and (150,185.52) .. (150,180) -- cycle ;
\draw  [color={rgb, 255:red, 31; green, 119; blue, 180 }  ,draw opacity=1 ][fill={rgb, 255:red, 255; green, 255; blue, 255 }  ,fill opacity=1 ] (180,180) .. controls (180,174.48) and (184.48,170) .. (190,170) .. controls (195.52,170) and (200,174.48) .. (200,180) .. controls (200,185.52) and (195.52,190) .. (190,190) .. controls (184.48,190) and (180,185.52) .. (180,180) -- cycle ;
\draw  [color={rgb, 255:red, 0; green, 0; blue, 0 }  ,draw opacity=1 ][fill={rgb, 255:red, 255; green, 255; blue, 255 }  ,fill opacity=1 ] (105,70) .. controls (105,64.48) and (109.48,60) .. (115,60) .. controls (120.52,60) and (125,64.48) .. (125,70) .. controls (125,75.52) and (120.52,80) .. (115,80) .. controls (109.48,80) and (105,75.52) .. (105,70) -- cycle ;
\draw    (348,70) -- (288,130) ;
\draw [color={rgb, 255:red, 214; green, 39; blue, 40 }  ,draw opacity=1 ]   (288,130) -- (258,180) ;
\draw [color={rgb, 255:red, 214; green, 39; blue, 40 }  ,draw opacity=1 ]   (288,130) -- (318,180) ;
\draw  [color={rgb, 255:red, 214; green, 39; blue, 40 }  ,draw opacity=1 ][fill={rgb, 255:red, 255; green, 255; blue, 255 }  ,fill opacity=1 ] (278,130) .. controls (278,124.48) and (282.48,120) .. (288,120) .. controls (293.52,120) and (298,124.48) .. (298,130) .. controls (298,135.52) and (293.52,140) .. (288,140) .. controls (282.48,140) and (278,135.52) .. (278,130) -- cycle ;
\draw  [color={rgb, 255:red, 214; green, 39; blue, 40 }  ,draw opacity=1 ][fill={rgb, 255:red, 255; green, 255; blue, 255 }  ,fill opacity=1 ] (248,180) .. controls (248,174.48) and (252.48,170) .. (258,170) .. controls (263.52,170) and (268,174.48) .. (268,180) .. controls (268,185.52) and (263.52,190) .. (258,190) .. controls (252.48,190) and (248,185.52) .. (248,180) -- cycle ;
\draw  [color={rgb, 255:red, 214; green, 39; blue, 40 }  ,draw opacity=1 ][fill={rgb, 255:red, 255; green, 255; blue, 255 }  ,fill opacity=1 ] (308,180) .. controls (308,174.48) and (312.48,170) .. (318,170) .. controls (323.52,170) and (328,174.48) .. (328,180) .. controls (328,185.52) and (323.52,190) .. (318,190) .. controls (312.48,190) and (308,185.52) .. (308,180) -- cycle ;
\draw  [color={rgb, 255:red, 0; green, 0; blue, 0 }  ,draw opacity=1 ][fill={rgb, 255:red, 255; green, 255; blue, 255 }  ,fill opacity=1 ] (338,70) .. controls (338,64.48) and (342.48,60) .. (348,60) .. controls (353.52,60) and (358,64.48) .. (358,70) .. controls (358,75.52) and (353.52,80) .. (348,80) .. controls (342.48,80) and (338,75.52) .. (338,70) -- cycle ;
\draw [color={rgb, 255:red, 31; green, 119; blue, 180 }  ,draw opacity=1 ]   (408,130) -- (378,180) ;
\draw [color={rgb, 255:red, 31; green, 119; blue, 180 }  ,draw opacity=1 ]   (408,130) -- (438,180) ;
\draw  [color={rgb, 255:red, 31; green, 119; blue, 180 }  ,draw opacity=1 ][fill={rgb, 255:red, 255; green, 255; blue, 255 }  ,fill opacity=1 ] (398,130) .. controls (398,124.48) and (402.48,120) .. (408,120) .. controls (413.52,120) and (418,124.48) .. (418,130) .. controls (418,135.52) and (413.52,140) .. (408,140) .. controls (402.48,140) and (398,135.52) .. (398,130) -- cycle ;
\draw  [color={rgb, 255:red, 31; green, 119; blue, 180 }  ,draw opacity=1 ][fill={rgb, 255:red, 255; green, 255; blue, 255 }  ,fill opacity=1 ] (368,180) .. controls (368,174.48) and (372.48,170) .. (378,170) .. controls (383.52,170) and (388,174.48) .. (388,180) .. controls (388,185.52) and (383.52,190) .. (378,190) .. controls (372.48,190) and (368,185.52) .. (368,180) -- cycle ;
\draw  [color={rgb, 255:red, 31; green, 119; blue, 180 }  ,draw opacity=1 ][fill={rgb, 255:red, 255; green, 255; blue, 255 }  ,fill opacity=1 ] (428,180) .. controls (428,174.48) and (432.48,170) .. (438,170) .. controls (443.52,170) and (448,174.48) .. (448,180) .. controls (448,185.52) and (443.52,190) .. (438,190) .. controls (432.48,190) and (428,185.52) .. (428,180) -- cycle ;

\draw (32,73.5) node [anchor=west] [inner sep=0.75pt]  [font=\large] [align=left] {VN:};
\draw (250,73.5) node [anchor=west] [inner sep=0.75pt]  [font=\large] [align=left] {PN:};
\draw (305,91) node    {$8/8$};
\draw (171,161) node    {$1$};
\draw (145,161) node    {$1$};
\draw (121,161) node    {$2$};
\draw (95,161) node    {$4$};
\draw (71,161) node    {$3$};
\draw (45,161) node    {$5$};
\draw (348,71) node    {$c$};
\draw (395,91) node    {$8/8$};
\draw (253,151) node    {$1/8$};
\draw (445,151) node    {$1/8$};
\draw (323,151) node    {$5/8$};
\draw (373,151) node    {$3/8$};
\draw (115,71) node    {$c$};

\end{tikzpicture}

    \caption{Reducing 3PP to star on $2$-star with capacities for $A = [5, 3, 4, 2, 1 , 1]$.}
    \label{fig:3PP_to_WS2S}
\end{figure}
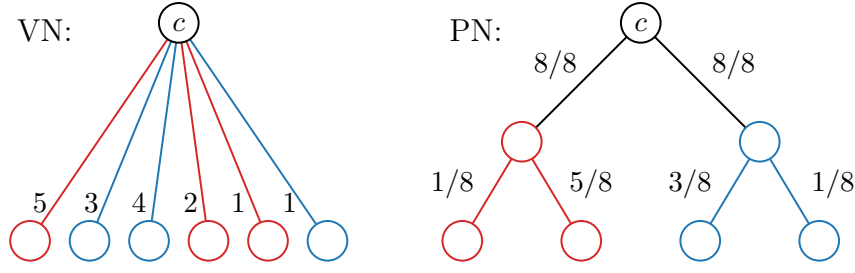

For the next proof, we need an additional NP-complete problem~\cite{Karp72_NP}.
\begin{problem}[Partition Problem (PP)]
    Given an array of $n$ positive integers $A = [a_1, a_2, \ldots, a_n]$ that sum to $2 \cdot T$:

    Is it possible to partition $A$ into two disjoint subsets, such that the sum of the elements in each subset equals exactly $T$?
\end{problem}

\begin{theorem}\label{th:CSLEP_NPC}
cVNE of star VN on linear PN is NP-complete.
\end{theorem}
\begin{proof}

    To show NP-hardness, we reduce PP to the problem.
    VN is constructed the same way as in the proof of Theorem~\ref{th:CSTEP_NPC}, containing $n$ leaves with demand $a_i$ on edges.
    PN is a line with $n + 1$ vertices and the capacity of each edge is $\frac{1}{2} \sum a_i$.
    See example on Figure~\ref{fig:PP_to_WSL}.

    $\Rightarrow$:
    Suppose that there exists a correct partition into two sets $S_1$ and $S_2$ of equal sum, we can place the root of VN in the position $x = |S_1| + 1$.
    Then, we map corresponding vertices of $S_1$ to the left of $x$ and vertices of $S_2$ to the right of it.
    Since the total capacity out of $S_i$ is exactly $\frac{1}{2} \sum a_i$ to each side, the constraints are satisfied.

    $\Leftarrow$:
    Now, consider the opposite: there exists a correct mapping of VN onto PN.
    Then, independent on the position of the root of the $VN$, vertices are partitioned into two sets: on the left and on the right from the root.
    Due to the capacity constraints, each of the two subsets sum to $\frac{1}{2} \sum a_i$.
    
\end{proof}

\begin{figure}[!ht]
    \centering

\tikzset{every picture/.style={line width=0.75pt}} 

\begin{tikzpicture}[x=0.75pt,y=0.75pt,yscale=-1,xscale=1]

\draw [color={rgb, 255:red, 44; green, 160; blue, 44 }  ,draw opacity=1 ]   (160,90) -- (200,90) ;
\draw  [draw opacity=0] (90.2,87.47) .. controls (92.78,72.1) and (118.56,60.02) .. (149.98,60.01) .. controls (183.12,60) and (209.99,73.41) .. (210,89.97) .. controls (210,90) and (210,90.03) .. (210,90.07) -- (149.99,89.99) -- cycle ; \draw  [color={rgb, 255:red, 31; green, 119; blue, 180 }  ,draw opacity=1 ] (90.2,87.47) .. controls (92.78,72.1) and (118.56,60.02) .. (149.98,60.01) .. controls (183.12,60) and (209.99,73.41) .. (210,89.97) .. controls (210,90) and (210,90.03) .. (210,90.07) ;  
\draw  [color={rgb, 255:red, 74; green, 144; blue, 226 }  ,draw opacity=1 ][fill={rgb, 255:red, 255; green, 255; blue, 255 }  ,fill opacity=1 ] (80,160) .. controls (80,154.48) and (84.48,150) .. (90,150) .. controls (95.52,150) and (100,154.48) .. (100,160) .. controls (100,165.52) and (95.52,170) .. (90,170) .. controls (84.48,170) and (80,165.52) .. (80,160) -- cycle ;
\draw  [color={rgb, 255:red, 44; green, 160; blue, 44 }  ,draw opacity=1 ][fill={rgb, 255:red, 255; green, 255; blue, 255 }  ,fill opacity=1 ] (140,160) .. controls (140,154.48) and (144.48,150) .. (150,150) .. controls (155.52,150) and (160,154.48) .. (160,160) .. controls (160,165.52) and (155.52,170) .. (150,170) .. controls (144.48,170) and (140,165.52) .. (140,160) -- cycle ;
\draw  [color={rgb, 255:red, 0; green, 0; blue, 0 }  ,draw opacity=1 ][fill={rgb, 255:red, 255; green, 255; blue, 255 }  ,fill opacity=1 ] (200,160) .. controls (200,154.48) and (204.48,150) .. (210,150) .. controls (215.52,150) and (220,154.48) .. (220,160) .. controls (220,165.52) and (215.52,170) .. (210,170) .. controls (204.48,170) and (200,165.52) .. (200,160) -- cycle ;
\draw    (100,160) -- (140,160) ;
\draw [color={rgb, 255:red, 0; green, 0; blue, 0 }  ,draw opacity=1 ]   (160,160) -- (200,160) ;
\draw  [color={rgb, 255:red, 31; green, 119; blue, 180 }  ,draw opacity=1 ][fill={rgb, 255:red, 255; green, 255; blue, 255 }  ,fill opacity=1 ] (80,90) .. controls (80,84.48) and (84.48,80) .. (90,80) .. controls (95.52,80) and (100,84.48) .. (100,90) .. controls (100,95.52) and (95.52,100) .. (90,100) .. controls (84.48,100) and (80,95.52) .. (80,90) -- cycle ;
\draw  [color={rgb, 255:red, 44; green, 160; blue, 44 }  ,draw opacity=1 ][fill={rgb, 255:red, 255; green, 255; blue, 255 }  ,fill opacity=1 ] (140,90) .. controls (140,84.48) and (144.48,80) .. (150,80) .. controls (155.52,80) and (160,84.48) .. (160,90) .. controls (160,95.52) and (155.52,100) .. (150,100) .. controls (144.48,100) and (140,95.52) .. (140,90) -- cycle ;
\draw  [color={rgb, 255:red, 0; green, 0; blue, 0 }  ,draw opacity=1 ][fill={rgb, 255:red, 255; green, 255; blue, 255 }  ,fill opacity=1 ] (200,90) .. controls (200,84.48) and (204.48,80) .. (210,80) .. controls (215.52,80) and (220,84.48) .. (220,90) .. controls (220,95.52) and (215.52,100) .. (210,100) .. controls (204.48,100) and (200,95.52) .. (200,90) -- cycle ;
\draw  [color={rgb, 255:red, 214; green, 39; blue, 40 }  ,draw opacity=1 ][fill={rgb, 255:red, 255; green, 255; blue, 255 }  ,fill opacity=1 ] (260,90) .. controls (260,84.48) and (264.48,80) .. (270,80) .. controls (275.52,80) and (280,84.48) .. (280,90) .. controls (280,95.52) and (275.52,100) .. (270,100) .. controls (264.48,100) and (260,95.52) .. (260,90) -- cycle ;
\draw  [color={rgb, 255:red, 214; green, 39; blue, 40 }  ,draw opacity=1 ][fill={rgb, 255:red, 255; green, 255; blue, 255 }  ,fill opacity=1 ] (260,160) .. controls (260,154.48) and (264.48,150) .. (270,150) .. controls (275.52,150) and (280,154.48) .. (280,160) .. controls (280,165.52) and (275.52,170) .. (270,170) .. controls (264.48,170) and (260,165.52) .. (260,160) -- cycle ;
\draw [color={rgb, 255:red, 0; green, 0; blue, 0 }  ,draw opacity=1 ]   (220,160) -- (260,160) ;
\draw [color={rgb, 255:red, 214; green, 39; blue, 40 }  ,draw opacity=1 ]   (220,90) -- (260,90) ;

\draw (209,90) node  [font=\normalsize]  {$c$};
\draw (32,88.5) node [anchor=west] [inner sep=0.75pt]  [font=\large] [align=left] {VN:};
\draw (32,158.5) node [anchor=west] [inner sep=0.75pt]  [font=\large] [align=left] {PN:};
\draw (209,160) node  [font=\normalsize]  {$c$};
\draw (150.99,51) node    {$3$};
\draw (181,81) node    {$2$};
\draw (241,81) node    {$5$};
\draw (181,151) node    {$5/5$};
\draw (241,151) node    {$5/5$};
\draw (121,151) node    {$3/5$};

\end{tikzpicture}

    \caption{Reducing PP to star on line cVNE for $A = [5, 3, 2]$.}
    \label{fig:PP_to_WSL}
    
\end{figure}
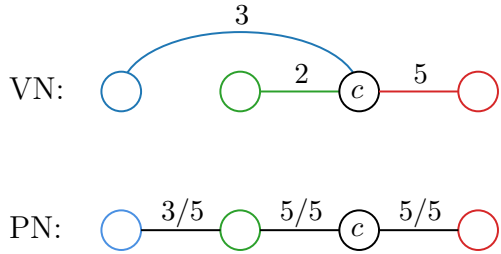

\subsection{cVNE of Oversubscribed $2$-star VN}\label{sec:oversub_2satr}

One of the problems with the uniform star VN is that the bandwidth required for the central node is high.
To relax the load on the central vertex in some applications, the \emph{oversubscribed} $2$-star VN topology is used~\cite{BallaniCKR11_virtualNet}.
This topology is a special case of a $2$-star graph.
An example of an oversubscribed $2$-star appears in Figure~\ref{fig:3DM_to_Oversub_2Star}.

The cVNE problem with oversubscribed $2$-star VN is easier than cVNE with a regular $2$-star VN.
For example, cVNE of an oversubscribed $2$-star VN on a tree PN has a polynomial algorithm as will be shown later in Section~\ref{sec:oversub_on_tree}.
However, for non-tree PN, embedding of an oversubscribed $2$-star VN remains NP-hard.
We show that by the reduction of the 3-dimensional matching problem.
The 3-dimensional matching is the extension of a bipartite matching.

\begin{problem}[3-Dimensional Matching (3DM)]
    Let $X$, $Y$, $Z$ be three disjoint sets, such that $|X| = |Y| = |Z| = q$ and $V = X \cup Y \cup Z$.
    Let $T \subseteq X \times Y \times Z$ be a set of triplets. 
    
    Determine whether there exist $M$: a subset of $T$ of size $q$, such that any two hyperedges $(x_i, y_i, z_i), (x_j, y_j, z_j) \in M$ do not intersect, i.e., $x_i \ne x_j$, $y_i \ne y_j$ and $z_i \ne z_j$.
\end{problem}

3DM is NP-complete~\cite{GareyJ79_NP}.

\begin{figure}[!ht]
    \centering

\tikzset{every picture/.style={line width=0.75pt}} 

\begin{tikzpicture}[x=0.75pt,y=0.75pt,yscale=-1,xscale=1]

\draw  [dash pattern={on 4.5pt off 4.5pt}]  (560,60) -- (542.5,91.88) ;
\draw  [dash pattern={on 4.5pt off 4.5pt}]  (460,60) -- (462.5,91.88) ;
\draw  [dash pattern={on 4.5pt off 4.5pt}]  (380,70) -- (437.5,91.88) ;
\draw  [dash pattern={on 4.5pt off 4.5pt}]  (300,60) -- (317.5,91.88) ;
\draw  [dash pattern={on 4.5pt off 4.5pt}]  (360,60) -- (342.5,91.88) ;
\draw  [fill={rgb, 255:red, 255; green, 255; blue, 255 }  ,fill opacity=1 ] (300,10.5) .. controls (300,4.7) and (304.7,0) .. (310.5,0) -- (349.5,0) .. controls (355.3,0) and (360,4.7) .. (360,10.5) -- (360,59.5) .. controls (360,65.3) and (355.3,70) .. (349.5,70) -- (310.5,70) .. controls (304.7,70) and (300,65.3) .. (300,59.5) -- cycle ;
\draw  [fill={rgb, 255:red, 255; green, 255; blue, 255 }  ,fill opacity=1 ] (370,12.25) .. controls (370,5.48) and (375.48,0) .. (382.25,0) -- (447.75,0) .. controls (454.52,0) and (460,5.48) .. (460,12.25) -- (460,57.75) .. controls (460,64.52) and (454.52,70) .. (447.75,70) -- (382.25,70) .. controls (375.48,70) and (370,64.52) .. (370,57.75) -- cycle ;
\draw  [fill={rgb, 255:red, 255; green, 255; blue, 255 }  ,fill opacity=1 ] (470,12.25) .. controls (470,5.48) and (475.48,0) .. (482.25,0) -- (547.75,0) .. controls (554.52,0) and (560,5.48) .. (560,12.25) -- (560,57.75) .. controls (560,64.52) and (554.52,70) .. (547.75,70) -- (482.25,70) .. controls (475.48,70) and (470,64.52) .. (470,57.75) -- cycle ;
\draw [line width=2.25]    (370,140) -- (429.48,180) ;
\draw [line width=2.25]    (409.48,140) -- (429.48,180) ;
\draw [line width=2.25]    (449.48,140) -- (429.48,180) ;
\draw [line width=2.25]    (490,140) -- (430,180) ;
\draw [color={rgb, 255:red, 214; green, 39; blue, 40 }  ,draw opacity=1 ][line width=2.25]    (100,7.5) -- (200,7.5) ;
\draw [color={rgb, 255:red, 44; green, 160; blue, 44 }  ,draw opacity=1 ][line width=2.25]    (100,62.5) -- (200,62.5) ;
\draw [color={rgb, 255:red, 31; green, 119; blue, 180 }  ,draw opacity=1 ][line width=2.25]    (150,12.5) -- (200,12.5) ;
\draw [color={rgb, 255:red, 31; green, 119; blue, 180 }  ,draw opacity=1 ][line width=2.25]    (93,62.5) -- (148,12.5) ;
\draw [color={rgb, 255:red, 255; green, 127; blue, 14 }  ,draw opacity=1 ][line width=2.25]    (100,62.5) -- (150,17.46) ;
\draw [color={rgb, 255:red, 255; green, 127; blue, 14 }  ,draw opacity=1 ][line width=2.25]    (200,62.5) -- (150,17.46) ;
\draw  [color={rgb, 255:red, 0; green, 0; blue, 0 }  ,draw opacity=1 ][fill={rgb, 255:red, 255; green, 255; blue, 255 }  ,fill opacity=1 ] (90,12.5) .. controls (90,6.98) and (94.48,2.5) .. (100,2.5) .. controls (105.52,2.5) and (110,6.98) .. (110,12.5) .. controls (110,18.02) and (105.52,22.5) .. (100,22.5) .. controls (94.48,22.5) and (90,18.02) .. (90,12.5) -- cycle ;
\draw  [color={rgb, 255:red, 0; green, 0; blue, 0 }  ,draw opacity=1 ][fill={rgb, 255:red, 255; green, 255; blue, 255 }  ,fill opacity=1 ] (140,12.5) .. controls (140,6.98) and (144.48,2.5) .. (150,2.5) .. controls (155.52,2.5) and (160,6.98) .. (160,12.5) .. controls (160,18.02) and (155.52,22.5) .. (150,22.5) .. controls (144.48,22.5) and (140,18.02) .. (140,12.5) -- cycle ;
\draw  [color={rgb, 255:red, 0; green, 0; blue, 0 }  ,draw opacity=1 ][fill={rgb, 255:red, 255; green, 255; blue, 255 }  ,fill opacity=1 ] (190,12.5) .. controls (190,6.98) and (194.48,2.5) .. (200,2.5) .. controls (205.52,2.5) and (210,6.98) .. (210,12.5) .. controls (210,18.02) and (205.52,22.5) .. (200,22.5) .. controls (194.48,22.5) and (190,18.02) .. (190,12.5) -- cycle ;
\draw  [color={rgb, 255:red, 0; green, 0; blue, 0 }  ,draw opacity=1 ][fill={rgb, 255:red, 255; green, 255; blue, 255 }  ,fill opacity=1 ] (90,62.5) .. controls (90,56.98) and (94.48,52.5) .. (100,52.5) .. controls (105.52,52.5) and (110,56.98) .. (110,62.5) .. controls (110,68.02) and (105.52,72.5) .. (100,72.5) .. controls (94.48,72.5) and (90,68.02) .. (90,62.5) -- cycle ;
\draw  [color={rgb, 255:red, 0; green, 0; blue, 0 }  ,draw opacity=1 ][fill={rgb, 255:red, 255; green, 255; blue, 255 }  ,fill opacity=1 ] (140,62.5) .. controls (140,56.98) and (144.48,52.5) .. (150,52.5) .. controls (155.52,52.5) and (160,56.98) .. (160,62.5) .. controls (160,68.02) and (155.52,72.5) .. (150,72.5) .. controls (144.48,72.5) and (140,68.02) .. (140,62.5) -- cycle ;
\draw  [color={rgb, 255:red, 0; green, 0; blue, 0 }  ,draw opacity=1 ][fill={rgb, 255:red, 255; green, 255; blue, 255 }  ,fill opacity=1 ] (190,62.5) .. controls (190,56.98) and (194.48,52.5) .. (200,52.5) .. controls (205.52,52.5) and (210,56.98) .. (210,62.5) .. controls (210,68.02) and (205.52,72.5) .. (200,72.5) .. controls (194.48,72.5) and (190,68.02) .. (190,62.5) -- cycle ;
\draw [line width=2.25]    (100,220) -- (150,180) ;
\draw    (100,220) -- (130,260) ;
\draw    (100,220) -- (70,260) ;
\draw [line width=2.25]    (150,180) -- (200,220) ;
\draw    (200,220) -- (230,260) ;
\draw    (200,220) -- (170,260) ;
\draw  [color={rgb, 255:red, 0; green, 0; blue, 0 }  ,draw opacity=1 ][fill={rgb, 255:red, 255; green, 255; blue, 255 }  ,fill opacity=1 ] (140,180) .. controls (140,174.48) and (144.48,170) .. (150,170) .. controls (155.52,170) and (160,174.48) .. (160,180) .. controls (160,185.52) and (155.52,190) .. (150,190) .. controls (144.48,190) and (140,185.52) .. (140,180) -- cycle ;
\draw  [color={rgb, 255:red, 0; green, 0; blue, 0 }  ,draw opacity=1 ][fill={rgb, 255:red, 255; green, 255; blue, 255 }  ,fill opacity=1 ] (120,260) .. controls (120,254.48) and (124.48,250) .. (130,250) .. controls (135.52,250) and (140,254.48) .. (140,260) .. controls (140,265.52) and (135.52,270) .. (130,270) .. controls (124.48,270) and (120,265.52) .. (120,260) -- cycle ;
\draw  [color={rgb, 255:red, 0; green, 0; blue, 0 }  ,draw opacity=1 ][fill={rgb, 255:red, 255; green, 255; blue, 255 }  ,fill opacity=1 ] (160,260) .. controls (160,254.48) and (164.48,250) .. (170,250) .. controls (175.52,250) and (180,254.48) .. (180,260) .. controls (180,265.52) and (175.52,270) .. (170,270) .. controls (164.48,270) and (160,265.52) .. (160,260) -- cycle ;
\draw  [color={rgb, 255:red, 0; green, 0; blue, 0 }  ,draw opacity=1 ][fill={rgb, 255:red, 255; green, 255; blue, 255 }  ,fill opacity=1 ] (220,260) .. controls (220,254.48) and (224.48,250) .. (230,250) .. controls (235.52,250) and (240,254.48) .. (240,260) .. controls (240,265.52) and (235.52,270) .. (230,270) .. controls (224.48,270) and (220,265.52) .. (220,260) -- cycle ;
\draw  [color={rgb, 255:red, 0; green, 0; blue, 0 }  ,draw opacity=1 ][fill={rgb, 255:red, 255; green, 255; blue, 255 }  ,fill opacity=1 ] (60,260) .. controls (60,254.48) and (64.48,250) .. (70,250) .. controls (75.52,250) and (80,254.48) .. (80,260) .. controls (80,265.52) and (75.52,270) .. (70,270) .. controls (64.48,270) and (60,265.52) .. (60,260) -- cycle ;
\draw  [color={rgb, 255:red, 0; green, 0; blue, 0 }  ,draw opacity=1 ][fill={rgb, 255:red, 255; green, 255; blue, 255 }  ,fill opacity=1 ] (90,220) .. controls (90,214.48) and (94.48,210) .. (100,210) .. controls (105.52,210) and (110,214.48) .. (110,220) .. controls (110,225.52) and (105.52,230) .. (100,230) .. controls (94.48,230) and (90,225.52) .. (90,220) -- cycle ;
\draw  [color={rgb, 255:red, 0; green, 0; blue, 0 }  ,draw opacity=1 ][fill={rgb, 255:red, 255; green, 255; blue, 255 }  ,fill opacity=1 ] (190,220) .. controls (190,214.48) and (194.48,210) .. (200,210) .. controls (205.52,210) and (210,214.48) .. (210,220) .. controls (210,225.52) and (205.52,230) .. (200,230) .. controls (194.48,230) and (190,225.52) .. (190,220) -- cycle ;
\draw   (60.14,273.29) .. controls (60.15,277.96) and (62.48,280.29) .. (67.15,280.28) -- (90.22,280.28) .. controls (96.89,280.27) and (100.22,282.6) .. (100.22,287.27) .. controls (100.22,282.6) and (103.55,280.27) .. (110.22,280.27)(107.22,280.27) -- (133.29,280.26) .. controls (137.96,280.26) and (140.29,277.93) .. (140.29,273.26) ;
\draw   (160,273.43) .. controls (160,278.1) and (162.33,280.43) .. (167,280.43) -- (190.07,280.42) .. controls (196.74,280.41) and (200.07,282.74) .. (200.08,287.41) .. controls (200.07,282.74) and (203.4,280.41) .. (210.07,280.41)(207.07,280.41) -- (233.15,280.4) .. controls (237.82,280.4) and (240.15,278.07) .. (240.14,273.4) ;
\draw   (59.64,303.29) .. controls (59.64,307.96) and (61.97,310.29) .. (66.64,310.29) -- (139.9,310.29) .. controls (146.57,310.29) and (149.9,312.62) .. (149.9,317.29) .. controls (149.9,312.62) and (153.23,310.29) .. (159.9,310.29)(156.9,310.29) -- (233.17,310.29) .. controls (237.84,310.29) and (240.17,307.96) .. (240.17,303.29) ;
\draw [line width=2.25]    (379.48,220) -- (429.48,180) ;
\draw    (379.48,220) -- (409.48,260) ;
\draw    (379.48,220) -- (349.48,260) ;
\draw [line width=2.25]    (429.48,180) -- (479.48,220) ;
\draw    (479.48,220) -- (509.48,260) ;
\draw    (479.48,220) -- (449.48,260) ;
\draw  [color={rgb, 255:red, 0; green, 0; blue, 0 }  ,draw opacity=1 ][fill={rgb, 255:red, 255; green, 255; blue, 255 }  ,fill opacity=1 ] (419.48,180) .. controls (419.48,174.48) and (423.95,170) .. (429.48,170) .. controls (435,170) and (439.48,174.48) .. (439.48,180) .. controls (439.48,185.52) and (435,190) .. (429.48,190) .. controls (423.95,190) and (419.48,185.52) .. (419.48,180) -- cycle ;
\draw  [color={rgb, 255:red, 0; green, 0; blue, 0 }  ,draw opacity=1 ][fill={rgb, 255:red, 255; green, 255; blue, 255 }  ,fill opacity=1 ] (399.48,260) .. controls (399.48,254.48) and (403.95,250) .. (409.48,250) .. controls (415,250) and (419.48,254.48) .. (419.48,260) .. controls (419.48,265.52) and (415,270) .. (409.48,270) .. controls (403.95,270) and (399.48,265.52) .. (399.48,260) -- cycle ;
\draw  [color={rgb, 255:red, 0; green, 0; blue, 0 }  ,draw opacity=1 ][fill={rgb, 255:red, 255; green, 255; blue, 255 }  ,fill opacity=1 ] (439.48,260) .. controls (439.48,254.48) and (443.95,250) .. (449.48,250) .. controls (455,250) and (459.48,254.48) .. (459.48,260) .. controls (459.48,265.52) and (455,270) .. (449.48,270) .. controls (443.95,270) and (439.48,265.52) .. (439.48,260) -- cycle ;
\draw  [color={rgb, 255:red, 0; green, 0; blue, 0 }  ,draw opacity=1 ][fill={rgb, 255:red, 255; green, 255; blue, 255 }  ,fill opacity=1 ] (499.48,260) .. controls (499.48,254.48) and (503.95,250) .. (509.48,250) .. controls (515,250) and (519.48,254.48) .. (519.48,260) .. controls (519.48,265.52) and (515,270) .. (509.48,270) .. controls (503.95,270) and (499.48,265.52) .. (499.48,260) -- cycle ;
\draw  [color={rgb, 255:red, 0; green, 0; blue, 0 }  ,draw opacity=1 ][fill={rgb, 255:red, 255; green, 255; blue, 255 }  ,fill opacity=1 ] (339.48,260) .. controls (339.48,254.48) and (343.95,250) .. (349.48,250) .. controls (355,250) and (359.48,254.48) .. (359.48,260) .. controls (359.48,265.52) and (355,270) .. (349.48,270) .. controls (343.95,270) and (339.48,265.52) .. (339.48,260) -- cycle ;
\draw  [color={rgb, 255:red, 0; green, 0; blue, 0 }  ,draw opacity=1 ][fill={rgb, 255:red, 255; green, 255; blue, 255 }  ,fill opacity=1 ] (369.48,220) .. controls (369.48,214.48) and (373.95,210) .. (379.48,210) .. controls (385,210) and (389.48,214.48) .. (389.48,220) .. controls (389.48,225.52) and (385,230) .. (379.48,230) .. controls (373.95,230) and (369.48,225.52) .. (369.48,220) -- cycle ;
\draw  [color={rgb, 255:red, 0; green, 0; blue, 0 }  ,draw opacity=1 ][fill={rgb, 255:red, 255; green, 255; blue, 255 }  ,fill opacity=1 ] (469.48,220) .. controls (469.48,214.48) and (473.95,210) .. (479.48,210) .. controls (485,210) and (489.48,214.48) .. (489.48,220) .. controls (489.48,225.52) and (485,230) .. (479.48,230) .. controls (473.95,230) and (469.48,225.52) .. (469.48,220) -- cycle ;
\draw   (339.64,273.29) .. controls (339.65,277.96) and (341.98,280.29) .. (346.65,280.28) -- (369.72,280.28) .. controls (376.39,280.27) and (379.72,282.6) .. (379.72,287.27) .. controls (379.72,282.6) and (383.05,280.27) .. (389.72,280.27)(386.72,280.27) -- (412.79,280.26) .. controls (417.46,280.26) and (419.79,277.93) .. (419.79,273.26) ;
\draw   (439.83,273.43) .. controls (439.84,278.1) and (442.17,280.43) .. (446.84,280.43) -- (469.91,280.42) .. controls (476.58,280.41) and (479.91,282.74) .. (479.91,287.41) .. controls (479.91,282.74) and (483.24,280.41) .. (489.91,280.41)(486.91,280.41) -- (512.98,280.4) .. controls (517.65,280.4) and (519.98,278.07) .. (519.98,273.4) ;
\draw   (339.48,302.95) .. controls (339.48,307.62) and (341.81,309.95) .. (346.48,309.95) -- (419.74,309.95) .. controls (426.41,309.95) and (429.74,312.28) .. (429.74,316.95) .. controls (429.74,312.28) and (433.07,309.95) .. (439.74,309.95)(436.74,309.95) -- (513,309.95) .. controls (517.67,309.95) and (520,307.62) .. (520,302.95) ;
\draw [color={rgb, 255:red, 214; green, 39; blue, 40 }  ,draw opacity=1 ]   (370,140) -- (365,100) ;
\draw [color={rgb, 255:red, 214; green, 39; blue, 40 }  ,draw opacity=1 ]   (370,140) -- (375,100) ;
\draw [color={rgb, 255:red, 31; green, 119; blue, 180 }  ,draw opacity=1 ]   (410,140) -- (405,100) ;
\draw [color={rgb, 255:red, 31; green, 119; blue, 180 }  ,draw opacity=1 ]   (410,140) -- (415,100) ;
\draw [color={rgb, 255:red, 255; green, 127; blue, 14 }  ,draw opacity=1 ]   (450,140) -- (445,100) ;
\draw [color={rgb, 255:red, 255; green, 127; blue, 14 }  ,draw opacity=1 ]   (450,140) -- (455,100) ;
\draw [color={rgb, 255:red, 44; green, 160; blue, 44 }  ,draw opacity=1 ]   (490,140) -- (485,100) ;
\draw [color={rgb, 255:red, 44; green, 160; blue, 44 }  ,draw opacity=1 ]   (490,140) -- (495,100) ;
\draw [color={rgb, 255:red, 214; green, 39; blue, 40 }  ,draw opacity=1 ]   (370,140) -- (415,100) ;
\draw [color={rgb, 255:red, 214; green, 39; blue, 40 }  ,draw opacity=1 ]   (370,140) -- (405,100) ;
\draw [color={rgb, 255:red, 214; green, 39; blue, 40 }  ,draw opacity=1 ]   (335,100) -- (370,140) ;
\draw [color={rgb, 255:red, 214; green, 39; blue, 40 }  ,draw opacity=1 ]   (325,100) -- (370,140) ;
\draw [color={rgb, 255:red, 31; green, 119; blue, 180 }  ,draw opacity=1 ]   (409.48,140) -- (365,100) ;
\draw [color={rgb, 255:red, 31; green, 119; blue, 180 }  ,draw opacity=1 ]   (410,140) -- (375,100) ;
\draw [color={rgb, 255:red, 31; green, 119; blue, 180 }  ,draw opacity=1 ]   (410,140) -- (445,100) ;
\draw [color={rgb, 255:red, 31; green, 119; blue, 180 }  ,draw opacity=1 ]   (409.48,140) -- (455,100) ;
\draw [color={rgb, 255:red, 255; green, 127; blue, 14 }  ,draw opacity=1 ]   (450,140) -- (525,100) ;
\draw [color={rgb, 255:red, 255; green, 127; blue, 14 }  ,draw opacity=1 ]   (535,100) -- (450,140) ;
\draw [color={rgb, 255:red, 255; green, 127; blue, 14 }  ,draw opacity=1 ]   (365,100) -- (449.48,140) ;
\draw [color={rgb, 255:red, 255; green, 127; blue, 14 }  ,draw opacity=1 ]   (375,100) -- (450,140) ;
\draw [color={rgb, 255:red, 44; green, 160; blue, 44 }  ,draw opacity=1 ]   (490,140) -- (445,100) ;
\draw [color={rgb, 255:red, 44; green, 160; blue, 44 }  ,draw opacity=1 ]   (490,140) -- (470.44,117.64) -- (455,100) ;
\draw [color={rgb, 255:red, 44; green, 160; blue, 44 }  ,draw opacity=1 ]   (490,140) -- (525,100) ;
\draw [color={rgb, 255:red, 44; green, 160; blue, 44 }  ,draw opacity=1 ]   (490,140) -- (535,100) ;
\draw  [color={rgb, 255:red, 31; green, 119; blue, 180 }  ,draw opacity=1 ][fill={rgb, 255:red, 255; green, 255; blue, 255 }  ,fill opacity=1 ] (400,140) .. controls (400,134.48) and (404.48,130) .. (410,130) .. controls (415.52,130) and (420,134.48) .. (420,140) .. controls (420,145.52) and (415.52,150) .. (410,150) .. controls (404.48,150) and (400,145.52) .. (400,140) -- cycle ;
\draw  [color={rgb, 255:red, 255; green, 127; blue, 14 }  ,draw opacity=1 ][fill={rgb, 255:red, 255; green, 255; blue, 255 }  ,fill opacity=1 ] (440,140) .. controls (440,134.48) and (444.48,130) .. (450,130) .. controls (455.52,130) and (460,134.48) .. (460,140) .. controls (460,145.52) and (455.52,150) .. (450,150) .. controls (444.48,150) and (440,145.52) .. (440,140) -- cycle ;
\draw  [color={rgb, 255:red, 44; green, 160; blue, 44 }  ,draw opacity=1 ][fill={rgb, 255:red, 255; green, 255; blue, 255 }  ,fill opacity=1 ] (480,140) .. controls (480,134.48) and (484.48,130) .. (490,130) .. controls (495.52,130) and (500,134.48) .. (500,140) .. controls (500,145.52) and (495.52,150) .. (490,150) .. controls (484.48,150) and (480,145.52) .. (480,140) -- cycle ;
\draw  [color={rgb, 255:red, 214; green, 39; blue, 40 }  ,draw opacity=1 ][fill={rgb, 255:red, 255; green, 255; blue, 255 }  ,fill opacity=1 ] (360,140) .. controls (360,134.48) and (364.48,130) .. (370,130) .. controls (375.52,130) and (380,134.48) .. (380,140) .. controls (380,145.52) and (375.52,150) .. (370,150) .. controls (364.48,150) and (360,145.52) .. (360,140) -- cycle ;
\draw  [fill={rgb, 255:red, 255; green, 255; blue, 255 }  ,fill opacity=1 ] (317.5,91.88) .. controls (317.5,89.46) and (319.46,87.5) .. (321.88,87.5) -- (338.13,87.5) .. controls (340.54,87.5) and (342.5,89.46) .. (342.5,91.88) -- (342.5,108.13) .. controls (342.5,110.54) and (340.54,112.5) .. (338.13,112.5) -- (321.88,112.5) .. controls (319.46,112.5) and (317.5,110.54) .. (317.5,108.13) -- cycle ;
\draw  [fill={rgb, 255:red, 255; green, 255; blue, 255 }  ,fill opacity=1 ] (357.5,91.88) .. controls (357.5,89.46) and (359.46,87.5) .. (361.88,87.5) -- (378.13,87.5) .. controls (380.54,87.5) and (382.5,89.46) .. (382.5,91.88) -- (382.5,108.13) .. controls (382.5,110.54) and (380.54,112.5) .. (378.13,112.5) -- (361.88,112.5) .. controls (359.46,112.5) and (357.5,110.54) .. (357.5,108.13) -- cycle ;
\draw  [fill={rgb, 255:red, 255; green, 255; blue, 255 }  ,fill opacity=1 ] (397.6,91.88) .. controls (397.6,89.46) and (399.56,87.5) .. (401.98,87.5) -- (418.23,87.5) .. controls (420.64,87.5) and (422.6,89.46) .. (422.6,91.88) -- (422.6,108.13) .. controls (422.6,110.54) and (420.64,112.5) .. (418.23,112.5) -- (401.98,112.5) .. controls (399.56,112.5) and (397.6,110.54) .. (397.6,108.13) -- cycle ;
\draw  [fill={rgb, 255:red, 255; green, 255; blue, 255 }  ,fill opacity=1 ] (437.5,91.88) .. controls (437.5,89.46) and (439.46,87.5) .. (441.88,87.5) -- (458.13,87.5) .. controls (460.54,87.5) and (462.5,89.46) .. (462.5,91.88) -- (462.5,108.13) .. controls (462.5,110.54) and (460.54,112.5) .. (458.13,112.5) -- (441.88,112.5) .. controls (439.46,112.5) and (437.5,110.54) .. (437.5,108.13) -- cycle ;
\draw  [fill={rgb, 255:red, 255; green, 255; blue, 255 }  ,fill opacity=1 ] (477.5,91.88) .. controls (477.5,89.46) and (479.46,87.5) .. (481.88,87.5) -- (498.13,87.5) .. controls (500.54,87.5) and (502.5,89.46) .. (502.5,91.88) -- (502.5,108.13) .. controls (502.5,110.54) and (500.54,112.5) .. (498.13,112.5) -- (481.88,112.5) .. controls (479.46,112.5) and (477.5,110.54) .. (477.5,108.13) -- cycle ;
\draw  [fill={rgb, 255:red, 255; green, 255; blue, 255 }  ,fill opacity=1 ] (517.5,91.88) .. controls (517.5,89.46) and (519.46,87.5) .. (521.88,87.5) -- (538.13,87.5) .. controls (540.54,87.5) and (542.5,89.46) .. (542.5,91.88) -- (542.5,108.13) .. controls (542.5,110.54) and (540.54,112.5) .. (538.13,112.5) -- (521.88,112.5) .. controls (519.46,112.5) and (517.5,110.54) .. (517.5,108.13) -- cycle ;
\draw  [dash pattern={on 4.5pt off 4.5pt}]  (480,70) -- (517.5,91.88) ;
\draw [color={rgb, 255:red, 214; green, 39; blue, 40 }  ,draw opacity=1 ]   (315,15) -- (330,60) ;
\draw [color={rgb, 255:red, 214; green, 39; blue, 40 }  ,draw opacity=1 ]   (345,15) -- (330,60) ;
\draw [color={rgb, 255:red, 31; green, 119; blue, 180 }  ,draw opacity=1 ]   (385,15) -- (385,60) ;
\draw [color={rgb, 255:red, 31; green, 119; blue, 180 }  ,draw opacity=1 ]   (415,15) -- (385,60) ;
\draw [color={rgb, 255:red, 255; green, 127; blue, 14 }  ,draw opacity=1 ]   (415,60) -- (385,15) ;
\draw [color={rgb, 255:red, 255; green, 127; blue, 14 }  ,draw opacity=1 ]   (415,60) -- (445,15) ;
\draw [color={rgb, 255:red, 44; green, 160; blue, 44 }  ,draw opacity=1 ]   (445,60) -- (415,15) ;
\draw [color={rgb, 255:red, 44; green, 160; blue, 44 }  ,draw opacity=1 ]   (445,60) -- (445,15) ;
\draw [color={rgb, 255:red, 255; green, 127; blue, 14 }  ,draw opacity=1 ]   (500,60) -- (515,15) ;
\draw [color={rgb, 255:red, 255; green, 127; blue, 14 }  ,draw opacity=1 ]   (500,60) -- (485,15) ;
\draw [color={rgb, 255:red, 44; green, 160; blue, 44 }  ,draw opacity=1 ]   (530,60) -- (545,15) ;
\draw [color={rgb, 255:red, 44; green, 160; blue, 44 }  ,draw opacity=1 ]   (530,60) -- (515,15) ;
\draw  [draw opacity=0][fill={rgb, 255:red, 255; green, 255; blue, 255 }  ,fill opacity=1 ] (320,55) .. controls (320,49.48) and (324.48,45) .. (330,45) .. controls (335.52,45) and (340,49.48) .. (340,55) .. controls (340,60.52) and (335.52,65) .. (330,65) .. controls (324.48,65) and (320,60.52) .. (320,55) -- cycle ;
\draw  [draw opacity=0][fill={rgb, 255:red, 255; green, 255; blue, 255 }  ,fill opacity=1 ] (375,55) .. controls (375,49.48) and (379.48,45) .. (385,45) .. controls (390.52,45) and (395,49.48) .. (395,55) .. controls (395,60.52) and (390.52,65) .. (385,65) .. controls (379.48,65) and (375,60.52) .. (375,55) -- cycle ;
\draw  [draw opacity=0][fill={rgb, 255:red, 255; green, 255; blue, 255 }  ,fill opacity=1 ] (405,55) .. controls (405,49.48) and (409.48,45) .. (415,45) .. controls (420.52,45) and (425,49.48) .. (425,55) .. controls (425,60.52) and (420.52,65) .. (415,65) .. controls (409.48,65) and (405,60.52) .. (405,55) -- cycle ;
\draw  [draw opacity=0][fill={rgb, 255:red, 255; green, 255; blue, 255 }  ,fill opacity=1 ] (435,55) .. controls (435,49.48) and (439.48,45) .. (445,45) .. controls (450.52,45) and (455,49.48) .. (455,55) .. controls (455,60.52) and (450.52,65) .. (445,65) .. controls (439.48,65) and (435,60.52) .. (435,55) -- cycle ;
\draw  [draw opacity=0][fill={rgb, 255:red, 255; green, 255; blue, 255 }  ,fill opacity=1 ] (490,55) .. controls (490,49.48) and (494.48,45) .. (500,45) .. controls (505.52,45) and (510,49.48) .. (510,55) .. controls (510,60.52) and (505.52,65) .. (500,65) .. controls (494.48,65) and (490,60.52) .. (490,55) -- cycle ;
\draw  [draw opacity=0][fill={rgb, 255:red, 255; green, 255; blue, 255 }  ,fill opacity=1 ] (520,55) .. controls (520,49.48) and (524.48,45) .. (530,45) .. controls (535.52,45) and (540,49.48) .. (540,55) .. controls (540,60.52) and (535.52,65) .. (530,65) .. controls (524.48,65) and (520,60.52) .. (520,55) -- cycle ;
\draw  [color={rgb, 255:red, 0; green, 0; blue, 0 }  ,draw opacity=1 ][fill={rgb, 255:red, 255; green, 255; blue, 255 }  ,fill opacity=1 ] (305,15) .. controls (305,9.48) and (309.48,5) .. (315,5) .. controls (320.52,5) and (325,9.48) .. (325,15) .. controls (325,20.52) and (320.52,25) .. (315,25) .. controls (309.48,25) and (305,20.52) .. (305,15) -- cycle ;
\draw  [color={rgb, 255:red, 0; green, 0; blue, 0 }  ,draw opacity=1 ][fill={rgb, 255:red, 255; green, 255; blue, 255 }  ,fill opacity=1 ] (335,15) .. controls (335,9.48) and (339.48,5) .. (345,5) .. controls (350.52,5) and (355,9.48) .. (355,15) .. controls (355,20.52) and (350.52,25) .. (345,25) .. controls (339.48,25) and (335,20.52) .. (335,15) -- cycle ;
\draw  [color={rgb, 255:red, 0; green, 0; blue, 0 }  ,draw opacity=1 ][fill={rgb, 255:red, 255; green, 255; blue, 255 }  ,fill opacity=1 ] (375,15) .. controls (375,9.48) and (379.48,5) .. (385,5) .. controls (390.52,5) and (395,9.48) .. (395,15) .. controls (395,20.52) and (390.52,25) .. (385,25) .. controls (379.48,25) and (375,20.52) .. (375,15) -- cycle ;
\draw  [color={rgb, 255:red, 0; green, 0; blue, 0 }  ,draw opacity=1 ][fill={rgb, 255:red, 255; green, 255; blue, 255 }  ,fill opacity=1 ] (405,15) .. controls (405,9.48) and (409.48,5) .. (415,5) .. controls (420.52,5) and (425,9.48) .. (425,15) .. controls (425,20.52) and (420.52,25) .. (415,25) .. controls (409.48,25) and (405,20.52) .. (405,15) -- cycle ;
\draw  [color={rgb, 255:red, 0; green, 0; blue, 0 }  ,draw opacity=1 ][fill={rgb, 255:red, 255; green, 255; blue, 255 }  ,fill opacity=1 ] (435,15) .. controls (435,9.48) and (439.48,5) .. (445,5) .. controls (450.52,5) and (455,9.48) .. (455,15) .. controls (455,20.52) and (450.52,25) .. (445,25) .. controls (439.48,25) and (435,20.52) .. (435,15) -- cycle ;
\draw  [color={rgb, 255:red, 0; green, 0; blue, 0 }  ,draw opacity=1 ][fill={rgb, 255:red, 255; green, 255; blue, 255 }  ,fill opacity=1 ] (475,15) .. controls (475,9.48) and (479.48,5) .. (485,5) .. controls (490.52,5) and (495,9.48) .. (495,15) .. controls (495,20.52) and (490.52,25) .. (485,25) .. controls (479.48,25) and (475,20.52) .. (475,15) -- cycle ;
\draw  [color={rgb, 255:red, 0; green, 0; blue, 0 }  ,draw opacity=1 ][fill={rgb, 255:red, 255; green, 255; blue, 255 }  ,fill opacity=1 ] (505,15) .. controls (505,9.48) and (509.48,5) .. (515,5) .. controls (520.52,5) and (525,9.48) .. (525,15) .. controls (525,20.52) and (520.52,25) .. (515,25) .. controls (509.48,25) and (505,20.52) .. (505,15) -- cycle ;
\draw  [color={rgb, 255:red, 0; green, 0; blue, 0 }  ,draw opacity=1 ][fill={rgb, 255:red, 255; green, 255; blue, 255 }  ,fill opacity=1 ] (535,15) .. controls (535,9.48) and (539.48,5) .. (545,5) .. controls (550.52,5) and (555,9.48) .. (555,15) .. controls (555,20.52) and (550.52,25) .. (545,25) .. controls (539.48,25) and (535,20.52) .. (535,15) -- cycle ;
\draw  [color={rgb, 255:red, 214; green, 39; blue, 40 }  ,draw opacity=1 ][dash pattern={on 5.25pt off 4.5pt}] (357.5,140) .. controls (357.5,133.1) and (363.1,127.5) .. (370,127.5) .. controls (376.9,127.5) and (382.5,133.1) .. (382.5,140) .. controls (382.5,146.9) and (376.9,152.5) .. (370,152.5) .. controls (363.1,152.5) and (357.5,146.9) .. (357.5,140) -- cycle ;
\draw  [color={rgb, 255:red, 44; green, 160; blue, 44 }  ,draw opacity=1 ][dash pattern={on 5.25pt off 4.5pt}] (477.5,140) .. controls (477.5,133.1) and (483.1,127.5) .. (490,127.5) .. controls (496.9,127.5) and (502.5,133.1) .. (502.5,140) .. controls (502.5,146.9) and (496.9,152.5) .. (490,152.5) .. controls (483.1,152.5) and (477.5,146.9) .. (477.5,140) -- cycle ;
\draw  [color={rgb, 255:red, 214; green, 39; blue, 40 }  ,draw opacity=1 ][dash pattern={on 5.25pt off 4.5pt}] (332.5,15) .. controls (332.5,8.1) and (338.1,2.5) .. (345,2.5) .. controls (351.9,2.5) and (357.5,8.1) .. (357.5,15) .. controls (357.5,21.9) and (351.9,27.5) .. (345,27.5) .. controls (338.1,27.5) and (332.5,21.9) .. (332.5,15) -- cycle ;
\draw  [color={rgb, 255:red, 214; green, 39; blue, 40 }  ,draw opacity=1 ][dash pattern={on 5.25pt off 4.5pt}] (302.5,15) .. controls (302.5,8.1) and (308.1,2.5) .. (315,2.5) .. controls (321.9,2.5) and (327.5,8.1) .. (327.5,15) .. controls (327.5,21.9) and (321.9,27.5) .. (315,27.5) .. controls (308.1,27.5) and (302.5,21.9) .. (302.5,15) -- cycle ;
\draw  [color={rgb, 255:red, 44; green, 160; blue, 44 }  ,draw opacity=1 ][dash pattern={on 5.25pt off 4.5pt}] (502.5,15) .. controls (502.5,8.1) and (508.1,2.5) .. (515,2.5) .. controls (521.9,2.5) and (527.5,8.1) .. (527.5,15) .. controls (527.5,21.9) and (521.9,27.5) .. (515,27.5) .. controls (508.1,27.5) and (502.5,21.9) .. (502.5,15) -- cycle ;
\draw  [color={rgb, 255:red, 44; green, 160; blue, 44 }  ,draw opacity=1 ][dash pattern={on 5.25pt off 4.5pt}] (532.5,15) .. controls (532.5,8.1) and (538.1,2.5) .. (545,2.5) .. controls (551.9,2.5) and (557.5,8.1) .. (557.5,15) .. controls (557.5,21.9) and (551.9,27.5) .. (545,27.5) .. controls (538.1,27.5) and (532.5,21.9) .. (532.5,15) -- cycle ;
\draw  [color={rgb, 255:red, 44; green, 160; blue, 44 }  ,draw opacity=1 ][dash pattern={on 5.25pt off 4.5pt}] (432.5,15) .. controls (432.5,8.1) and (438.1,2.5) .. (445,2.5) .. controls (451.9,2.5) and (457.5,8.1) .. (457.5,15) .. controls (457.5,21.9) and (451.9,27.5) .. (445,27.5) .. controls (438.1,27.5) and (432.5,21.9) .. (432.5,15) -- cycle ;
\draw  [color={rgb, 255:red, 44; green, 160; blue, 44 }  ,draw opacity=1 ][dash pattern={on 5.25pt off 4.5pt}] (402.5,15) .. controls (402.5,8.1) and (408.1,2.5) .. (415,2.5) .. controls (421.9,2.5) and (427.5,8.1) .. (427.5,15) .. controls (427.5,21.9) and (421.9,27.5) .. (415,27.5) .. controls (408.1,27.5) and (402.5,21.9) .. (402.5,15) -- cycle ;

\draw (101,13.5) node    {$v_{1}$};
\draw (151,13.5) node    {$v_{2}$};
\draw (201,13.5) node    {$v_{3}$};
\draw (201,63.5) node    {$v_{6}$};
\draw (151,63.5) node    {$v_{5}$};
\draw (101,63.5) node    {$v_{4}$};
\draw (33,13.5) node [anchor=west] [inner sep=0.75pt]  [font=\large] [align=left] {3DM:};
\draw (151,221) node    {$\cdots $};
\draw (101,261) node    {$\cdots $};
\draw (201,261) node    {$\cdots $};
\draw (33,181) node [anchor=west] [inner sep=0.75pt]  [font=\large] [align=left] {VN:};
\draw (101,289.4) node [anchor=north] [inner sep=0.75pt]    {$s$};
\draw (201,289.4) node [anchor=north] [inner sep=0.75pt]    {$s$};
\draw (151,319.4) node [anchor=north] [inner sep=0.75pt]    {$g+q$};
\draw (430.48,221) node    {$\cdots $};
\draw (380.48,261) node    {$\cdots $};
\draw (480.48,261) node    {$\cdots $};
\draw (380.48,289.4) node [anchor=north] [inner sep=0.75pt]    {$s-1$};
\draw (480.48,289.4) node [anchor=north] [inner sep=0.75pt]    {$s-1$};
\draw (430.48,319.4) node [anchor=north] [inner sep=0.75pt]    {$g$};
\draw (430.48,181) node    {$r$};
\draw (379.5,220.5) node    {$f_{1}^{0}$};
\draw (349.5,260.5) node    {$f_{1}^{1}$};
\draw (409.5,260.5) node    {$f_{1}^{5}$};
\draw (479.5,220.5) node    {$f_{6}^{0}$};
\draw (449.5,260.5) node    {$f_{6}^{1}$};
\draw (509.5,260.5) node    {$f_{6}^{5}$};
\draw (331,101) node    {$J_{1}$};
\draw (371,101) node    {$J_{2}$};
\draw (411,101) node    {$J_{3}$};
\draw (451,101) node    {$J_{4}$};
\draw (491,101) node    {$J_{5}$};
\draw (531,101) node    {$J_{6}$};
\draw (371,141) node    {$t_{1}$};
\draw (411,141) node    {$t_{2}$};
\draw (451,141) node    {$t_{3}$};
\draw (491,141) node    {$t_{4}$};
\draw (331,56) node    {$t_{1}$};
\draw (386,56) node    {$t_{2}$};
\draw (416,56) node    {$t_{3}$};
\draw (446,56) node    {$t_{4}$};
\draw (501,56) node    {$t_{3}$};
\draw (531,56) node    {$t_{4}$};
\draw (253,13.5) node [anchor=west] [inner sep=0.75pt]  [font=\large] [align=left] {PN:};

\end{tikzpicture}

    \caption{
        Example of reducing 3DM to cVNE of oversubscribed $2$-star VN. 
        1)~3DM:~$6$ vertices $v_i$ and $4$ colored hyperedges $t_j$ are given.
        The maximum vertex degree in this example is $B = 3$.
        Green and red hyperedges ($t_1$ and $t_4$) is the matching of size $q = 2$.
        2)~PN:~The sets of vertices $J_i$ that substitute $v_i$ are shown as rounded rectangles.
        The number of edges from each $t_i$ to $J_u$ for $u \in t_i$ is $B - 1 = 2$.
        The structure of $J_1$, $J_4$ and $J_6$ is shown in detail.
        The total number of edges from $t_i$ to vertices of $J$ is $s = 6$.
        Choosing a matching $t_1$, $t_4$ occupies $14$ vertices from $T \cup J$ (shown with dashed outline).
        There are $g = 6$ vertices in $T \cup J$ that remain unoccupied.
        For each such vertex there is a filler subtree.
        3)~VN:~ The oversubscribed tree has $g + q$ groups with $s$ leaves in each of them.
        Edges with capacity or weight greater than one are thick.
    }
    \label{fig:3DM_to_Oversub_2Star}
    
\end{figure}
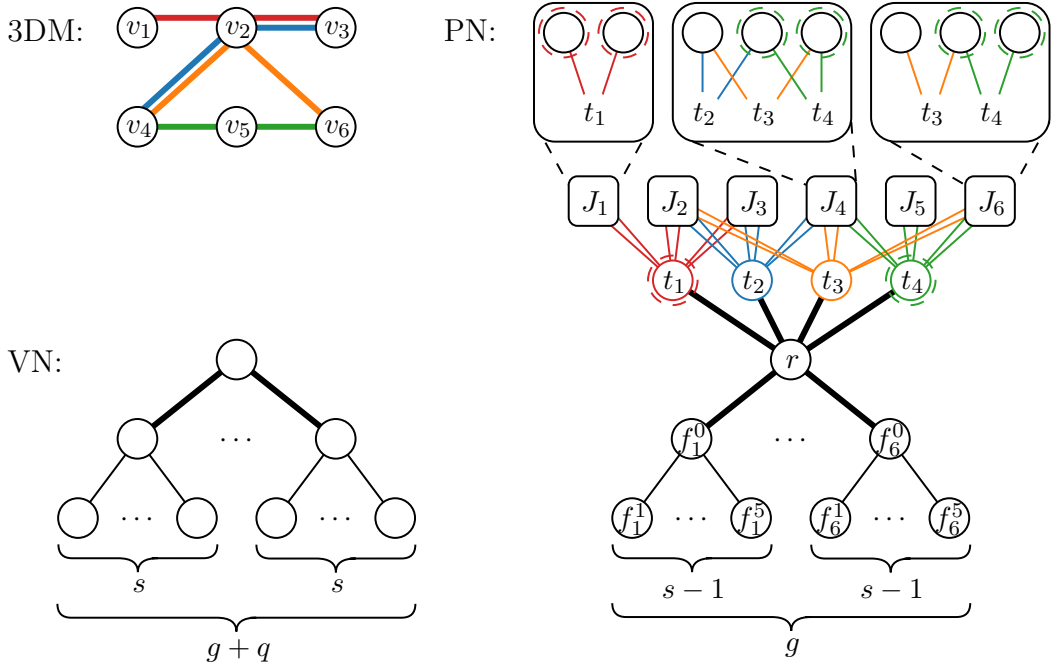

\begin{theorem}\label{th:oversub_2Star_cVNE_NP}
Oversubscribed $2$-star VN on the set of all graph cVNE is NP-complete.
\end{theorem}
\begin{proof}
    We reduce NP-complete problem 3DM to cVNE with oversubscribed $2$-star VN.
    Let $V = X \cup Y \cup Z$, $T \subseteq X \times Y \times Z$ be the instance of 3DM, with $|T| = m$ and $|X| = |Y| = |Z| = q \geq 3$.
    The degree of node $u \in V$ is the number of occurrences in $T$.
    Let $B$ be the maximum degree among all $u \in V$.
    By definition, $B \leq |T| = m$.

    Figure~\ref{fig:3DM_to_Oversub_2Star} shows an example of building PN and VN for a given instance of 3DM.
    First, we describe the general idea of the reduction.
    
    We use subtrees of VN to choose the hyperedges for the matching: 
    embedding of the root of each VN subtree determines which hyperedges are added to the matching and embedding of the leaves ensures that no two chosen hyperedges intersect.
    I.e., we use elements of $T$ as a vertices of PN and if some subtree root of VN is embedded on $t_i \in T$, it is considered as chosen for the matching (in Figure~\ref{fig:3DM_to_Oversub_2Star} these have dotted outline).
    For that we need exactly $q$ subtrees in VN.
    
    To prohibit selecting two intersecting hyperedges we use a set of vertices $J$ shown as rounded rectangle in Figure~\ref{fig:3DM_to_Oversub_2Star}.
    $J$ consist of subsets $J_i$, each substituting for vertex $v_i \in V$.
    Edges between $J_i$ and $t_i$ must be constructed such that if $v_i$ is a common vertex of two hyperedges $a$ and $b$, then both $a$ and $b$ cannot embed subtree roots of VN (i.e., be added to the matching).
    To achieve that, these properties of $J_i$ are of interest: 
    1)~each vertex inside $J$ has max degree $2$, which is required to limit the flow through these vertices; 
    2)~if $v_i$ is a common vertex of hyperedges $a$ and $b$, then there is $j \in J_i$ adjacent to both $a$ and $b$;
    3)~if the subtree root of VN occupies $t \in T$ then all edges between $t$ and $J$ are traversed by subtree leaves.
    
    We add a dedicated node $r$, such that the VN root is always forced to be embedded on $r$.
    
    The last concern is that we want to show that the number of vertices in VN and PN can be the same.
    However, in the current reduction there may be $g$ unoccupied vertices in PN.
    To counter that, we add $g$ filler subtrees in VN.
    For each filler subtree with $s$ leafs in VN there is a filler subtree with $s-1$ leafs in PN.
    Since the number of leaves in PN subtrees is $1$ less, all but one vertex of a filler subtree from VN can be embedded on a filler subtree of PN.
    The remaining leaf is embedded on some previously unoccupied vertex of PN.
    Thus, $g$ filler subtrees fill all unoccupied vertices of PN.
    The precise reduction is described below.
    
    First, we define PN in more detail.
    Let $r$ be a special node (for embedding the root of VN).
    Then, we create one vertex for each hyperedge from $T$.
    Moreover, for each $u \in V$ with the degree $d_u$, we introduce sets of substituting vertices, i.e., they will substitute $u$ in a specific manner described later:
    
        \[P_u = \{p_u^{ab} \, | \, \forall a, b \in T \text{ and } u \in a, u \in b, a \ne b \}\]
        \[C_u = \{c_u^{aj} \, | \, \forall a \in T, j \in \{1, 2, \ldots, B - d_u\} \text{ and } u \in a \}\]
        \[J_u = P_u \cup C_u\]
        \[J = \bigcup_{u=1}^{3 \cdot q} J_{u}\]
        
    Here, $P_u$ is a set of vertices for all pairs of hyperedges that are adjacent to $u$.
    $|P_u| = {\binom{d_u}{2}} = \frac{d_u \cdot (d_u - 1)}{2}$.
    $C_u$ is a set that contains $(B - d_u)$ vertices for every hyperedge that is adjacent to $u$.
    $|C_u| = d_u \cdot (B - d_u)$.
    Note that if the degree $d_u$ is $B$, then, $C_u = \{\}$.
    $J$ is the set of vertices that substitutes $V$ in PN.

    To achieve the same number of vertices in VN and PN, we add $g$ ($g$ is introduced later) filler vertices to PN and VN.
    The filler vertices of PN connected together become filler subtrees.
    $F_i$ is the set of vertices for $i$-th filler subtree $1 \leq i \leq g$ and contains $s = 3 \cdot (B - 1)$ vertices each.
    Set $F$ is the join of all subtrees $F_i$.
    
        \[F_i = \{f_i^0, f_i^1, f_i^2, \ldots, f_i^{s-1}\}\]
        \[F = \bigcup_{i=1}^{g} F_{i}\]
        
    Here $f_i^0$ will be the root of the filler subtree with $f_i^j$ being its child for $1 \leq j \leq (s - 1)$.

    Now, we introduce the edges and their capacities in PN:
    
    \begin{equation}\label{eq:A}
        \begin{cases}
            (r, t) : s & \text{for $t \in T$}.\\
            (t, p_u^{ab}) : 1 & \text{for $t \in T$ and $p_u^{ab} \in J$ with $t = a$ or $b$}.\\
            (t, c_u^{aj}) : 1 & \text{for $t \in T$ and $c_u^{aj} \in J$ with $t = a$}.\\
            (r, f_i^0) : s + 1 & \text{for $f_i^0 \in F$}.\\
            (f_i^0, f_i^j) : 1 & \text{for $f_i^0, f_i^j \in F$ and $1 \leq j \leq (s - 1)$}.
        \end{cases}
    \end{equation}

    Let $O$ be an oversubscribed $2$-star VN with group sizes (i.e., the number of leaves in each subtree) of $s = 3 \cdot (B - 1)$ and oversubscription factor $o = 1$.
    There are two types of groups:
    1)~$q$ effective groups devoted to select a perfect matching, and 2)~$g = m + |J| - q \cdot (s + 1)$ filler groups to match the number of vertices between VN and PN.
    In total, $O$ has $g + q$ second-level subtrees.
    The demand of edges from leaves in $O$ is $1$ and edges from root is $\frac{s}{o} = s$.
    Let us calculate the sizes of PN and VN.
    PN has $1 + m + |J|$ normal nodes and $g \cdot s$ filler nodes.
    VN has $1 + (q + g) \cdot (s + 1)$.
    The difference is $1 + m + |J| - 1 - q \cdot (s + 1) - g$ which becomes zero as we choose $g$.
    
    Now we show the correctness of our reduction from 3DM to cVNE.
    
    $\Rightarrow$: Suppose there exist a solution for 3DM~--- $M \subseteq T$ with $|M| = q$.
    
    Then, we construct an embedding $f$.
    It maps the root of VN onto the special node $r$.
    $f$ maps level-one vertices of $q$ subtrees of $O$ on nodes $t$ which are in $M$.
    For each level-one vertex $v \in O$ that is mapped to $t \in M$, its $s$ children are mapped to the adjacent vertices of $t$ that are from $J$.
    By construction $t$ is adjacent to $3 \cdot (B - 1) = s$ vertices in $J$: if $t = (x, y, z)$, then in each of $J_x$, $J_y$ and $J_z$ there are exactly $(B - 1)$ vertices adjacent to $t$.
    
    After that, $g = |J| + m - q \cdot (s + 1)$ vertices in $J \cup T$ and $g$ filler subtrees are left unmapped.
    We map roots of these $g$ subtrees in $O$ on root vertices $f_i^0$ of the $g$ filler subtrees of PN.
    In each of the groups $s - 1$ leaves are mapped on the leaves of filler subtrees $f_i^j$.
    That leaves $g$ vertices of $O$ unmapped: one from each remaining subtree of $O$.
    
    Now, the remaining capacity of $(r, f_i^0)$ is $1$.
    The remaining capacity of $(r, t)$, where $t \in T$, but $t \notin M$ is $s$.
    The remaining capacity of $(t, u)$ is $1$, where $t \in T \setminus M$ and $u \in J$.
    We map the remaining leaves of $O$ on remaining vertices of $J \cup T$ arbitrarily
    Suppose we have such a remaining leaf $\ell$ and its parent $p$ mapped on $f_i^0$.
    Then, we map an edge $(p, \ell)$ the following path:
    
    \begin{itemize}
        \item If $f(\ell) = v \in T$, then $(p, \ell)$ is mapped on the path
        $f_i^0 \rightarrow r \rightarrow  v$;
        \item If $f(\ell) = c_u^{aj} \in C_u$, then $(p, \ell)$ is mapped on the path
        $f_i^0 \rightarrow  r \rightarrow a \rightarrow c_u^{aj}$;
        \item If $f(\ell) = p_u^{ab} \in P_u$, then $(p, \ell)$ is mapped on the path
        $f_i^0 \rightarrow  r \rightarrow a \rightarrow p_u^{ab}$.
    \end{itemize}
    
    Finally, let us check that all the capacity bounds are satisfied.
    For every edge in PN, this is how many capacity is used by the remaining leaves:
    
    \begin{itemize}
        \item For each $1 \leq i \leq g$, we use capacity one of $f_i^0 \rightarrow r$.
        
        \item For every $t \in T\setminus M$, we use capacity at most $s-2$ of $r \rightarrow t$. 
        $1$ by some leaf mapped on $t$ itself and $(B - 2)$ for each $J_u$, that is adjacent to $t$, and there are exactly three of them.
        Note that since $M$ corresponds to a perfect matching, there exists $b \in M$ such that $p_u^{tb}$ is already occupied.
        Meaning, at most $(B - 2)$ vertices in $J_u$ are adjacent to $t$ and are occupied by remaining leaves.
        Giving us $1 + 3 \cdot (B - 2) = 3 \cdot (B - 1) - 2 = s - 2$.
        
        \item For every $t \in T \setminus M$ and $u \in J$, we use at most one capacity of $t \rightarrow u$.
    \end{itemize}
    
    The presented mapping $f$ is a bijection between VN and PN and does not violate the capacity constraints, meaning it is the correct solution for cVNE.
    
    $\Leftarrow$: Suppose there exists $\phi$~--- a correct mapping from VN to PN. We show that we can solve 3DM.
    
    \begin{itemize}
        \item The outgoing capacity of $r$ is $g \cdot (s + 1) + m \cdot s$: $g$ edges with capacity $(s + 1)$ to $F$ and $m$ edges with capacity $s$ to $T$.
        \item The outgoing capacity of $c_u^{aj} \in J$ is $1$:
            a single edge with capacity one to $a$.
        \item The outgoing capacity of $p_u^{ab} \in J$ is $2$: 
            two edges with capacity one to $a$ and $b$.
        \item The outgoing capacity of $t \in T$ is $2 \cdot s = s + 3 \cdot (B - 1)$: 
            one edge with capacity $s$ to $r$ 
            and $(B + 1)$ edges with capacity one to every $J_u$, that is adjacent to $t$, and there are exactly three of them.
        \item The outgoing capacity of $f_i^j$ for $1 \leq j \leq (s - 1)$ is $1$:
            a single edge with capacity one to $f_i^0$.
        \item The outgoing capacity of $f_i^0$ is $2 \cdot s$:
            one edge with capacity $s + 1$ to $r$ 
            and $s - 1$ edges with capacity one to $f_i^j$ for $1 \leq j \leq (s - 1)$.
    \end{itemize}

    The root of $O$ has $g + q$ subtrees, each with $s$ leaves.
    The weight of the edge between the root and any subtree is $s$.
    Thus, the outgoing demand of the root of $O$ is $s \cdot (g + q)$.
    $r$ is the only vertex in PN that has the outgoing capacity at least $s \cdot (g + q)$.
    None of the other vertices in PN satisfy us since $q \geq 3$.
    Hence, $\phi$ maps the root of $O$ on $r$.
    
    We write that the vertex $u$ in PN is \emph{selected} if $\phi$ maps a level one vertex of $O$ to $u$.
    
    Only the vertices $t \in T$ or $f_i^0 \in F$ can be selected since these are the only vertices (apart from $r$) with outgoing capacity at least $2 \cdot s$.
    
    If $u \in V$ is adjacent to hyperedges $a, b \in T$, then both $a$ and $b$ can not be selected.
    That follows from the two observations:
    if $a$ is selected, all the outgoing capacity of $a$ is maxed and one leaf travels the path starting with $a \rightarrow p_u^{ab}$; 
    if both $a$ and $b$ are selected, there should be a vertex $p_u^{ab} \in J$ that has to simultaneously contain two leaves of $O$, which is impossible.
    
    Since $|X| = q$, at most $q$ vertices of $T$ are selected.
    Otherwise, if at least $q + 1$ vertices of $T$ are selected, there must exist $u \in X$ that is adjacent to two selected hyperedges, which contradicts the previous statement.
    Since there are $g$ more vertices to be selected and at most $g$ vertices from $F$ can be selected, there are exactly $g$ vertices selected from $F$ and exactly $q$ from $T$.
        
    Thus, $q$ selected vertices of $T$ correspond to the correct complete matching $M$: two hyperedges of $M$ cannot be adjacent to the same vertex from $V$ due to the construction.
\end{proof}

Note that in the proof above the constructed PN is bipartite.
Thus, we can strengthen Theorem~\ref{th:oversub_2Star_cVNE_NP}.

\section{VNE with weights and capacities}
\label{sec:wcVNE}

Sometimes we are interested in VNE with both cost and capacity restrictions (wcVNE)~\cite{FischerBBMH13_survey, RostFS15_stars}.
wVNE and cVNE problems are its restrictions.
Thus, if we show a variant of wcVNE to be P, then wVNE and cVNE is also P.
The opposite is also true: if cVNE or wVNE variant is NP-complete, it remains NP-complete in wcVNE.

Obviously, the generic wcVNE lies in NP, since we can check the solution in polynomial time.
Thus, we can refer to NP-hard variants of wcVNE as NP-complete.
However, in this section we only discuss polynomial variants since all NP-complete variants follow from NP-completeness of this variant in either wVNE or cVNE form.
Note that some presented algorithms might not be optimal~--- our main goal is to show their polynomial complexity.

\subsection{wcVNE of Uniform Line VN on Tree PN}\label{sec:line_on_tree_DP}

Now, we again consider the uniform line as VN.
The problem is NP-complete for an arbitrary PN as shown in Theorem~\ref{th:ULEP_wVNE_NPC}~and by Wu et al.~\cite{Wu20_path}.
However, when PN is a tree, wcVNE can be solved in polynomial time.
Thus, more restricted problems, wVNE and cVNE, are also in P.

We consider the embedding of a line on a tree as the path that traverses the following nodes: $f(1) \rightarrow f(2) \rightarrow \ldots \rightarrow f(n)$. The total cost is the cost of this path.

\begin{lemma}\label{lem:uni_line_2}
    If there exists an optimal embedding $f$ of a uniform line on a tree $T$, there exists an optimal embedding $g$, such that in a single direction every edge is traversed at most once.
\end{lemma}
\begin{proof}
    Consider some optimal embedding $f$ of a line on tree $T$ with edge $u \rightarrow v$ traversed at least twice.
    In other words, $f$ maps VN on a path $p = \ldots \rightarrow u \rightarrow v \rightarrow \sigma \rightarrow u \rightarrow v \rightarrow \ldots$, where $\sigma$ is a subpath of $p$.
    Since $T$ is undirected, there exist path $p' = \ldots \rightarrow u \rightarrow \overline{\sigma} \rightarrow v \rightarrow \ldots$, where $\overline{\sigma}$ is the reverse of $\sigma$.
    $p'$ visits the same vertices $p$ visits, but $p'$ has a smaller cost.
    Since $p$ is finite, we can repeat the path reduction until it has no repeated edges.
\end{proof}

\begin{theorem}
wcVNE of uniform linear VN on tree PN is in P.
\end{theorem}
\begin{proof}
An optimal embedding of a line may not exist due to capacity restrictions.
However, when it exists, there is a direct consequence to Lemma~\ref{lem:uni_line_2}.
There exists an optimal embedding $g$ with every edge traversed at most twice (once in every direction). 

Assume an optimal embedding $g$ from Lemma~\ref{lem:uni_line_2} that starts traversing tree PN $T$ at vertex $s$ and ends at vertex $t$.
Then, by Lemma~\ref{lem:uni_line_2} every edge on a simple path between $s$ and $t$ is traversed once and all other edges are traversed twice.
Thus, the wcVNE problem of embedding a line on a tree $T$ is equivalent to finding a simple path with the maximal length in $T$ that includes all edges with capacity of $1$.
This can be solved in $O(n)$ time.
\end{proof}

The problem can be solved using dynamic programming with the polynomial complexity even when the sizes of VN and PN are different.

\subsection{wcVNE of Oversubscribed $2$-star on Tree}\label{sec:oversub_on_tree}

In Section~\ref{sec:oversub_2satr}, we showed that an embedding of an oversubscribed $2$-star VN on an arbitrary graph PN is NP-complete.
However, that does not necessarily mean that its embedding is NP-complete for all PN topologies.
For example, we can embed an oversubscribed $2$-star VN on tree PN in polynomial time.
Interestingly, when we remove the oversubscribed condition from $2$-star VN, the problem of embedding it on a tree becomes NP-complete for both wVNE and cVNE as shown earlier in Sections~\ref{sec:2star_wVNE}~and~\ref{sec:cVNE}.

\begin{theorem}
wcVNE of oversubscribed $2$-star VN on tree PN is in P.
\end{theorem}
\begin{proof}
We can solve this problem using the dynamic programming approach.
Let $G$ be a given tree PN and VN $O$~---~an~oversubscribed $2$-star with $s$ leaves in each subtree.
The oversubscriptions factor is $o$.
At first, we try to map the root of $O$ to all possible vertices.
Suppose that it is mapped on a vertex $R$ in $G$ that is becoming a root.

Now, suppose the following:
\begin{itemize}
    \item $r \ne R$ is the root of some subtree $T_r$ in $G$;
    \item $c_r$ and $w_r$ are the capacity and the cost of an edge between $r$ and its parent;
    \item $x$ subtree roots of $O$ are embedded on vertices of $T_r$;
    \item $\ell_{out}$ leaves are leaving $T_r$, i.e., $\ell_{out}$ leaves of $O$ are embedded outside of $T_r$ while their parents are embedded in $T_r$;
    \item $\ell_{in}$ leaves embedded inside of $T_r$ while their parents are embedded outside of $T_r$.
\end{itemize}

We denote the dynamic programming $d(r, x, \ell_{out}, \ell_{in})$ as the minimal possible cost of such a configuration on $T_r$, including the cost of an edge from $r$ to its parent.
If this configuration is not achievable (for example, due to capacity constraints), we assume $d(\ldots) = \infty$.

For a given $r \ne R$ and $d(r, x, \ell_{out}, \ell_{in})$, the bandwidth required for the edge from $r$ to its parent is $b = x \cdot \frac{s}{o} + \ell_{out} + \ell_{in}$.
If $b$ is bigger than the capacity $c_r$ of the edge, then $d(r, x, \ell_{out}, \ell_{in}) = \infty$.
Otherwise, we compute $d(r, x, \ell_{out}, \ell_{in})$.

Suppose the children of $r$ are denoted as $c_1, c_2, \ldots, c_k$.
We denote $d_i(r, x, \ell_{out}, \ell_{in})$ as the minimal possible cost of $d(r, x, \ell_{out}, \ell_{in})$, such that only children $c_1, c_2, \ldots, c_i$ are used.
$d(r, x, \ell_{out}, \ell_{in}) = b \cdot w_r + d_k(r, x, \ell_{out}, \ell_{in})$:
the minimal embedding cost of all children subtrees plus the cost of an edge from $r$ to its parent.
$d_0$ means no children are used.
\begin{align*}
d_0(r, 0, 0, 1) = 0 && \text{$r$ embeds the leaf of $O$} \\
d_0(r, 1, s, 0) = 0 && \text{$r$ embeds the subtree root of $O$} \\
d_0(r,  x, \ell_{out}, \ell_{in}) = \infty && \text{otherwise}
\end{align*}

These equations show that $r$ either 1)~``consumes'' a leaf of $O$ or 2)~``consumes'' a subtree root of $O$ and ``generates'' $s$ new leaves.

When adding a new $c_i$ into consideration we have to choose how many vertices the subtree of $c_i$ ``consumes'': we need to choose $x^1, x^2, \ell_{out}^1, \ell_{out}^2, \ell_{in}^1, \ell_{in}^2$, with $x^2, \ell_{out}^2, \ell_{in}^2$ being the parameters for the subtree of $c_i$.
\begin{align*}
    x & = x^1 + x^2 \\
    \ell_{out} + \ell_{in}^1 + \ell_{in}^2 & = \ell_{in} + \ell_{out}^1 + \ell_{out}^2 \\
    d_i(r,  x, \ell_{out}, \ell_{in}) & = min (d_{i - 1}(r,  x^1, \ell_{out}^1, \ell_{in}^1) \\
    &\phantom{= min(} + d(c_i, x^2, \ell_{out}^2, \ell_{in}^2))
\end{align*}

By that, we choose how many leaves and subtree roots of $O$ go to the $c_i$ and the remaining must be embedded on the rest of $T_r$.

For the root $R$ the formulas are similar.
$d_0$ is only defined at $d_0(R, 0, 0, 0) = 0$ since $R$ already embeds the root of $O$ and $w_R = 0$, since $R$ has no parent.
$d(R, g, 0, 0)$ is the minimal possible cost of embedding $O$ on $G$ (with $g$ being the number of subtrees in $O$).

It is easy to verify that this algorithm is polynomial: 
choosing of $R$ is linear;
four dimensions of $d(r, x, \ell_{out}, \ell_{in})$ are restricted by $n$;
each $d(r, x, \ell_{out}, \ell_{in})$ is calculated in polynomial time with another dynamic programming by $d_i(r, x, \ell_{out}, \ell_{in})$;
each $d_i(r, x, \ell_{out}, \ell_{in})$ is calculated in polynomial time since all the parameters are restricted by $n$.

The correctness of the algorithm follows by the construction.
\end{proof}

\subsection{wcVNE of Arbitrary Graph on Star}\label{sec:star_PN_Poly}

In Section~\ref{sec:star_wVNE} we discussed that wVNE with star VN is polynomial.
Similarly, when PN is a star, the VNE can be reduced to solving an \emph{assignment problem}~\cite{Kuhn55_hungarian}.
However, that only holds with the assumption of one-to-one vertex embedding.
If we allow multiple VN nodes to be embedded on a single PN node, even cVNE on star PN is NP-hard~\cite{LiZWGZ15}.

\begin{theorem}
wcVNE of arbitrary VN on star PN is in P.
\end{theorem}
\begin{proof}
Consider an instance of VNE with VN $G$ and a star PN $S$, with $|S| = |G| = n$.
In PN $S$, we denote the weight of edge $e$ from node $j \leq n$ to the central node $n$ as $d_j$ and the capacity of $e$ is $c_j$, when $d_n = 0$ and $c_n = \infty$.
In VN $G$, the demand of an edge $uv$ is $w_{uv}$.
Additionally, we denote the sum of the costs of adjacent edges of $u$ as $w_u$.

Since $S$ is a star, for any embedding $f$ its cost is:
$$Cost(f) = \sum_{uv \in E(G)} w_{uv} \cdot (d_{f(u)} + d_{f(v)}) = \sum_{u = 1}^{n} d_{f(u)} \cdot w_{u}$$

Hence, the cost of an embedding of vertex $u \in V(G)$ on vertex $x \in V(S)$ is $d_{x} \cdot w_{u}$.
In the same way, we can show that $u$ can be embedded on $x$ if and only if $w_{u} \leq c_x$.
The way we define $d_n$ and $c_n$ accounts for the 
uniqueness of the root in PN.

Notice that we reduced VNE to the problem of matching vertices of VN to vertices of PN with an independent cost of mapping vertex $u$ on vertex $x$.
That is exactly the assignment problem that has a polynomial solution.
For example, it can be solved with the Hungarian algorithm~\cite{Kuhn55_hungarian}.
\end{proof}

\section{Related Work}
\label{sec:related}
%
%
Over the past few decades, several variations of the VNE problem have been studied.
Fischer et al.~\cite{FischerBBMH13_survey} provided a classification of VNE instances based on the algorithmic approach.
%

\paragraph{Problem statements}
VNE problem is commonly used to reserve bandwidth and process requests in a network.
Hence, the main parameters of the problem are capacities defined on the physical network~\cite{FischerBBMH13_survey}, i.e., each link can handle only some fixed number of connections.
Most of the studies try to satisfy the requirements posed by capacities and minimize some cost of an embedding at the same time~\cite{RostFS15_stars, HouidiLBZ11}.
There exist several definitions of the cost~\cite{FischerBBMH13_survey}, for example: the total length of requests~\cite{RostFS15_stars, Fuerst17_replication, DiazPS02_survey}, the revenue~\cite{Chowdhury11_vineyard}, an acceptance ratio~\cite{HouidiLBZ11}, an energy usage~\cite{Botero12_energy}, etc.
%

Regardless of the variation, both optimizing the cost and satisfying the capacity requirements are known to be NP-hard~\cite{FischerBBMH13_survey} on generic graphs.
Rost and Schmid~\cite{RostS20_complexity} proved the NP-completeness of VNE not only for enforcing the capacity of a physical network, but also under placement, routing, or latency restrictions.
Introducing server capacities restrictions (i.e. allowing servers to embed multiple virtual actors) is hopeless for non-uniform demand.
For that formulation, Figiel et al.~\cite{FigielKNR0Z21_Tree} showed that VNE is NP-hard even for arbitrary VN and PN of constant size.

\paragraph{Approaches to solve general VNE}
As discussed above, the generic VNE is NP-hard.
For small networks, efficient exact algorithms can be designed by using Integer linear Programming~\cite{HouidiLBZ11, Botero12_energy}.
Heuristic or meta-heuristic approaches are more appropriate for larger instances where not the exact answer is necessary.
For example, the VNE can be reduced to the Subgraph Isomorphism Detection problem as was shown by Lischka and Karl~\cite{Lischka09_heuristic}.
Several other works provide heuristic solutions, including the Max-Min-Ant Colony meta-heuristic used by Fajjari et al.~\cite{Fajjari11_vneAnts}. 
Nevertheless, there remain few applications where the optimal solution is possible to compute on large instances, for example, for stars~\cite{RostFS15_stars}.

\paragraph{Topologies restrictions}
A popular way to simplify the VNE problem is to consider restricted topologies of physical and virtual networks.

Most PNs have a prescribed topology: e.g., FatTree~\cite{Leiserson85_fat}, BCube~\cite{Guo09_bcube}. 
At first, we overview a line of work that assumes restricted topologies of PN.
VNE with an arbitrary VN remains NP-hard even for linear PN.
This result follows from the Minimal Linear Arrangement~\cite{DiazPS02_survey} problem that embeds VN on a uniform linear network.
Nevertheless, there are works considering different topologies.
For example, an exact exponential DP algorithm for PN with tree topology was proposed by Figiel et al.~\cite{FigielKNR0Z21_Tree}.
Some authors also assume a topology of the PN when conducting their study: e.g., tree~\cite{Fuerst17_replication}, cycle~\cite{Wu20_path}.

Other papers restrict the topology of the VN.
For example, some VN abstractions are introduced by Ballani et al.~\cite{BallaniCKR11_virtualNet}: a ``virtual cluster'' (also called star~\cite{RostFS15_stars}) and its oversubscribed version.
Rost et al.~\cite{RostFS15_stars} provide a polynomial algorithm that embeds a star VN on an arbitrary graph.
Using a similar algorithm, they solve the hose embedding problem on physical networks without capacity restrictions.
The VNE for planar~\cite{RostS20_complexity}, linear~\cite{Wu20_path}, and cycle~\cite{Wu20_path} VNs are shown to be NP-hard.
However, as provided by Wu et al.~\cite{Wu20_path}, there are polynomial algorithms for uniform path-on-path and uniform cycle-on-cycle VNE.
Also, there exist several approximation algorithms for chain~\cite{Even16_chainApproximation} and cactus~\cite{Rost19_cactusVirtual} virtual networks.

\section{Conclusion}
\label{sec:conclusion}

In this paper, we provided an extensive set of results on the complexity of VNE with the constraints of popular tree-like topologies on VN and PN.
Similar to the general VNE, many variations remain NP-hard.
One of our most important results is that VNE of oversubscribed $2$-star VN on a generic PN is NP-complete, while it is polynomial for tree PN.

Our work motivates the further study of the polynomial variants of VNE.
In particular, we would be interested in the generalization of our version:
not assuming the same size of VN and PN.
Additionally, it would be interesting to optimize the presented polynomial algorithms.

Finally, it is interesting to look into other topologies, for example, a FatTree which are widely deployed in datacenters.
 
\bibliographystyle{abbrv}
\bibliography{references}
 
\end{document}